\documentclass[onecolumn,journal]{IEEEtran}

\usepackage{booktabs}
\usepackage{tabularx}
\usepackage{amsmath,amssymb}
\usepackage{makecell}

\usepackage{cite}
\usepackage{amsmath,amssymb,amsfonts}
\usepackage{graphicx}
\usepackage{textcomp}
\usepackage{multirow}
\usepackage{amsthm}
\usepackage{caption}
\usepackage[numbers,sort&compress]{natbib}
\usepackage{url}
\usepackage{rotating}
\usepackage{natbib,lineno,color,subfigure}
\usepackage{comment}
\usepackage{ragged2e}
\usepackage{bm}
\usepackage{algorithm}
\usepackage{algorithmic}
\newtheorem{definition}{Definition}
\newtheorem{theorem}{Theorem}
\newtheorem{lemma}{Lemma}
\newtheorem{remark}{Remark}

\newcommand{\Tian}[1]{#1}

\begin{document}

\title{Algorithms for the Shortest Vector Problem in $2$-dimensional Lattices, Revisited}

\author{Lihao Zhao,  Chengliang Tian, Jingguo Bi, Guangwu Xu, Jia Yu

\IEEEcompsocitemizethanks{\IEEEcompsocthanksitem This research is supported by Natural Science Foundation of Shandong Province (No. ZR2022MF250), National Natural Science Foundation of China (No. 61702294,12271306), the State Key Laboratory of Cryptography and Digital Economy Security, Shandong University (No. KFYB2505),  National Key Research and Development Program of China (No. 2024YFB4504700, 2022YFB2702804),  and the Youth Science and Technology Innovation Talent Support Program Project of Beijing University of Posts and Telecommunications (No. 2023ZCJH10). (\emph{Corresponding author: C. Tian})

\IEEEcompsocthanksitem
L. Zhao, C. Tian and J. Yu are with the College
of Computer Science and Technology,  Qingdao University,  Qingdao 266071,  China, 
and C. Tian is also with State Key Laboratory of Cryptography and Digital Economy Security,
Shandong University, Qingdao, China.
E-mail: zlh\_qdu@163.com; 
tianchengliang@qdu.edu.cn; qduyujia@163.com.

\IEEEcompsocthanksitem J. Bi is with the School of Cyberspace Security, Beijing University of Posts and Telecommunications, Beijing 100876,  China.
 E-mail:jguobi@bupt.edu.cn.

\IEEEcompsocthanksitem G. Xu is with School of Cyber Science and Technology, Shandong University (Qingdao), Qingdao 266273,  China. E-mail: gxu4sdq@sdu.edu.cn.
}

\thanks{Manuscript received October 10, 2019; revised ****, 20**.}}

\markboth{Journal of \LaTeX\ Class Files,~Vol.~14, No.~8, August~2015}%
{Shell \MakeLowercase{\textit{et al.}}: Bare Demo of IEEEtran.cls for IEEE Communications Society Journals}

\maketitle

\begin{abstract}
Efficiently solving the Shortest Vector/Basis Problem (SVP/SBP) in two-dimensional lattices holds practical significance in cryptography and computational geometry. While simpler than its high-dimensional counterpart, two-dimensional SVP motivates scalable solutions for high-dimensional lattices and benefits applications like sequence cipher cryptanalysis involving large integers. In this work, we first propose a novel definition of reduced bases and develop an efficient adaptive lattice reduction algorithm \textbf{CrossEuc} that strategically applies the Euclidean algorithm across dimensions. Building on this framework, we introduce \textbf{HVec}, 
a vectorized generalization of the integer Half-GCD (HGCD) algorithm that efficiently halves vector bit-lengths, which may be of independent interest. By iteratively invoking \textbf{HVec}, our optimized algorithm \textbf{HVecSBP} achieves a reduced basis in $O(\log n M(n) )$ time for arbitrary input bases with bit-length $n$, where \(M(n)\) denotes the cost of multiplying two \(n\)-bit integers. \Tian{Compared with prior work, the shortness of our newly defined reduced basis holds for all $\ell_p (1\le p\le \infty)$ norms, and our design saves at least four multiplications per iteration without increasing the iteration count. }
Comprehensive experimental results demonstrate the following for the $\ell_2$-SBP: (1) \(\mathbf{CrossEuc}\) achieves speedups of approximately $500\times$, $350\times$ and \Tian{$10$-$15\times$} over Yap's fundamental reduction algorithm \textbf{CRS}, the Lagrange reduction algorithm \textbf{LagRed}, \Tian{and Pornin's improved Lagrange reduction algorithm \textbf{LagRed+}, respectively}. (2) Our optimized variant, \(\mathbf{HVecSBP}\),  delivers an additional 14$\times$ and \Tian{$1.8$-$3\times$} acceleration over Yap's optimized $\mathbf{HalfGaussianSBP}$ \Tian{and the implementation of Sch\"onhage's algorithm \textbf{Flatter}, respectively.}
 For the $\ell_\infty$-SBP: (1) For HNF-form bases, \(\mathbf{HVecSBP}\) achieves $13\times$ acceleration over prior designs in worst-case scenarios;  (2) For general-form bases, $\mathbf{CrossEuc}$ attains $\geq$4.5$\times$ efficiency gains over non-HGCD methods. Crucially, $\mathbf{HVecSBP}$  exhibits progressively intensifying advantages against HGCD-optimized designs as lattice basis vectors linear dependence weakens.

\end{abstract}
 

\IEEEpeerreviewmaketitle

\section{Introduction}

A lattice is a classic object of study in the geometry of numbers, originating from research on sphere packing and covering problems in the $17$th century \Tian{\cite{conway2013sphere}}. Around 1831, Gauss introduced the concept of a lattice and determined the maximum density of sphere packing in three-dimensional space.  In the past thirty years, lattice theory has demonstrated its powerful applications in coding theory and cryptography, particularly through lattice-based cryptography \Tian{\cite{NS01}}. The application of lattice theory in these domains underscores its importance in tackling contemporary challenges in data security and communication. Its role in developing cryptographic methods that are resistant to both classical and quantum attacks highlights its value in securing digital information.

In general, a lattice is a set of discrete points in $n$-dimensional real space that has a periodic structure. Specifically,  for $d$ linearly independent vectors $\boldsymbol{b}_1,\cdots, \boldsymbol{b}_d$ in the $m$-dimensional real space $\mathbb{R}^{m}$, a lattice $\mathcal{L}$ 
\begin{equation*}
\mathcal{L}=L(\mathbf{B})=\left\{\sum_{i=1}^dz_i\boldsymbol{b}_i:z_i\in\mathbb{Z}\right\}.
\end{equation*}
 is the set of all integer linear combinations of these vectors, where $d$ is called the dimension of the lattice and the set of $d$ vectors $\mathbf{B}=[\boldsymbol{b}_1,\cdots, \boldsymbol{b}_d]$ is called a basis of the lattice. The  $k$-th successive minimum of a lattice is defined as 
 \begin{align*}
   \lambda_k(\mathcal{L})=\mathrm{inf}\{r>0: \mathrm{dim}(\mathrm{span}(\mathcal{B}(\boldsymbol{o},r)\cap \mathcal{L}))\ge k\}, 
 \end{align*}
 which $\mathcal{B}(\boldsymbol{o},r)$ represents the open ball of radius $r$ centered at the origin. In particular, $\lambda_1(L)$ denotes the length of the shortest non-zero vector in the lattice.


One of the most fundamental computational problems in lattice theory, known as the Shortest Vector Problem (SVP), is to find the shortest vector in a lattice. For a given lattice \(\mathcal{L}\), this involves identifying a non-zero lattice vector \(\boldsymbol{v}\) such that \(||\boldsymbol{v}|| \leq ||\boldsymbol{u}||\) holds for all non-zero vectors \(\boldsymbol{u} \in \mathcal{L}\).   The difficulty of solving this problem grows significantly with the dimension of the lattice. However, in the simplest case of two-dimensional lattices, solving the SVP—or even the Shortest Basis Problem (SBP), which seeks a basis $\mathbf{B}=[\boldsymbol{a}\ \boldsymbol{b}]$ satisfying $\|\boldsymbol{a}\|=\lambda_1(\mathcal{L})$ and $\|\boldsymbol{b}\|=\lambda_2(\mathcal{L})$-is relatively straightforward, making it an ideal starting point for studying lattice problems. Developing fast algorithms for solving the two-dimensional SVP/SBP is of both theoretical and practical importance.  From an algorithm design perspective, solving the SVP in two-dimensional lattices provides foundational insights for addressing higher-dimensional cases. For example, Lagrange’s pioneering work \cite{Lag73} on two-dimensional lattice reduction laid the groundwork for the celebrated LLL algorithm by A. K. Lenstra, H. W. Lenstra, and L. Lov\'asz \cite{LLL}. The LLL algorithm extends Lagrange’s approach by introducing techniques for reducing lattice bases in higher dimensions, combining ideas of approximate orthogonality and length minimization. 
\Tian{Furthermore, further analysis of Lagrange's reduction algorithm and subsequent improvements \cite{NS09,Sch91} have directly contributed to the optimization design and improved complexity analysis of the LLL algorithm \cite{NSV11,NSsiam09,RH23}.} 
As such, understanding and refining two-dimensional lattice reduction algorithms can directly contribute to advancements in higher-dimensional lattice computations.  On the practical side, the SVP in two-dimensional lattices has direct applications in cryptanalysis \cite{CTJ22} \Tian{and the efficient implementation of cryptographic algorithms (e.g., accelerating EdDSA signature verification \cite{Por20,MPM25}). }
One critical indicator of the security of cryptographic random sequences is their \(2\)-adic complexity, which measures the sequence’s resistance to linear feedback shift register attacks. Computing the \(2\)-adic complexity involves determining the Minimal Rational Fractional Representation (MRFR) of the sequence, a problem that can be directly transformed into finding the shortest basis of a two-dimensional lattice in the $\ell_\infty$ norm. Specifically, given a sequence, a specialized two-dimensional lattice is constructed in Hermite Normal Form (HNF), and the shortest vector in this lattice corresponds to the MRFR.  Efforts to optimize algorithms for solving the MRFR problem have shown significant progress. Arnault et al. \cite{TIT04} proposed a Euclidean-based algorithm in 2004, with subsequent improvements in 2008 \cite{TIT08}, to compute the shortest basis of the lattice. However, these methods occasionally produced incorrect results in certain cases. Recently, Che et al. \cite{CTJ22} proposed a novel and faster algorithm targeting the two-dimensional lattice shortest vector problem, significantly enhancing both the computational efficiency and reliability of solutions.  Further optimization of algorithms for finding the shortest basis in two-dimensional lattices would directly accelerate the computation of \(2\)-adic complexity, enabling faster and more accurate assessments of cryptographic random sequence security. 

In summary, the above discussion highlights the two-fold significance of studying two-dimensional lattice problems: as a foundation for high-dimensional lattice theory and as a means to improve practical cryptographic applications. Therefore, continued research in this area is vital for advancing both theoretical lattice studies and practical algorithmic solutions.

\subsection{Related works}
  
 For the SVP in two-dimensional lattices, the well-known solution is the Lagrange reduction algorithm (often called Gaussian reduction) \cite{Lag73, MGbook}, which resembles the integer Euclidean algorithm.  
Given any lattice basis vectors $\boldsymbol{a}=[a_1, a_2]^T$ and $\boldsymbol{b}=[b_1, b_2]^T$ under a computable norm $\|\cdot\|$, we assume $\|\boldsymbol{a}\| \geq \|\boldsymbol{b}\|$ without loss of generality. The core reduction step employs a greedy strategy: it seeks an integer  
$$
q = \underset{\mu}{\mathrm{argmin}} \|\boldsymbol{a} - \mu\boldsymbol{b}\|
$$  
to minimize the length of $\boldsymbol{c} = \boldsymbol{a} - q\boldsymbol{b}$, updating the basis to $[\boldsymbol{b}, \boldsymbol{c}]$. Iterating this procedure yields a shortest basis.  For the frequently used Euclidean metric ($\ell_2$ norm), Lagrange reduction uses: 
$$
q = \left\lceil \frac{\langle \boldsymbol{a},\boldsymbol{b} \rangle}{\langle \boldsymbol{b},\boldsymbol{b} \rangle} \right\rfloor.
$$  
\Tian{Lagarias \cite{Laga80} proved that the complexity of the above algorithm is $O(M(n)n)$, where $n = \max\{\lceil \log \|\boldsymbol{a}\| \rceil, \lceil \log \|\boldsymbol{b}\| \rceil\}$ and $M(n)$ is the complexity of multiplying two $n$-bit integers. Subsequently, Vallée \cite{Val91} derived an exact upper bound on the number of iterations, and the current best complexity analysis, achieving $O(n^2)$, is due to Nguyen and Stehl\'e \cite{NS09}. Recently, Pornin \cite{Por20} further improved Lagrange's reduction algorithm. The core idea is to set $$q = 2^s$$ in the basic reduction step, where $s \approx \operatorname{bitlength}(\langle\boldsymbol{a},\boldsymbol{b}\rangle) - \operatorname{bitlength}(\langle\boldsymbol{b},\boldsymbol{b}\rangle)$. This avoids the division operations required in the original computation of $q$, while replacing the subsequent multiplications involving $q$ with efficient shift operations. Although this increases the number of iterations, experimental results show that it still leads to significant practical efficiency gains.  Clearly, Lagrange's} $q$-selection differs from the standard integer Euclidean algorithm. \Tian{To further accelerate Lagrange reduction using a recursive divide-and-conquer strategy, } Yap \cite{yap1992fast} studied a directly analogous approach, updating the basis to $[\boldsymbol{b}, \boldsymbol{c}]$ with $\boldsymbol{c} = \boldsymbol{a} - q\boldsymbol{b}$ via:  
$$
q = \left\lfloor \frac{\langle \boldsymbol{a},\boldsymbol{b} \rangle}{\langle \boldsymbol{b},\boldsymbol{b} \rangle} \right\rfloor.
$$  
The algorithm terminates if the basis remains unchanged after an update. This method, termed $\mathbf{CRS}$, was proven to have slightly higher complexity than Lagrange's reduction, but with negligible efficiency loss ("not giving up too much" as said in \cite{yap1992fast}). Crucially, it preserves the geometric property that basis vectors always form acute angles. Further optimizing $\mathbf{CRS}$, Yap designed a Half-Gaussian algorithm for integer vectors inspired by the integer HGCD method, reducing the time complexity to $O(M(n)\log n)$. \Tian{Independently, Schönhage also investigated a recursive divide-and-conquer acceleration algorithm for binary quadratic forms, with
$$
q = \left\lfloor \frac{\langle \boldsymbol{a},\boldsymbol{b} \rangle-s}{\langle \boldsymbol{b},\boldsymbol{b} \rangle} \right\rfloor.
$$
where $s = 2^m > 0$ is a preset threshold. Similar to Yap's Half-Gaussian algorithm, Schönhage's optimization also achieves $O(M(n)\log n)$ complexity and the detailed implementation has been recently presented by Ryan and Heninger \cite{RH23}. Additionally, } Rote's algorithm \cite{ROTE97} leverages Yap's lattice basis reduction \cite{yap1992fast} to solve SVP for a two-dimensional lattice module $m$ in $O(\log m (\log\log m)^2)$ time.


For the special Minkowski metric ($\ell_\infty$ norm),  Eisenbrand \cite{EIS01} and Che et al. \cite{CTJ22} subsequently investigated fast algorithms for two-dimensional lattices with a known HNF basis. Specifically, when the input lattice basis satisfies $a_1> b_1\ge 0$, $a_2=0$ and $b_2\ne 0$. it was noted that by taking 
  \begin{equation*}
   q=\left\lfloor \frac{a_1}{b_1}\right\rfloor
 \end{equation*} 
 the shortest vector can be found more efficiently than with the Lagrange reduction algorithm. 
In contrast to Eisenbrand's method \cite{EIS01} which addresses the standard HNF, Che et al. \cite{CTJ22} specifically target a specialized HNF variant characterized by the constraint \( b_2 = 1 \). Notably, Che et al. \cite{CTJ22} formalize an exact termination criterion for basis reduction and propose an optimized algorithm through systematic integration of the integer HGCD algorithm. Their implementation achieves a time complexity of \( O(M(n)\log n) \), supported by comprehensive technical details. In comparison, Eisenbrand \cite{EIS01} interprets their algorithm through continued fraction theory and theoretically suggests potential optimizations via the integer HGCD algorithm to attain the same asymptotic complexity \( O(M(n)\log n) \), though without providing explicit implementation specifics.
 Wu and Xu \cite{wu2023qin} linked Qin Jiushao's algorithm with the solution of SVP in a special kind of two-dimensional lattices with $a_1=1$, $b_1=0$ and $\gcd(a_2,b_2)=1$, and proposed an efficient method to solve the SVP of this kind of two-dimensional lattices with
   \begin{equation*}
   q=\left\lfloor \frac{a_1-1}{b_1}\right\rfloor.
 \end{equation*}  By examining each state matrix generated during the execution of the algorithm, the two vectors in the current state and their simple combinations can be utilized to extract the shortest lattice vector. 

Regarding the relationship between the shortest vector under the $\ell_\infty$ norm and that under the $\ell_2$ norm, Eisenbrand \cite{EIS01} cites Lagarias' work \cite{Laga80} to demonstrate that reduced lattice basis containing the $\ell_\infty$-shortest vector can be reduced to the $\ell_2$-shortest vector through a constant number of Lagrange reductions.

\subsection{Motivation and Our Contribution}

 Existing research indicates that the SVP (SBP) in two-dimensional lattices remains imperfectly resolved. Under the commonly used Euclidean metric,  Although Yap's Half-Gaussian algorithm \cite{yap1992fast} achieves lower complexity than Lagrange reduction, its underlying fundamental reduction algorithm (CRS) – as Yap himself acknowledges – is not inherently more efficient than classical Lagrange reduction. \Tian{While Pornin's \cite{Por20} improved Lagrange reduction algorithm achieves significant efficiency gains, this still raises a pivotal question: } \emph{Does there exist a fundamental reduction algorithm that is more efficient than the current reduction algorithms, along with an optimized version that can outperform the Half-Gaussian \Tian{and Schönhage's optimization} algorithm in finding the $\ell_2$-shortest vector (basis)}? For the special Minkowski metric, current algorithms face the following limitations:
\begin{itemize}  
\item \textbf{For HNF-form bases} \(\mathbf{B} = \begin{pmatrix} a_1 & b_1 \\ 0 & b_2 \end{pmatrix}\): While prior works suggest that integer HGCD algorithms can optimize continued fraction-based methods, detailed technical analyses and implementations remain scarce. As noted in \cite{moller2008schonhage}, “the integer HGCD algorithm is intricate, error-prone, and rarely fully detailed”, and \cite{Mor22} emphasizes its implementation challenges due to “numerous sub-cases and limited practical adoption.”  Notably, \cite{CTJ22} provides a concrete HGCD-based SVP solution for the special case \(b_2 = 1\), but \emph{the implementation details for general  $b_2$ is not available}.

\item \textbf{For general-form bases}: Current algorithms first convert inputs to HNF via (HGCD-based) extended Euclidean algorithms. However, this conversion often inflates integer sizes, especially for short initial bases, leading to longer lattice vectors. This raises two interesting questions: (i) \emph{Is HNF conversion necessary}? (ii) \emph{Can direct reduction algorithms bypass this step for a faster performance}?  
 \end{itemize}

 To address these challenges, this work aims to develop a novel fundamental reduction algorithm achieving: (1) For Euclidean metric: Superior efficiency over both Lagrange reduction and Yap's CRS, with an optimized version surpassing Half-Gaussian.  (2) For Minkowski metric: Elimination of HNF conversion and generalization of HGCD-based optimization to arbitrary $b_2$, thereby bridging the gap between theoretical possibilities and practical implementations.
Concretely, our main contributions can be summarized as follows:

\begin{enumerate}
\item \textbf{Novel Definition of a Reduced Basis}.  
For the first time, we explicitly propose a novel definition of a reduced basis distinct from the conventional Lagrange-reduced basis. This newly defined basis not only contains the $\ell_\infty$-shortest vector but also yields the \Tian{general $\ell_p (p\ge1)$}-shortest vector in at most one additional addition/subtraction step. This framework extends the theoretical foundation of lattice reduction and is expected to provide a solid foundation for further research in this field.

\item \textbf{Efficient Algorithms with Detailed Implementation and Rigorous Complexity Analysis}.  
For our newly introduced reduced basis, we design a novel fundamental algorithm named \textbf{CrossEuc} that bypasses HNF conversion, achieving a complexity of \(O(n^2)\). This approach eliminates the computational overhead associated with HNF transformation. Furthermore, we introduce a vectorized adaptation of the HGCD algorithm, termed \textbf{HVec}, which represents the first explicit extension of HGCD-like techniques to vector pairs. This innovation is of independent theoretical interest and significantly advances the efficiency of solving the SVP in two-dimensional lattices. Building on this, we present an optimized version \textbf{HVecSBP} of the fundamental algorithm \textbf{CrossEuc} that reduces the complexity to \(O(\log n \cdot M(n))\), where \(M(n)\) denotes the cost of multiplying two \(n\)-bit integers.  \Tian{Compared with prior work,  our design saves at least four multiplications per iteration without increasing the iteration count.}

\item \textbf{Comprehensive Experimental Performance Analysis}. We conduct extensive experiments comparing our methods against state-of-the-art algorithms. Results show that (1) For computing the $\ell_2$-shortest basis with input lattice basis size $2 \times 10^5$, regardless of whether the input is in HNF form or general form, our proposed fundamental reduction algorithm, \(\mathbf{CrossEuc}\), achieves approximately  $500\times$, $350\times$ and \Tian{$10$-$15\times$} over Yap's fundamental reduction algorithm \textbf{CRS}, the classical Lagrange reduction algorithm \textbf{LagRed}, \Tian{and Pornin's improved Lagrange reduction algorithm \textbf{LagRed+},} respectively
In comparison with Yap's optimized algorithm \(\mathbf{HalfGaussianSBP}\) \Tian{and the implementation of  Sch\"onhage's algorithm \textbf{Flatter},} our optimized variant, \(\mathbf{HVecSBP}\), achieves an additional speedup of approximately $14\times$ and $1.8$-$3\times$, respectively.
 (2) For computing the $\ell_\infty$-shortest basis with input lattice basis size $10^6$, our new algorithms demonstrate significant improvements:   (2.1) For HNF-form bases, \(\mathbf{HVecSBP}\) achieves $13\times$ acceleration over prior designs in worst-case scenarios. (2.2) For general-form bases, compared to the previous designs without using HGCD optimization, the proposed algorithm \textbf{CrossEuc} achieves a at least $4.5\times$ efficiency improvement. Compared to the previous designs using HGCD optimization,  when the input vectors are nearly degenerate (i.e., treated as integers), \textbf{HGCD-HNF-HVecSBP} method retain a slight advantage. As the linear dependency between vectors weakens, \textbf{HVecSBP} exhibits remarkably growing efficiency gains.
\end{enumerate}

\subsection{Road Map}
The paper is organized as follows. Section II introduces the notations, terminologies, and the mathematical concepts and properties frequently referenced in this work. In Section III, we define our proposed reduced basis and establish its theoretical properties. Section IV presents the \textbf{CrossEuc} algorithm alongside a detailed analysis. In Section V, we propose the \textbf{HVec} algorithm, discuss its analysis, and then introduce the \textbf{HVecSBP} algorithm, which iteratively invokes \textbf{HVec} to solve the shortest basis problem in two-dimensional lattices. Section VI provides a comprehensive experimental evaluation comparing the practical performance of our improved algorithms with that of existing methods, and the paper concludes with a brief summary in the final section.

\section{Preliminaries}\label{sec:Pre}
For completeness, this section introduces the necessary preliminaries for the design and analysis of our new algorithm.
\subsection{Notations and Terminologies}  
Throughout this paper, \(\mathbb{Z}\) denotes the ring of integers, and \(\mathbb{R}\) represents the field of real numbers. For \(x \in \mathbb{R}\), the round function \(\left\lceil x \right\rfloor\) represents the nearest integer to \(x\), the floor function \(\left\lfloor x \right\rfloor\) represents the greatest integer less than or equal to \(x\), and the ceiling function \(\left\lceil x \right\rceil\) represents the smallest integer greater than or equal to \(x\).  The "rounding towards zero" $\left\lfloor x\right\rfloor_o$ of \(x\) and the sign function $\operatorname{sgn}(x)$ of \(x\) are defined as follows:  
\[
\left\lfloor x\right\rfloor_o = 
\begin{cases} 
\left\lfloor x\right\rfloor, & \text{if } x \geq 0 \\
\left\lceil x \right\rceil, & \text{else}
\end{cases}, \ 
\operatorname{sgn}(x) = 
\begin{cases} 
1, & \text{if } x \geq 0 \\
-1, & \text{else}
\end{cases}.
\]  
We use $M(n)$ to denote the time complexity of the multiplication of two $n$-bit integers.  Bold upright uppercase letters denote matrices, while bold italic lowercase letters denote vectors. 
For \(x, y \in \Tian{\mathbb{Z}}\), \(\#x = \left\lceil\log_2(|x| + 1)\right\rceil\) denotes the bit size of \(x\), while \(\#(x, y) = \max\{\#x, \#y\}\) and \(\underline{\#}(x, y) = \min\{\#x, \#y\}\).
For any two integer vectors \(\boldsymbol{a} = (a_1, a_2)^T\) and \(\boldsymbol{b} = (b_1, b_2)^T\), we define $\#\boldsymbol{a}=\max\{\# a_1, \# a_2\}$, \(\#(\boldsymbol{a}, \boldsymbol{b}) = \max\{\#\boldsymbol{a}, \#\boldsymbol{b}\}\) and \(\underline{\#}(\boldsymbol{a}, \boldsymbol{b}) = \min\{\#\boldsymbol{a}, \#\boldsymbol{b}\}\).  For a vector \(\boldsymbol{x} = (x_1, x_2, \ldots, x_m) \in \mathbb{R}^m\) and $1\le p<\infty$, its \(\ell_p\) norm is defined as \(\|\boldsymbol{x}\|_p = (\sum_{i=1}^m |x_i|^p)^{1/p}\),
and the Minkowski length (\(\ell_\infty\) norm) is defined as \(\|\boldsymbol{x}\|_\infty = \max_{1 \leq i \leq m} |x_i|\).

\subsection{Lattices and related properties}
In this section, we introduce the definition of lattices and some necessary properties we used in the rest of the paper.

\begin{definition}[\cite{MGbook}]\label{def:Lattice}
(Lattice). Given linearly independent vectors $\boldsymbol{b_1},\boldsymbol{b_2},...,\boldsymbol{b_d}\in\mathbb{R}^{m}$,  the lattice generated by these vectors is defined as \begin{align*}
\mathcal{L}=L(\mathbf{B})=L(\boldsymbol{b_1},\boldsymbol{b_2},...,\boldsymbol{b_d})=\left\{\sum_{i=1}^dz_i\boldsymbol{b}_i:z_i\in\mathbb{Z}\right\},
\end{align*}
where $d$ and $m$ are called the rank and the dimension of the lattice, respectively. Specially, if $d=m$, the lattice is called full-rank.
\end{definition}
\begin{definition}[\cite{MGbook}]\label{def:Determinant}
(Determinant). Let $\mathcal{L}=L(\mathbf{B})$ be a lattice of rank $d$.  The determinant of $\mathcal{L}$, denoted $\det(\mathcal{L})=\sqrt{\det\left(\mathbf{B}^{T}\mathbf{B} \right)}$. Specially, if $\mathcal{L}$ is full-rank, $\det(\mathcal{L})=|\det(\mathbf{B})|$.
\end{definition}
\begin{definition}[\cite{MGbook}]
(Unimodular Matrix). A matrix $\mathbf{U}\in \mathbb{Z}^{d\times d}$ is called unimodular if $\det(\mathbf{U})=\pm 1$.
\end{definition}
\begin{lemma}[\cite{MGbook}]\label{lem:B=BU}
Two bases $\mathbf{B_1},\mathbf{B_2}\in \mathbb{R}^{m\times d}$ are equivalent if and only if $\mathbf{B_2}=\mathbf{B_1}\mathbf{U}$ for some unimodular matrix $\mathbf{U}$.
\end{lemma}
\begin{definition}[\cite{MGbook}]\label{def:Successive minimum}
(Successive Minimum). Let $\mathcal{L}$ be a lattice of rank $d$. For any $k\in \left\{1,...,d \right\}$, the $k$th successive minimum is defined as
\begin{align*}
   \lambda_k(\mathcal{L})=\mathrm{inf}\{r>0: \mathrm{dim}(\mathrm{span}(\mathcal{B}(\boldsymbol{o},r)\cap \mathcal{L}))\ge k\}, 
\end{align*}
where, for any computable norm $\|\cdot\|$, $\mathcal{B}(\boldsymbol{o},r)=\left\{ \boldsymbol{x}\in \mathbb{R}^{m}:||\boldsymbol{x}||<r \right\}$ represents the open ball of radius $r$ centered at the origin.
\end{definition}

\Tian{For a two-dimensional lattice with basis $\mathbf{B}=[\boldsymbol{a}\ \boldsymbol{b}]$, Lagrange \cite{Lag73} introduced the concept of a reduced basis using the language of binary quadratic forms. Equivalently, $\mathbf{B}=[\boldsymbol{a}\ \boldsymbol{b}]$ is reduced if and only if $|\langle\boldsymbol{a},\boldsymbol{b}\rangle| \le \frac{1}{2}\langle\boldsymbol{b},\boldsymbol{b}\rangle$ and $|\langle\boldsymbol{a},\boldsymbol{b}\rangle| \le \frac{1}{2}\langle\boldsymbol{a},\boldsymbol{a}\rangle$. Micciancio and Goldwasser \cite{MGbook} further generalized the above definition to any computable norm. That is,}
\begin{definition}[\cite{Lag73,MGbook}]
Let $\mathbf{B}=[\boldsymbol{a}\ \boldsymbol{b}]$ with $\boldsymbol{a}=(a_1\ a_2)^T$ and $\boldsymbol{b}=(b_1\ b_2)^T$ be a basis of the lattice $\mathcal{L}=L(\mathbf{B})$. The basis is called Lagrange-reduced if it satisfies the condition: $\|\boldsymbol{a}\|,\|\boldsymbol{b}\|\le \|\boldsymbol{a}+\boldsymbol{b}\|,\|\boldsymbol{a}-\boldsymbol{b}\|.$
Here, $\|\cdot\|$ refers to any computable norm.
\end{definition}
For a Lagrange-reduced basis, the following result is well-known:
\begin{lemma}[\cite{MGbook}]\label{lem:Lag}
If $\mathbf{B}=[\boldsymbol{a}\ \boldsymbol{b}]$ is Lagrange-reduced, then for any computable norm $\|\cdot\|$,  we have
$\min\{ \|\boldsymbol{a}\|,\|\boldsymbol{b}\|\}=\lambda_1(\mathcal{L})$ and $\max\{\|\boldsymbol{a}\|,\|\boldsymbol{b}\|\}=\lambda_2(\mathcal{L})$. 
\end{lemma}
The above property is derived from the following general fact:
\begin{lemma}[\cite{MGbook}] \label{lem:key}
For any two vectors $\boldsymbol{x}$, $\boldsymbol{y}$ and any computable norm $\|\cdot\|$, if $\|\boldsymbol{x}\|\le \ (resp. <)\|\boldsymbol{x}+\boldsymbol{y}\|$, then $\|\boldsymbol{x}+\boldsymbol{y}\|\le\  (resp. <)\|\boldsymbol{x}+\alpha\boldsymbol{y}\|$ for any $\alpha>1$·
\end{lemma}

\subsection{Hermite Normal Form}
The Hermite Normal Form (HNF) of a matrix is a canonical form used in linear algebra and number theory for integer matrices. 
\begin{definition}[\cite{PYB2019}]
(Hermite Normal Form). A nonsingular matrix $\mathbf{H} \in \mathbb{Z}^{d\times d}$ is said to be in HNF if (1) $h_{i,i}>0$ for $1\leq i \leq d$. (2) $h_{j,i}=0$ for $i<j\leq d$. (3) $0\leq h_{j,i}<h_{i,i}$ for $1\leq j<i$.
\end{definition}
For any nonsingular integer matrix \(\boldsymbol{a}\), there exists a unique matrix \(\mathbf{H}\) in Hermite Normal Form (HNF) and a unique unimodular matrix \(\mathbf{U}\) such that \(\mathbf{H} = \mathbf{U}\boldsymbol{a}\) \cite{PYB2019}. Specifically, for any two-dimensional lattice, the basis matrix can be efficiently converted into its HNF using the following lemma.  
\begin{lemma}[\cite{EIS01}]
Given a matrix $\mathbf{B}=\begin{pmatrix}
a_1 & b_1\\
a_2 & b_2
\end{pmatrix}\in \mathbb{Z}^{2\times 2}$, let $c=\gcd(a_2,b_2)=xa_2+yb_2$ be the greatest common divisor of $a_2$ and $b_2$, then
\begin{align*}
\begin{pmatrix}
a_1 & b_1\\
a_2 & b_2
\end{pmatrix}
\begin{pmatrix}
b_2/\gcd(a_2,b_2)  & x\\
-a_2/\gcd(a_2,b_2) & y
\end{pmatrix}=
\begin{pmatrix}
a & b\\
0 & c
\end{pmatrix}\in \mathbb{Z}^{2\times 2}
\end{align*}
with $a=(a_1b_2-a_2b_1)/\gcd(a_2,b_2)$ and $b=a_1x+b_1y$.
By applying elementary vector transformations, we can ensure that the matrix satisfy the HNF conditions.
\end{lemma}

\section{Our new reduced basis and its properties}\label{sec:relation}
In this section, we first introduce a new reduced basis and then present its properties. Throughout this section and the rest of the paper, unless otherwise specified, the notations \(\|\boldsymbol{x}\|\) is used to represent the Minkowski metric (\(\ell_\infty\) norm) of \(\boldsymbol{x}\). 
\begin{definition}\label{def:Eucred}
Given a lattice $\mathcal{L}=L(\mathbf{B})$ with a basis 
\begin{align}\label{eq:basis}
   \mathbf{B}=[\boldsymbol{a}\ \boldsymbol{b}]=\begin{pmatrix}
		a_{1} & b_{1}\\
		a_{2} & b_{2} 
	\end{pmatrix}\in\mathbb{R}^{2\times 2},
\end{align} we call $\mathbf{B}$ is reduced if 
\begin{equation}\label{eq:con}
a_1a_2b_1b_2\le 0 \wedge (|a_1|-|a_2|)(|b_1|-|b_2|)\le 0
\end{equation}	
\end{definition}
It should be remarked that the equation (\ref{eq:con}) also can be equivalently reformulated as
$a_1a_2b_1b_2\le 0 \wedge |a_1b_1-a_2b_2|\le |\det(\mathbf{B})|.$ \Tian{Geometrically, our new reduced basis is based on a fundamental region that partitions the two‑dimensional space into eight pie slices along the $x$, $y$ and diagonal axes. A basis $[\boldsymbol{a}\ \boldsymbol{b}]$ is reduced if $\boldsymbol{a}$ and $\boldsymbol{b}$ lie in “orthogonal” slices. Fig. \ref{fig:reducedbasis} provides an intuitive visual comparison between our new reduced basis and the classical Lagrange reduced basis.}

\begin{figure}[htp]
	\centering
	\subfigure[\Tian{Illustration of Lagrange reduced basis $[\boldsymbol{a}\ \boldsymbol{b}]$ under Euclidean metric with $\|\boldsymbol{a}\|\le \|\boldsymbol{b}\|\le\|\boldsymbol{a}\pm\boldsymbol{b}\|$}]{
		\includegraphics[width=0.35\textwidth]{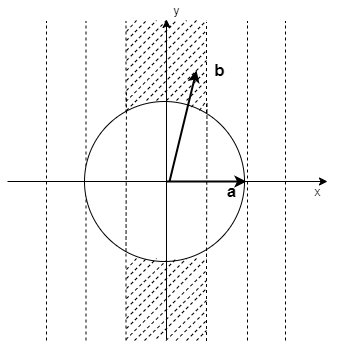}
		\label{fig:lagrange}
	}\qquad
	\subfigure[\Tian{Illustration of the new reduced basis where $\boldsymbol{a}$ forms a reduced basis with both $\boldsymbol{b}$ and $\boldsymbol{b}^\prime$}]{
		\includegraphics[width=0.35\textwidth]{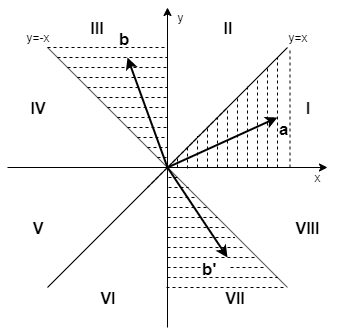}
		\label{fig:newbasis}
	}
	\caption{\Tian{Comparison of reduced bases}}
	\label{fig:reducedbasis}
\end{figure}

Now we present two important properties of the new defined reduced basis. The first theorem establishes the minimality of the reduced basis under the Minkowski ($\ell_{\infty}$) norm, while the second theorem demonstrates a similar property under the Euclidean ($\ell_2$) norm. 

\begin{theorem}\label{thm:Eucred}
Let $\mathcal{L}=L(\mathbf{B})$ be a lattice with basis $\mathbf{B}$ defined in equation (\ref{eq:basis}), and let $\lambda_1(\mathcal{L})$ and $\lambda_2(\mathcal{L})$ denote the successive minima under $\ell_\infty$ norm. If $\mathbf{B}$ is reduced, then
$\lambda_1(\mathcal{L})=\min\{\|\boldsymbol{a}\|,\|\boldsymbol{b}\|\}.$ Further,
let $\boldsymbol{c}=(c_1\ c_2)\in\mathcal{L}$ be a vector achieving $\lambda_2(\mathcal{L})$. Then $\boldsymbol{c}=\mathrm{argmax}\{\|\boldsymbol{a}\|,\|\boldsymbol{b}\|\}-z\cdot\mathrm{argmin}\{\|\boldsymbol{a}\|,\|\boldsymbol{b}\|\}$, where $z$ is determined as follows:

(1) If $a_1b_1\ge 0,a_2b_2\leq 0$, then
\begin{align*}
z=\left\{
\begin{array}{ll}
\left\lfloor\frac{|b_1|-|b_2|}{|a_1|+|a_2|}\right\rfloor\ \mbox{or}\ \left\lceil\frac{|b_1|-|b_2|}{|a_1|+|a_2|}\right\rceil,&\|\boldsymbol{a}\|\le \|\boldsymbol{b}\|\\
\left\lfloor\frac{|a_1|-|a_2|}{|b_1|+|b_2|}\right\rfloor\ \mbox{or}\ \left\lceil\frac{|a_1|-|a_2|}{|b_1|+|b_2|}\right\rceil,&\|\boldsymbol{a}\|>\|\boldsymbol{b}\|
\end{array}
\right.
\end{align*}

(2) If $a_1b_1<0,a_2b_2\ge 0$, then
\begin{align*}
z=\left\{
\begin{array}{ll}
\left\lfloor\frac{|b_2|-|b_1|}{|a_1|+|a_2|}\right\rfloor\ \mbox{or}\ \left\lceil\frac{|b_2|-|b_1|}{|a_1|+|a_2|}\right\rceil,&\|\boldsymbol{a}\|\le \|\boldsymbol{b}\|\\
\left\lfloor\frac{|a_2|-|a_1|}{|b_1|+|b_2|}\right\rfloor\ \mbox{or}\ \left\lceil\frac{|a_2|-|a_1|}{|b_1|+|b_2|}\right\rceil,&\|\boldsymbol{a}\|>\|\boldsymbol{b}\|
\end{array}
\right.
\end{align*}
Specially,  if $a_1a_2b_1b_2=0$ and none of the above cases apply, negate  $\boldsymbol{a}$ or $\boldsymbol{b}$ as necessary to fit into one of the above cases.
\end{theorem}

\begin{proof}\Tian{ Without loss of generality, we only consider the case $a_1b_1\ge 0$ and $a_2b_2\le 0$. If the sign condition differs (e.g., $a_1b_1<0$, $a_2b_2\ge0$), we simply swap the roles of the coordinates or negate vectors appropriately.
We first show $\lambda_1(\mathcal{L})=\min\{\|\boldsymbol{a}\|,\|\boldsymbol{b}\|\}$. 
Consider any nonzero lattice vector $\boldsymbol{v}=z_1\boldsymbol{a}+z_2\boldsymbol{b}$. If $z_1=0$ or $z_2=0$, then clearly $\|\boldsymbol{v}\|\ge\|\boldsymbol{a}\|$ or $\|\boldsymbol{b}\|$, so $\|\boldsymbol{v}\|\ge\min\{\|\boldsymbol{a}\|,\|\boldsymbol{b}\|\}$. If $z_1z_2\neq0$, we use the reduction condition $(|a_1|-|a_2|)(|b_1|-|b_2|)\le0$, which implies that the larger component of $\boldsymbol{a}$ and the larger component of $\boldsymbol{b}$ lie in opposite coordinates. Together with the sign condition $a_1a_2b_1b_2\le0$ (ensuring that the signs are aligned so that $a_1$ and $b_1$ have the same sign, and $a_2$ and $b_2$ opposite signs after possible negations), one can verify that either $|z_1a_1+z_2b_1|\ge\max
\{|a_1|,|b_1|\}\ge\min\{\|\boldsymbol{a}\|,\|\boldsymbol{b}\|\}$ or $|z_1a_2+z_2b_2|\ge\max\{|a_2|,|b_2|\}\ge\min\{\|\boldsymbol{a}\|,\|\boldsymbol{b}\|\}$, depending on the sign of $z_1z_2$. In either case, $\|\boldsymbol{v}\|\ge\max\{\|\boldsymbol{a}\|,\|\boldsymbol{b}\|\}\ge\min\{\|\boldsymbol{a}\|,\|\boldsymbol{b}\|\}$. This confirms $\lambda_1(\mathcal{L})=\min\{\|\boldsymbol{a}\|,\|\boldsymbol{b}\|\}$.}

Now assume without loss that $\|\boldsymbol{a}\|\le\|\boldsymbol{b}\|$, so $\lambda_1(\mathcal{L})=\|\boldsymbol{a}\|$. Define $f(x)=\|\boldsymbol{b}-x\boldsymbol{a}\|=\max\{|b_1-xa_1|,|b_2-xa_2|\}.$ 
Let $f(z)=\min_{z\in\mathbb{Z}}f(x)$. Then, $\|\boldsymbol{a}\|\le f(z)=\|\boldsymbol{b}-z\boldsymbol{a}\|\le f(z+1)=\|\boldsymbol{b}-(z+1)\boldsymbol{a}\|=\|\boldsymbol{a}-(\boldsymbol{b}-z\boldsymbol{a})\|$ and $f(z-1)=\|\boldsymbol{b}-(z-1)\boldsymbol{a}\|=\|\boldsymbol{a}+(\boldsymbol{b}-z\boldsymbol{a})\|$. By \textbf{Lemma \ref{lem:Lag}}, it follows that $\lambda_2(\mathcal{L})=f(z)$. Thus, we only need to estimate $z$. We now consider the case where  $a_1a_2\neq 0$.  The argument is essentially the same for cases where $a_1=0$ or $a_2=0$. From the definitions of $f_1(x)$ and $f_2(x)$, we have
\begin{align*}
f_1(x)=\left\{\begin{array}{ll}
|a_1|x-|b_1|,& x\ge \frac{b_1}{a_1}\\
-|a_1|x+|b_1|,& x< \frac{b_1}{a_1}
\end{array}\right.,\
f_2(x)=\left\{\begin{array}{ll}
|a_2|x+|b_2|,& x\ge \frac{b_2}{a_2}\\
-|a_2|x-|b_2|,& x< \frac{b_2}{a_2}
\end{array}\right..
\end{align*}
The minimum $f(z)$ is achieved at the integer closest to the point of intersection where $-|a_1|x+|b_1|=|a_2|x+|b_2|$ (\emph{i.e.}, $x=\frac{|b_1|-|b_2|}{|a_1|+|a_2|}$). Therefore,
\begin{equation*}
z=\left\lfloor\frac{|b_1|-|b_2|}{|a_1|+|a_2|}\right\rfloor\ \mbox{or}\ \left\lceil\frac{|b_1|-|b_2|}{|a_1|+|a_2|}\right\rceil.
\end{equation*}
The case $\|\boldsymbol{a}\|>\|\boldsymbol{b}\|$ is symmetric, considering $g(y)=\|\boldsymbol{a}-y\boldsymbol{b}\|$, leading to $y = \frac{|a_1|-|a_2|}{|b_1|+|b_2|}$.
\end{proof}

\begin{theorem}\label{thm:l2}
\Tian{Let $\mathcal{L}=L(\mathbf{B})$ be a lattice with basis $\mathbf{B}=[\boldsymbol{a}\ \boldsymbol{b}]$ as in (\ref{eq:basis}), and let $\lambda_1^{(p)}(\mathcal{L})$ and $\lambda_2^{(p)}(\mathcal{L})$ denote its successive minima with respect to the $\ell_p$ norm ($1\le p<\infty$). If $\mathbf{B}$ is reduced, then the vector
$\boldsymbol{u}= \operatorname{argmin}\bigl\{\|\boldsymbol{a}\|_p,\ \|\boldsymbol{b}\|_p,\ \|\boldsymbol{a}+\boldsymbol{b}\|_p,\ \|\boldsymbol{a}-\boldsymbol{b}\|_p\bigr\}$ is a shortest non‑zero lattice vector, i.e. $\lambda_1^{(p)}(\mathcal{L})=\|\boldsymbol{u}\|_p$. Moreover, let $\boldsymbol{x}$ be the vector in $\{\boldsymbol{a},\boldsymbol{b}\}\setminus\{\boldsymbol{u}\}$ and set
$\boldsymbol{v}= \boldsymbol{x} - q\boldsymbol{u}$ with $q = \operatorname{argmin}_{z\in\mathbb{Z}} \|\boldsymbol{x}-z\boldsymbol{u}\|_p.$
Then $\boldsymbol{v}$ attains the second successive minimum: $\lambda_2^{(p)}(\mathcal{L})=\|\boldsymbol{v}\|_p$. For the Euclidean case $p=2$, the optimal integer $q$ is given by the nearest integer to the ratio of inner products:$
q = \left\lceil\frac{\langle\boldsymbol{x},\boldsymbol{u}\rangle}{\langle\boldsymbol{u},\boldsymbol{u}\rangle}\right\rfloor.$}
\end{theorem}

 \begin{proof}
\Tian{Since the reduction conditions: $a_1b_1a_2b_2\le0$ and $(|a_1|-|a_2|)(|b_1|-|b_2|)\le0$, without loss of generality, we assume $\|\boldsymbol{a}\|_p \le \|\boldsymbol{b}\|_p$, $a_1b_1\ge 0$, and $a_2b_2\le 0$. We first prove that $\lambda_1(\mathcal{L}) = \|\boldsymbol{u}\|_p$. That is, we need to prove 
\begin{align}\label{eq:lpmin}
\forall (z_1\ z_2)\in \mathbb{Z}^2\backslash \{\mathbf{0}\}, \|\boldsymbol{w}\|_p=\|z_1\boldsymbol{a}+z_2\boldsymbol{b}\|_p\geq \|\boldsymbol{u}\|_p.
\end{align}
	 Clearly, if $z_1z_2=0$, $|z_1|=|z_2|=1$ or $z_1=z_2$, we can easily deduce that $\|\boldsymbol{w}\|_p\ge \min\{\|\boldsymbol{a}\|_p,\|\boldsymbol{b}\|_p, \|\boldsymbol{a}+\boldsymbol{b}\|_p, \|\boldsymbol{a}-\boldsymbol{b}\|_p\}=\|\boldsymbol{u}\|_p$. Thus, we assume $|z_1|,|z_2|\ge 1$ and  $|z_1|\neq |z_2|$; in particular, 
$\max\{|z_1|,|z_2|\}\ge 2$.
\paragraph{Case \(z_1z_2>0\)}
Then we can write
			\begin{align}\label{eq:z1z2less11}
			\|\boldsymbol{w}\|_p^p=\|z_1\boldsymbol{a}+z_2\boldsymbol{b}\|_p^p = (|z_1||a_1|+|z_2||b_1|)^p+(\big||z_1||a_2|-|z_2||b_2|\big|)^p 
			\end{align}
			We consider subcases based on the relative sizes of \(|a_1|,|a_2|,|b_1|,|b_2|\).
\begin{itemize}
    \item If \(|b_1|\ge |a_2|\), then the equation (\ref{eq:z1z2less11})$\ge (|z_1||a_1|+|z_2||b_1|)^p \geq |z_1|^p |a_1|^p+|z_2|^p |b_1|^p\ge |a_1|^p+|a_2|^p\geq \|\boldsymbol{a}\|_p^p\ge \|\boldsymbol{u}\|_p^p$.
    \item If \(|a_1|\ge |b_2|\), similarly the equation (\ref{eq:z1z2less11})$\ge (|z_1||a_1|+|z_2||b_1|)^p \geq |z_1|^p |a_1|^p+|z_2|^p |b_1|^p\ge |b_1|^p+|b_2|^p\geq \|\boldsymbol{b}\|_p^p\ge \|\boldsymbol{u}\|_p^p$.
    \item Otherwise we have \(|b_1|<|a_2|\) and \(|a_1|<|b_2|\). The reduction condition \((|a_1|-|a_2|)(|b_1|-|b_2|)\le0\) forces one of the two orderings: (i) \(|b_1|<|a_2|\le |a_1|<|b_2|\) and (ii) \(|a_1|<|b_2|\le |b_1|<|a_2|\).
   In subcase (i), if \(|z_1|>|z_2|\), then 
   $$\|\boldsymbol{w}\|_p^p \ge (|z_1||a_1|+|z_2||b_1|)^p \ge |z_1|^p |a_1|^p \ge 4|a_1|^p \ge |a_1|^p+3|a_2|^p \ge \|\boldsymbol{a}\|_p^p\ge \|\boldsymbol{u}\|_p^p.$$
    If \(|z_1|<|z_2|\), then we have
    \[
    \|\boldsymbol{w}\|_p^p \ge (|z_1||a_1|+|z_2||b_1|)^p + \big(|b_2|+(|z_2|-1)|b_2|-|z_1||a_2|\big)^p \ge |z_2|^p |b_1|^p + |b_2|^p \ge \|\boldsymbol{b}\|_p^p\ge \|\boldsymbol{u}\|_p^p.
    \]
    Subcase (ii) is symmetric and yields the same conclusions by swapping roles of \(\boldsymbol{a}\) and \(\boldsymbol{b}\) and considering \(|z_1|\) vs \(|z_2|\) appropriately.
\end{itemize}
\paragraph{Case \(z_1z_2<0\)}
Now \(z_1\) and \(z_2\) have opposite signs, and 
\begin{align}\label{eq:z1z2less0}
\|\boldsymbol{w}\|_p^p=\|z_1\boldsymbol{a}+z_2\boldsymbol{b}\|_p^p = (\big||z_1||a_1|-|z_2||b_1|\big|)^p+(|z_1||a_2|+|z_2||b_2|)^p.
\end{align} Following the same reasoning as in {Case \(z_1z_2>0\)}, we can argue similarly by considering the subcases \(|a_2|\ge |b_1|\), \(|b_2|\ge |a_1|\) and the two orderings (i) \(|b_2|<|a_1|\le |a_2|<|b_1|\) and (ii) \(|a_2|<|b_1|\le |b_2|<|a_1|\). The logic is completely analogous, so we omit the details here for brevity.}

\Tian{Therefore, in every possible situation, we have shown \(\|\boldsymbol{w}\|_p^p \ge \|\boldsymbol{u}\|_p^p\), establishing \(\lambda_1(\mathcal{L}) = \|\boldsymbol{u}\|_p\).
}

\Tian{To prove \(\lambda_2(\mathcal{L}) = \|\boldsymbol{v}\|_p\), by \textbf{Lemma \ref{lem:Lag}}, it suffices to show that the pair \([\boldsymbol{u}, \boldsymbol{v}]\) is Lagrange-reduced, i.e.,  
\[
\|\boldsymbol{u}\|_p \leq \|\boldsymbol{v}\|_p \leq \min\left\{\|\boldsymbol{u}+\boldsymbol{v}\|_p, \|\boldsymbol{u}-\boldsymbol{v}\|_p\right\}.
\]  
In fact, by the choice of \(q\), we have:  
\[
\begin{aligned}
\|\boldsymbol{v}\|_p &= \|\boldsymbol{x} - q\boldsymbol{u}\|_p \leq \|\boldsymbol{x} - (q+1)\boldsymbol{u}\|_p = \|\boldsymbol{v} - \boldsymbol{u}\|_p, \\
\|\boldsymbol{v}\|_p &= \|\boldsymbol{x} - q\boldsymbol{u}\|_p \leq \|\boldsymbol{x} - (q-1)\boldsymbol{u}\|_p = \|\boldsymbol{v} + \boldsymbol{u}\|_p.
\end{aligned}
\]  
Thus,  \([\boldsymbol{u}\ \boldsymbol{v}]\) is  Lagrange-reduced and \(\lambda_2(\mathcal{L}) = \|\boldsymbol{v}\|_p\).
Clearly, for the Euclidean case $p=2$, by the definition of rounding,  \(\boldsymbol{v} = \boldsymbol{x} - q\boldsymbol{u}\) achieves its minimum at
\[
q = \left\lceil \frac{\langle\boldsymbol{x},\boldsymbol{u}\rangle}{\langle\boldsymbol{u},\boldsymbol{u}\rangle} \right\rfloor \in \mathbb{Z}.
\]
This completes the proof.}
\end{proof}


\section{Our new reduced algorithm and its analysis}
 
Building upon our new defined reduced basis and its properties, we introduce an algorithm that directly transforms any given two-dimensional lattice basis into a reduced basis. This approach eliminates the necessity for an initial conversion to HNF as employed by Eisenbrand \cite{EIS01}.

\subsection{Design Idea}
For any input lattice basis  $[\boldsymbol{a}\ \boldsymbol{b}]$ with $\boldsymbol{a}=(a_1\ a_2)^T$, $\boldsymbol{b}=(b_1\ b_2)^T$ and $\|\boldsymbol{a}\|\ge \|\boldsymbol{b}\|$, if it is not reduced, then we have two possible cases:
(1) $a_1a_2b_1b_2>0$ or (2) $a_1a_2b_1b_2\le 0$ and $(|a_1|-|a_2|)(|b_1|-|b_2|)>0$. 

We first consider the case (2), without loss of generality, assume $a_1\ge b_1\ge 0$ and $a_2b_2\le 0$.  If these conditions are not met, we can adjust the basis vectors using negation or swapping operations to achieve them.

If $|a_1|>|a_2|$ and $|b_1|>|b_2|$, inspired by prior work \cite{EIS01,CTJ22}, we can compute a new basis $[\boldsymbol{b}\ \boldsymbol{c}]$ with $\boldsymbol{c}=\boldsymbol{a}-q\boldsymbol{b}$, where $q=\lfloor\frac{a_1}{b_1}\rfloor_o>0$. Then
\begin{align*}
a_1>b_1>c_1\ge 0,\ b_2c_2\le 0,\ |c_2|=|a_2-qb_2|=|a_2|+q|b_2|>|b_2|.
\end{align*}
As a result, $b_1b_2c_1c_2\le 0$. If $|c_1|-|c_2|\le 0$, the new basis $[\boldsymbol{b}\ \boldsymbol{c}]$ is reduced. Otherwise, we have $|b_1|-|b_2|>|c_1|-|c_2|>0$. Repeating this process iteratively reduces the difference between the absolute values of the first and second coordinates of the new vector until the terminal condition (the equation (\ref{eq:con})) is satisfied. 

If $|a_1|<|a_2|$ and $|b_1|<|b_2|$, in this case, we focus on the second coordinate. Using a similar process, we compute a new basis $[\boldsymbol{b}\ \boldsymbol{c}]$, where $\boldsymbol{c}=\boldsymbol{a}-q\boldsymbol{b}$ and $q=\lfloor\frac{a_2}{b_2}\rfloor_o\le0$. Then 
\begin{align*}
c_1>a_1>b_1\ge 0,\ b_2c_2\le 0,\ |c_2|=|a_2-qb_2|=|a_2|-|q||b_2|<|b_2|.
\end{align*}
Thus, $b_1b_2c_1c_2\le 0$. If $|c_1|-|c_2|\ge 0$, the new basis $[\boldsymbol{b}\ \boldsymbol{c}]$ is reduced. Otherwise, we have $|b_1|-|b_2|<|c_1|-|c_2|<0$. Repeating this process iteratively reduces the difference between the absolute values of the coordinates until the terminal condition (the equation (\ref{eq:con})) is satisfied.

For the case (1). When $a_1a_2b_1b_2>0$, the key is to transform the basis such that $a_1a_2b_1b_2\le0$. Since $\boldsymbol{a}$ and $\boldsymbol{b}$ are linearly independent, it holds that $\frac{a_1}{b_1}\neq \frac{a_2}{b_2}$. Let $q_1\triangleq\lfloor\frac{a_1}{b_1}\rfloor_o, q_2\triangleq \lfloor\frac{a_2}{b_2}\rfloor_o$ and, if $q_1\neq q_2$, choose
\[
q = 
\begin{cases} q_1, & \text{if }  |a_1| \geq |a_2| \wedge q_1> q_2, \\
q_1+1, & \text{if } |a_1| \geq |a_2| \wedge q_1 < q_2,\\
q_2, & \text{if } |a_1| < |a_2| \wedge q_2 > q_1, \\
q_2+1, & \text{if }  |a_1| < |a_2| \wedge q_2 < q_1.
\end{cases}\
\] 
Then the new basis $[\boldsymbol{b}\ \boldsymbol{c}]$ with $\boldsymbol{c}=\boldsymbol{a}-q\boldsymbol{b}$ ensures $b_1b_2c_1c_2\le0$, achieving the goal. Otherwise $q=q_1=q_2$, iteratively reduce the basis $[\boldsymbol{b}\ \boldsymbol{c}]$ with the same method until $q_1\neq q_2$.

\subsection{Algorithm Description and Analysis}
Based on the proposed design concept, the detailed steps of our algorithm are provided in \textbf{Algorithm \ref{alg:CrossEuc}}. It should be noted that the \textbf{Algorithm \ref{alg:CrossEuc}} invokes two sub-algorithms: \textbf{UMTrans1}$^\star$ and \textbf{UMTrans2}$^\star$, which are variants of \textbf{UMTrans1} (\textbf{Algorithm \ref{alg:TransPositive}}) and \textbf{UMTrans2} (\textbf{Algorithm \ref{alg:TransNegative}}), respectively. The only difference is whether the returned result includes the integer \(q\).    

Before discussing the correctness and complexity of \textbf{Algorithm \ref{alg:CrossEuc}},  we first present two lemmas to analyze the properties of \textbf{Algorithm \ref{alg:TransPositive}} and \textbf{Algorithm \ref{alg:TransNegative}}.

\begin{algorithm}[!htp]
\caption{\textbf{UMTrans1}($\boldsymbol{a},\boldsymbol{b}$)}\label{alg:TransPositive}
\begin{algorithmic}[1]
\REQUIRE  A basis $\mathbf{B}=[\boldsymbol{a}\ \boldsymbol{b}]=\begin{pmatrix}
a_1 & b_1 \\
a_2 & b_2
\end{pmatrix}$ with \Tian{$\|\boldsymbol{a}\|\ge\|\boldsymbol{b}\|$ and}
 $a_1a_2b_1b_2>0$
\ENSURE  A new basis \Tian{$\boldsymbol{a}^\prime=(a_1\ a_2)^T$, $\boldsymbol{b}^\prime=(b_1\ b_2)^T$} and an integer $q$
\STATE \textbf{if} ($|a_1|\geq|a_2|$)
\STATE \quad $q :=\left\lfloor\frac{a_1}{b_1}\right\rfloor_o$ 
\STATE  \quad \textbf{if} ($|a_2-qb_2|\geq|b_2|\wedge sgn(a_2)=sgn(a_2-qb_2)$), $q := q+1$\label{alg1:step3}
\STATE \textbf{else}
\STATE \quad $q :=\left\lfloor\frac{a_2}{b_2}\right\rfloor_o$ 
\STATE  \quad \textbf{if} ($|a_1-qb_1|\geq|b_1|\wedge sgn(a_1)=sgn(a_1-qb_1)$), $q := q+1$
\STATE (\Tian{$\boldsymbol{a}^\prime,\boldsymbol{b}^\prime$}) := ($\boldsymbol{b},\boldsymbol{a}-q\boldsymbol{b}$)
\STATE \textbf{return} $\boldsymbol{a}^\prime, \boldsymbol{b}^\prime, q$
\end{algorithmic}
\end{algorithm}

\begin{algorithm}[!htp]
\caption{\textbf{UMTrans2}($\boldsymbol{a},\boldsymbol{b}$)}\label{alg:TransNegative}
\begin{algorithmic}[1]
\REQUIRE  A basis $\mathbf{B}=[\boldsymbol{a}\ \boldsymbol{b}]=\begin{pmatrix}
a_1 & b_1 \\
a_2 & b_2
\end{pmatrix}$ with \Tian{$\|\boldsymbol{a}\|\ge\|\boldsymbol{b}\|$,} $a_1a_2b_1b_2\leq0$ and $(|a_1|-|a_2|)(|b_1|-|b_2|)>0$
\ENSURE  A new basis \Tian{$\boldsymbol{a}^\prime=(a_1\ a_2)^T$, $\boldsymbol{b}^\prime=(b_1\ b_2)^T$} and an integer $q$
\STATE \textbf{if} ($|a_1|>|a_2|$)
\STATE \quad $q :=\left\lfloor\frac{a_1}{b_1}\right\rfloor_o$ 
\STATE \textbf{else}
\STATE \quad $q :=\left\lfloor\frac{a_2}{b_2}\right\rfloor_o$ 
\STATE (\Tian{$\boldsymbol{a}^\prime, \boldsymbol{b}^\prime$}) := ($\boldsymbol{b},\boldsymbol{a}-q\boldsymbol{b}$)
\STATE \textbf{return} $\boldsymbol{a}^\prime,\boldsymbol{b}^\prime, q$
\end{algorithmic}
\end{algorithm}

\begin{algorithm}[!htp]
\caption{\textbf{CrossEuc}($\boldsymbol{a},\boldsymbol{b}$)}\label{alg:CrossEuc}
\begin{algorithmic}[1]
\REQUIRE  A basis $\mathbf{B}=[\boldsymbol{a}\ \boldsymbol{b}]\in\mathbb{Z}^{2\times 2}$ with $\boldsymbol{a}=(a_1\ a_2)^T$, $\boldsymbol{b}=(b_1\ b_2)^T$
\ENSURE  A shortest basis $\boldsymbol{a}'=(a_1\ a_2)^T$, $\boldsymbol{b}'=(b_1\ b_2)^T$ such that $\|\boldsymbol{a}'\|=\lambda_1(L(\mathbf{B}))$ and $\|\boldsymbol{b}'\|=\lambda_2(L(\mathbf{B}))$ 
\STATE  \textbf{If} $\|\boldsymbol{a}\|<\|\boldsymbol{b}\|$, Swap($\boldsymbol{a}, \boldsymbol{b}$)
\STATE \textbf{While} ($a_1a_2b_1b_2>0$), \textbf{do}
\STATE \quad $(\boldsymbol{a}, \boldsymbol{b})\leftarrow\textbf{UMTrans1}^\star(\boldsymbol{a}, \boldsymbol{b})$
\STATE \textbf{EndWhile}
\STATE  \textbf{If} $\|\boldsymbol{a}\|<\|\boldsymbol{b}\|$, Swap($\boldsymbol{a}, \boldsymbol{b}$)
\STATE \textbf{While} ($(|a_1|-|a_2|)(|b_1|-|b_2|)> 0$), \textbf{do}
\STATE \quad $(\boldsymbol{a}, \boldsymbol{b})\leftarrow\textbf{UMTrans2}^\star(\boldsymbol{a}, \boldsymbol{b})$
\STATE \textbf{EndWhile}
\STATE \textbf{If} (($a_1b_1=0 \wedge sgn(a_2)\ne sgn(b_2))$ or ($a_2b_2=0 \wedge sgn(a_1)\ne sng(b_1)$)), $\boldsymbol{b} :=-\boldsymbol{b}$.
\STATE  \textbf{If} $\|\boldsymbol{a}\|> \|\boldsymbol{b}\|$, Swap($\boldsymbol{a}, \boldsymbol{b}$)
\STATE  \textbf{If} $\boldsymbol{a} =[0,0]$, \textbf{Return}\ [$\boldsymbol{a}', \boldsymbol{b}']:= [\boldsymbol{a}, \boldsymbol{b}$]
\STATE  \textbf{Else}
\STATE \quad \textbf{if} ($a_2b_2\leq0$)\label{crosseuc:step11}
\STATE  \quad \quad $\boldsymbol{b} :=\min\left\{\boldsymbol{b}-\left\lfloor\frac{|b_1|-|b_2|}{|a_1|+|a_2|}\right\rfloor \boldsymbol{a}, \boldsymbol{b}-\left\lceil\frac{|b_1|-|b_2|}{|a_1|+|a_2|}\right\rceil \boldsymbol{a}\right\}$
\STATE \quad \textbf{else}
\STATE  \quad \quad $\boldsymbol{b} :=\min\left\{\boldsymbol{b}-\left\lfloor\frac{|b_2|-|b_1|}{|a_1|+|a_2|}\right\rfloor \boldsymbol{a}, \boldsymbol{b}-\left\lceil\frac{|b_2|-|b_1|}{|a_1|+|a_2|}\right\rceil \boldsymbol{a}\right\}$\label{crosseuc:step14}
\STATE  \quad \textbf{Return}  [$\boldsymbol{a}', \boldsymbol{b}']:= [\boldsymbol{a}, \boldsymbol{b}$]

\end{algorithmic}
\end{algorithm}

 \begin{lemma}\label{lem:key1}
\Tian{Let $(\boldsymbol{a}', \boldsymbol{b}', q)$ be the output of \textbf{Algorithm \ref{alg:TransPositive}}. Then the matrix \(\mathbf{B}' = [\boldsymbol{a}' \ \boldsymbol{b}']\) forms a valid basis and satisfies exactly one of the following two conditions:
\begin{enumerate}
\item \( a_1' a_2' b_1' b_2' \leq 0 \), and $\bigl(\|\boldsymbol{a}'\| = \|\boldsymbol{b}\|\bigr) \;\wedge\;
\bigl(\|\boldsymbol{b}'\| < \max\{\|\boldsymbol{a}\|\cdot\|\boldsymbol{b}\|,\ \|\boldsymbol{a}\|,\ \|\boldsymbol{b}\|\}\bigr);$
\item \( a_1' a_2' b_1' b_2' > 0 \), and
$\bigl(\|\boldsymbol{b}'\| \le \min\{\|\boldsymbol{a} - \boldsymbol{b}\|,\ \|\boldsymbol{a} + \boldsymbol{b}\|\}\bigr) \;\wedge\;
\bigl(\|\boldsymbol{b}'\| < \|\boldsymbol{a}'\|\bigr) \;\wedge\;
\bigl(\|\boldsymbol{b}'\| < \tfrac{1}{2}\|\boldsymbol{a}\|\bigr).$
\end{enumerate}}
\end{lemma}

\begin{proof}
 Since $
\mathbf{B}' = \mathbf{B} \begin{pmatrix}
0 & 1 \\
1 & -q
\end{pmatrix}$ is obtained via a unimodular transformation, it is clearly a valid basis.  To simplify the following argument, note that the input of \textbf{Algorithm \ref{alg:TransPositive}} satisfies $\|\boldsymbol{a}\|\ge\|\boldsymbol{b}\|$ and $a_1a_2b_1b_2>0$. Without loss of generality, we assume \(a_1b_1 > 0\) and \(a_2b_2 > 0\), focusing on the case where \(|a_1| \geq |a_2|\); the argument for \(|a_1| < |a_2|\) follows analogously.  

\Tian{Before proving the properties of $\mathbf{B}'$,
we first establish the equivalence between the condition  
\(|a_2 - q_1 b_2| \geq |b_2| \land \operatorname{sgn}(a_2) = \operatorname{sgn}(a_2 - q_1 b_2)\) appearing in step~\ref{alg1:step3} and  \(q_1 < q_2\), where \( q_i = \lfloor \frac{a_i}{b_i} \rfloor_o \ge 0, i = 1,2 \).  
To see this, write $a_i = q_i b_i + r_i$ with $0\le |r_i| < |b_i|$ and $\operatorname{sgn}(r_i)=\operatorname{sgn}(a_i)$. Then $\operatorname{sgn}(q_1b_2)=\operatorname{sgn}(b_2)=\operatorname{sgn}(a_2)$. If $q_1<q_2$, then $q_1\le q_2-1$, so $|q_1b_2|\le|(q_2-1)b_2|<|q_2b_2|\le|a_2|$, implying $|a_2-q_1b_2|\ge|a_2-(q_2-1)b_2| = |r_2+b_2|\ge|b_2|$ and $\operatorname{sgn}(a_2-q_1b_2)=\operatorname{sgn}(a_2)$. Conversely, if the condition holds, then $|a_2-q_1b_2|\ge|b_2|>|r_2|=|a_2-q_2b_2|$ forces $|q_1b_2|<|q_2b_2|$, hence $q_1<q_2$.}

The above equivalence justifies the choice of $q$ in \textbf{Algorithm \ref{alg:TransPositive}}:
\begin{align}\label{eq:a1a2q}
q = \begin{cases}
q_1, & \text{if } q_1 \geq q_2,\\[.5em]
q_1+1, & \text{if } q_1 < q_2.
\end{cases}
\end{align}

Now, based on \eqref{eq:a1a2q}, we examine the properties of $\mathbf{B}'$ in the following three cases.

\emph{Case 1: $q_1 > q_2$.} Then $q=q_1$ and
\[
\mathbf{B}'=\begin{pmatrix}
a_1^\prime & b_1^\prime \\
a_2^\prime & b_2^\prime
\end{pmatrix}
:=
\begin{pmatrix}
b_1 & a_1 - q_1b_1 \\
b_2 & a_2 - q_1b_2
\end{pmatrix}.
\]
Because $b_1^\prime=a_1 - q_1b_1 = r_1$ has the same sign as $b_1$, we have $\operatorname{sgn}(a_1') = \operatorname{sgn}(b_1')$. For the second component, write $q_1 = q_2 + k$ with $k\ge 1$. Then $b_2^\prime=a_2 - q_1b_2 = r_2 - k b_2$ has opposite sign to $b_2$ (since $|r_2| < |b_2|$ and $k\ge 1$). Hence $\operatorname{sgn}(a_2') \neq \operatorname{sgn}(b_2')$, yielding $a_1'a_2'b_1'b_2' \le 0$. Moreover, 
\[
|a_2 - q_1b_2| = |r_2 - k b_2| \le k|b_2| \le q_1|b_2| \le \frac{a_1}{b_1}|b_2| = |a_1|\frac{|b_2|}{|b_1|} \le |a_1||b_2| \le \|\boldsymbol{a}\|\|\boldsymbol{b}\|,
\]
and $|a_1 - q_1b_1| = |r_1| < |b_1| \le \|\boldsymbol{b}\|$. Therefore $\|\boldsymbol{b}'\| \le \max\{\|\boldsymbol{b}\|, |a_1||b_2|\} \le \|\boldsymbol{a}\|\|\boldsymbol{b}\|$. Together with $\|\boldsymbol{a}'\| = \|\boldsymbol{b}\|$, we obtain the required bounds.

\emph{Case 2: $q_1 < q_2$.} Then $q=q_1+1$ and
\[
\mathbf{B}'=\begin{pmatrix}
a_1^\prime & b_1^\prime \\
a_2^\prime & b_2^\prime
\end{pmatrix}
:=
\begin{pmatrix}
b_1 & a_1 - (q_1 + 1)b_1 \\
b_2 & a_2 - (q_1 + 1)b_2
\end{pmatrix}.
\]
Here, $b_1^\prime=a_1 - (q_1+1)b_1 = r_1 - b_1$ has opposite sign to $b_1$, while for the second component, set $q_2 = q_1 + k$ ($k\ge 1$). Then $b_2^\prime=a_2 - (q_1+1)b_2 = r_2 + (k-1)b_2$ has the same sign as $b_2$. Consequently $a_1'a_2'b_1'b_2' \le 0$. Furthermore,
\[
|b_1^\prime|=|a_1 - (q_1+1)b_1| = |r_1 - b_1| \le |b_1| = |a_1'|, \
|b_2^\prime|=|a_2 - (q_1+1)b_2| = |r_2 + (k-1)b_2| < k|b_2| \le q_2|b_2| < |a_2|,
\]
because $|a_2| = q_2|b_2| + |r_2| > q_2|b_2|$. Hence $\|\boldsymbol{b}'\| \le \max\{|b_1|, |a_2|\} \le \max\{\|\boldsymbol{b}\|, \|\boldsymbol{a}\|\}$ and $\|\boldsymbol{a}'\| = \|\boldsymbol{b}\|$, which gives the desired estimates.

\emph{Case 3: $q_1 = q_2$.} Then $q=q_1$ and
\[
\mathbf{B}^\prime=\begin{pmatrix}
a_1^\prime & b_1^\prime \\
a_2^\prime & b_2^\prime
\end{pmatrix}
:=
\begin{pmatrix}
b_1 & a_1 - q_1b_1 \\
b_2 & a_2 - q_1b_2
\end{pmatrix}
=
\begin{pmatrix}
b_1 & r_1 \\
b_2 & r_2
\end{pmatrix}.
\]
If $r_1 = 0$ or $r_2 = 0$, then clearly $a_1'a_2'b_1'b_2' \le 0$ and $\|\boldsymbol{a}^\prime\| = \|\boldsymbol{b}\|$, $\|\boldsymbol{b}^\prime\| \le \|\boldsymbol{b}\|$. Otherwise, all signs are positive (or all negative) and $a_1'a_2'b_1'b_2' > 0$. In this situation, we have
\begin{align*}
&\left.\begin{array}{c}
    (|b_1^\prime| = |r_1|=|a_1-q_1b_1|\le \min\{|a_1-b_1|,|a_1+b_1|\}) \wedge (|b_1^\prime |=|r_1|< |b_1|=|a_1^\prime|)\wedge  (|b_1^\prime| =|r_1|< \frac{|a_1|}{2}),\\
   (|b_2^\prime| = |r_2|=|a_2-q_2b_2|\le \min\{|a_2-b_2|,|a_2+b_2|\})\wedge (|b_2^\prime |=|r_2|< |b_2|=|a_2^\prime|)\wedge (|b_2^\prime| = |r_2|< \frac{|a_2|}{2}).  
\end{array}\right\}\\
&\implies(\|\boldsymbol{b}^\prime\|\le \min\{\|\boldsymbol{a}-\boldsymbol{b}\|, \|\boldsymbol{a}+\boldsymbol{b}\|\}) \wedge(\|\boldsymbol{b}^\prime\|<\|\boldsymbol{a}^\prime\|) \wedge (\|\boldsymbol{b}^\prime\|<\frac{1}{2}\|\boldsymbol{a}\|).
\end{align*}
This completes the proof.
\end{proof}

\begin{lemma}\label{lem:key2}
 
\Tian{Let $(\boldsymbol{a}', \boldsymbol{b}', q)$ be the output of \textbf{Algorithm \ref{alg:TransNegative}}. Then the matrix \(\mathbf{B}' = [\boldsymbol{a}' \ \boldsymbol{b}']\) forms a valid basis and satisfies exactly one of the following two conditions:
\begin{enumerate}
\item \(\mathbf{B}^\prime\) is reduced, or  
\item \( a_1^\prime a_2^\prime b_1^\prime b_2^\prime \le 0 \), with \((\|\boldsymbol{b}^\prime\|\le \min\{\|\boldsymbol{a}-\boldsymbol{b}\|, \|\boldsymbol{a}+\boldsymbol{b}\|\}) \wedge(\|\boldsymbol{b}^\prime\|<\|\boldsymbol{a}^\prime\|) \wedge (\|\boldsymbol{b}^\prime\|<\frac{1}{2}\|\boldsymbol{a}\|)\wedge(||b_1^\prime|-|b_2^\prime|| < \frac{||a_1|-|a_2||}{2}) \).  
\end{enumerate}}
\end{lemma} 
\begin{proof}

Since
\[
\mathbf{B}' =
\begin{pmatrix}
a_1' & b_1' \\
a_2' & b_2'
\end{pmatrix}
= [\boldsymbol{b} \ \boldsymbol{a} - q\boldsymbol{b}] = \mathbf{B} \begin{pmatrix}
0 & 1 \\
1 & -q
\end{pmatrix},\] \(\mathbf{B}^\prime\) forms a basis as a direct consequence of the properties of unimodular matrix transformations. To simplify the following argument, we only focus on the case where \(|a_1|> |a_2|\); the argument for \(|a_1| < |a_2|\) follows analogously.

In case that \(|a_1|> |a_2|\), \(q= q_1 = \lfloor \frac{a_1}{b_1} \rfloor_o \) and \[
\mathbf{B}'=\begin{pmatrix}
a_1^\prime & b_1^\prime \\
a_2^\prime & b_2^\prime
\end{pmatrix}
:=
\begin{pmatrix}
b_1 & a_1 - q_1b_1 \\
b_2 & a_2 - q_1b_2
\end{pmatrix},
\] 
Note that the input of \textbf{Algorithm \ref{alg:TransNegative}} satisfies $\|\boldsymbol{a}\|\ge\|\boldsymbol{b}\|$, $a_1a_2b_1b_2\le 0$ and $(|a_1|-|a_2|)(|b_1|-|b_2|)>0$, we have $|a_1|>|a_2|\ge 0$, $|b_1|>|b_2|\ge 0$. Without loss of generality, we assume \(a_1b_1 > 0\) and \(a_2b_2 \le 0\). Then $q_1>0$. Let \( a_1- q_1b_1=r_1 \). 
Clearly $\operatorname{sgn}(r_1) = \operatorname{sgn}(b_1)$ indicates $\operatorname{sgn}(a_1')=\operatorname{sgn}(b_1')$ and, because $a_2$ and $b_2$ have opposite signs, $\operatorname{sgn}(a_2')\neq\operatorname{sgn}(b_2')$. Hence $a_1'a_2'b_1'b_2'\le0$.

Furthermore,  if $|b_1^\prime|-|b_2^\prime|\le 0$, then $(|a_1^\prime|-|a_2^\prime|)(|b_1^\prime|-|b_2^\prime|)=(|b_1|-|b_2|)(|b_1^\prime|-|b_2^\prime|)\le 0$ and thus $\mathbf{B}^\prime$ is reduced. Otherwise, $|b_2^\prime|<|b_1^\prime|=|r_1|$, which implies 
\begin{align*}
&\|\boldsymbol{b}^\prime\|=\max\{|b_1^\prime|, |b_2^\prime|\}=|b_1^\prime|=|r_1|=|a_1-q_1b_1|\le \min\{|a_1-b_1|,|a_1+b_1|\}\le \min\{\|\boldsymbol{a}-\boldsymbol{b}\|, \|\boldsymbol{a}+\boldsymbol{b}\|\},\\
 &\|\boldsymbol{b}^\prime\|=\max\{|b_1^\prime|, |b_2^\prime|\}=|b_1^\prime|=|r_1|<|b_1|\le \max\{|b_1|,|b_2|\}=\|\boldsymbol{b}\|=\|\boldsymbol{a}^\prime\|,\\
 &\|\boldsymbol{b}^\prime\|=\max\{|b_1^\prime|, |b_2^\prime|\}=|b_1^\prime|=|r_1|<\frac{|a_1|}{2}=\frac{1}{2}\max\{|a_1|,|a_2|\}=\frac{1}{2}\|\boldsymbol{a}\|,\\
&|b_1^\prime|-|b_2^\prime|=|r_1|-|a_2-q_1b_2|=|r_1|-(|a_2|+q_1|b_2|)\le|r_1|-|a_2|<\frac{|a_1|-|a_2|}{2}.
\end{align*}    
Thus all required inequalities hold. The case $|a_1|<|a_2|$ is symmetric and yields the same conclusion. This completes the proof.
\end{proof}

Based on \textbf{Lemma \ref{lem:key1}} and \textbf{Lemma \ref{lem:key2}}, we can establish the correctness and complexity of \textbf{Algorithm \ref{alg:CrossEuc}}. Specifically, the following theorem holds

\begin{theorem}\label{thm:Par-Euc}
  Given a lattice basis $\mathbf{B}=[\boldsymbol{a}\ \boldsymbol{b}]\in\mathbb{Z}^{2\times 2}$ with $\boldsymbol{a}=(a_1\ a_2)^T$, $\boldsymbol{b}=(b_1\ b_2)^T$, and $\#(\boldsymbol{a},\boldsymbol{b})=n$, \textbf{Algorithm \ref{alg:CrossEuc}}  outputs a basis $\mathbf{B}^\prime=[\boldsymbol{a}^\prime\ \boldsymbol{b}^\prime]$ satisfying $\|\boldsymbol{a}^\prime\|=\lambda_1(L(\mathbf{B}))$ and $\|\boldsymbol{b}^\prime\|=\lambda_2(L(\mathbf{B}))$ with a time complexity of $O(n^2)$.  
\end{theorem}

\begin{proof}
If the input lattice basis satisfies $a_1a_2b_1b_2>0$, we assume that the reduced lattice basis sequence in the first \textbf{While} loop of \textbf{Algorithm \ref{alg:CrossEuc}} is as follows:
\begin{align} 
\mathbf{B}^{(0)}=[\boldsymbol{a}^{(0)}\ \boldsymbol{b}^{(0)}]=\begin{pmatrix}
a_{1}^{(0)} & b_{1}^{(0)}\\
a_{2}^{(0)} & b_{2}^{(0)}
\end{pmatrix}=
\begin{pmatrix}
a_{1} & b_{1}\\
a_{2} & b_{2}
\end{pmatrix}\to \mathbf{B}^{(1)}=[\boldsymbol{a}^{(1)}\ \boldsymbol{b}^{(1)}]=\begin{pmatrix}
a_{1}^{(1)} & b_{1}^{(1)}\\
a_{2}^{(1)} & b_{2}^{(1)}
\end{pmatrix}\to \cdots\nonumber\\
\to
\mathbf{B}^{(k-1)}=[\boldsymbol{a}^{(k-1)}\ \boldsymbol{b}^{(k-1)}]=\begin{pmatrix}
a_{1}^{(k-1)} & b_{1}^{(k-1)}\\
a_{2}^{(k-1)} & b_{2}^{(k-1)}
\end{pmatrix}\to
\mathbf{B}^{(k)}=[\boldsymbol{a}^{(k)}\ \boldsymbol{b}^{(k)}]=\begin{pmatrix}
a_{1}^{(k)} & b_{1}^{(k)}\\
a_{2}^{(k)} & b_{2}^{(k)}
\end{pmatrix},\label{eq:latticeseq1}
\end{align}
where $\boldsymbol{a}^{(i)}=\boldsymbol{b}^{(i-1)}$ for $i=1,\cdots,k$, and, with loss of generality, we assume $k$ is even.
By \textbf{Lemma \ref{lem:key1}},  we have $a_1^{(k)}a_2^{(k)}b_1^{(k)}b_2^{(k)}\le 0$ and
\begin{align*}
&\|\boldsymbol{a}^{(k)}\|=\|\boldsymbol{b}^{(k-1)}\|<\frac{1}{2}\|\boldsymbol{a}^{(k-2)}\|<\cdots<\frac{1}{2^{k/2}}\|\boldsymbol{a}\|,\\
&\|\boldsymbol{a}^{(k-1)}\|<\frac{1}{2}\|\boldsymbol{a}^{(k-3)}\|<\cdots<\frac{1}{2^{(k-2)/2}}\|\boldsymbol{a}^{(1)}\|=\frac{1}{2^{(k-2)/2}}\|\boldsymbol{b}\|\\
&\|\boldsymbol{b}^{(k)}\|<\|\boldsymbol{a}^{(k-1)}\|\cdot\|\boldsymbol{b}^{(k-1)}\|<\frac{1}{2^{k/2}}\|\boldsymbol{a}\|\cdot \frac{1}{2^{(k-2)/2}}\|\boldsymbol{b}\|=\frac{1}{2^{k-1}}\|\boldsymbol{a}\|\cdot \|\boldsymbol{b}\|.
\end{align*}
Since $1\le \min\{\|\boldsymbol{a}^{(k)}\|,\|\boldsymbol{b}^{(k)}\|\}<\frac{1}{2^{k-1}}|\boldsymbol{a}\|\cdot \|\boldsymbol{b}\|$, it follows that $k<\log \|\boldsymbol{a}\|+\log \|\boldsymbol{b}\|+1$ is finite.

Further, if $(|a_1^{(k)}|-|a_2^{(k)}|)(|b_1^{(k)}|-|b_2^{(k)}|)>0$, then during the second \textbf{While} loop, we assume the reduced lattice basis sequence is 
\begin{align}\label{eq:latticeseq2}
\mathbf{B}^{(k)}=[\boldsymbol{a}^{(k)}\ \boldsymbol{b}^{(k)}]=\begin{pmatrix}
a_{1}^{(k)} & b_{1}^{(k)}\\
a_{2}^{(k)} & b_{2}^{(k)}
\end{pmatrix}\to \mathbf{B}^{(k+1)}=[\boldsymbol{a}^{(k+1)}\ \boldsymbol{b}^{(k+1)}]=\begin{pmatrix}
a_{1}^{(k+1)} & b_{1}^{(k+1)}\\
a_{2}^{(k+1)} & b_{2}^{(k+1)}
\end{pmatrix}\to \cdots\nonumber\\ \to
\mathbf{B}^{(k+s-1)}=[\boldsymbol{a}^{(k+s-1)}\ \boldsymbol{b}^{(k+s-1)}]=\begin{pmatrix}
a_{1}^{(k+s-1)} & b_{1}^{(k+s-1)}\\
a_{2}^{(k+s-1)} & b_{2}^{(k+s-1)}
\end{pmatrix}\to
\mathbf{B}^{(k+s)}=[\boldsymbol{a}^{(k+s)}\ \boldsymbol{b}^{(k+s)}]=\begin{pmatrix}
a_{1}^{(k+s)} & b_{1}^{(k+s)}\\
a_{2}^{(k+s)} & b_{2}^{(k+s)}
\end{pmatrix},
\end{align} where $\boldsymbol{a}^{(k+i)}=\boldsymbol{b}^{(k+i-1)}$ for $i=1,\cdots,s$, and, without loss of generality, we assume $s$ is even. Then, by \textbf{Lemma \ref{lem:key2}},  we have $a_1^{(k+i)}a_2^{(k+i)}b_1^{(k+i)}b_2^{(k+i)}\le 0$ for $i=0,\cdots,s$,
\begin{align*}
  1\le|a_1^{(k+s)}|-|a_2^{(k+s)}|= |b_1^{(k+(s-1))}|-|b_2^{(k+(s-1))}| <\frac{|a_1^{(k+(s-2))}|-|a_2^{(k+(s-2))}|}{2}<
   \cdots<\frac{|a_1^{(k)}|-|a_2^{(k)}|}{2^{s/2}},
\end{align*}
and $|b_1^{(k+s)}|-|b_2^{(k+s)}|\le 0.$
Consequently, after the second \textbf{While} loop, we have 
 \begin{align*}
 a_1^{(k+s)}a_2^{(k+s)}b_1^{(k+s)}b_2^{(k+s)}\le 0 \wedge (|a_1^{(k+s)}|-|a_2^{(k+s)}|)(|b_1^{(k+s)}|-|b_2^{(k+s)}|)\le 0,
\end{align*}
and thus, the lattice basis $\mathbf{B}^{(k+s)}=[\boldsymbol{a}^{(k+s)}\ \boldsymbol{b}^{(k+s)}]$ is reduced.
Finally, by \textbf{Theorem \ref{thm:Eucred}}, the returned basis $\mathbf{B}^\prime=[\boldsymbol{a}^\prime\ \boldsymbol{b}^\prime]$ of \textbf{Algorithm \ref{alg:CrossEuc}} satisfies $\|\boldsymbol{a}^\prime\|=\lambda_1(L(\mathbf{B}))$ and $\|\boldsymbol{b}^\prime\|=\lambda_2(L(\mathbf{B}))$

Now, we estimate the time complexity. Without loss of generality, we assume $\|\boldsymbol{a}^{(0)}\|\ge \|\boldsymbol{b}^{(0)}\|$, $\#\boldsymbol{a}^{(i)}=n_{i}$ and $\#q^{(i)}=\ell_{i}$ for $i=1,\cdots,k,k+1,\cdots, k+s$.  For the first \textbf{While} loop,  since $\boldsymbol{a}^{(i-1)}=q^{(i)}\boldsymbol{b}^{(i-1)}+\boldsymbol{b}^{(i)}=q^{(i)}\boldsymbol{a}^{(i)}+\boldsymbol{a}^{(i+1)}$, where 
\begin{align}\label{eq:q}
 q^{(i)}\in\left\{\left\lfloor\frac{a_1^{(i-1)}}{b_1^{(i-1)}}\right\rfloor_0=\left\lfloor\frac{a_1^{(i-1)}}{a_1^{(i)}}\right\rfloor_0, \left\lfloor\frac{a_1^{(i-1)}}{b_1^{(i-1)}}\right\rfloor_0+1=\left\lfloor\frac{a_1^{(i-1)}}{a_1^{(i)}}\right\rfloor_0+1,\right.\nonumber\\
 \left.\left\lfloor\frac{a_2^{(i-1)}}{b_2^{(i-1)}}\right\rfloor_0=\left\lfloor\frac{a_2^{(i-1)}}{a_2^{(i)}}\right\rfloor_0, \left\lfloor\frac{a_2^{(i-1)}}{b_2^{(i-1)}}\right\rfloor_0+1=\left\lfloor\frac{a_2^{(i-1)}}{a_2^{(i)}}\right\rfloor_0+1\right\},
\end{align} we have $n_{i-1}=n_i+\ell_i$. Therefore, in each loop, the time complexity of division is bounded by $O(n_{i-1}\ell_i)$ and the time complexity of multiplication is bounded by $O(n_i\ell_i)$. Consequently, the total complexity of the first \textbf{While} loop is $T_1=O(\sum_{i = 1}^{k}(n_{i-1}\ell_i+n_i\ell_i))=O(\sum_{i=1}^k (n_i+n_{i-1})(n_{i-1}-n_i))=O(\sum_{i=1}^k (n_{i-1}^2-n_i^2))=O(n_0^2-n_k^2)$. Similarly, for the second \textbf{While} loop, the total complexity is $T_2=O(\sum_{i=k+1}^{k+s}(n_{i-1}\ell_i+n_i\ell_i))=O(\sum_{i=k+1}^{k+s} (n_i+n_{i-1})(n_{i-1}-n_i))=O(\sum_{i=k+1}^{k+s} (n_{i-1}^2-n_i^2))=O(n_k^2-n_{k+s}^2)$.
 Therefore, the total complexity of  \textbf{Algorithm \ref{alg:CrossEuc}} is $T=O(T_1+T_2)=O(n_0^2)=O(n^2)$.
\end{proof}
 
\begin{remark}
According to the property established in \textbf{Theorem \ref{thm:l2}}, the $\ell_2$-shortest basis $[\boldsymbol{u}\ \boldsymbol{v}]$ can be obtained via a minor modification of \textbf{Algorithm \ref{alg:CrossEuc}}. Specifically, by replacing Steps \ref{crosseuc:step11}–\ref{crosseuc:step14} with the following:
$$
\boldsymbol{u} = \arg\min\left\{\|\boldsymbol{a}\|_2,\ \|\boldsymbol{b}\|_2,\ \|\boldsymbol{a} + \boldsymbol{b}\|_2,\ \|\boldsymbol{a} - \boldsymbol{b}\|_2\right\}, 
$$
and setting
$$
\boldsymbol{v} = \boldsymbol{x} - q\boldsymbol{u}, \quad \text{where } \boldsymbol{x} \in \{\boldsymbol{a}, \boldsymbol{b}\} \setminus \{\boldsymbol{u}\} \text{ and } q = \left\lceil \frac{\langle \boldsymbol{x}, \boldsymbol{u} \rangle}{\langle \boldsymbol{u}, \boldsymbol{u} \rangle} \right\rfloor.
$$
This substitution yields a basis satisfying the $\ell_2$-shortest condition. The same remark is also applicable to the upcoming \textbf{Algorithm \ref{alg:HGCD-Par-Euc2}} (\textbf{HVecSBP}).
\end{remark}

\section{Optimized algorithm and its analysis}
For two integers with bit lengths of \( n \), the well-known Half-GCD algorithm \cite{moller2006robust,moller2008schonhage} can efficiently find two integers of approximately \( n/2 \)-bit length that share the same common divisor.  
Inspired by the Half-GCD algorithm, this section first explores the \textbf{HVec} algorithm, which rapidly reduces a lattice basis to a new basis with bit lengths approximately halved from the original, along with a complexity analysis. Furthermore, we leverage this algorithm to optimize and accelerate the \textbf{CrossEuc} algorithm presented in the previous section.

\subsection{Design idea}

Analogous to the case of two integers, our design for the input of two vectors is based on the following key observation:

Assume $[\boldsymbol{a}^\prime\ \boldsymbol{b}^\prime]$ and $[\boldsymbol{c}^\prime\ \boldsymbol{d}^\prime]$ are two different bases for the same lattice, \emph{i.e.},
\begin{align}\label{eq:keyobser1}
[\boldsymbol{a}^\prime\ \boldsymbol{b}^\prime]=\begin{pmatrix}
a_{1}^\prime & b_{1}^\prime\\
a_{2}^\prime & b_{2}^\prime 
\end{pmatrix}=\begin{pmatrix}
c_{1}^\prime & d_{1}^\prime\\
c_{2}^\prime & d_{2}^\prime 
\end{pmatrix}\mathbf{M}^\prime=[\boldsymbol{c}^\prime\ \boldsymbol{d}^\prime] \mathbf{M}^\prime
\end{align}
for some unimodular matrix \(\mathbf{M}^\prime \). The two new vectors 
\begin{align}\label{eq:keyobser2}
[\boldsymbol{a}\ \boldsymbol{b}]=\begin{pmatrix}
a_1 & b_1\\
a_2 & b_2 
\end{pmatrix}=2^{n_\ell}\begin{pmatrix}
a_{1}^\prime & b_{1}^\prime\\
a_{2}^\prime & b_{2}^\prime 
\end{pmatrix}+\begin{pmatrix}
a_{1}^{\prime\prime} & b_{1}^{\prime\prime}\\
a_{2}^{\prime\prime} & b_{2}^{\prime\prime} 
\end{pmatrix}= 2^{n_\ell}[\boldsymbol{a}^\prime\ \boldsymbol{b}^\prime]+[\boldsymbol{a}^{\prime\prime}\ \boldsymbol{b}^{\prime\prime}],
\end{align}
where \( [\boldsymbol{a}^{\prime\prime}\ \boldsymbol{b}^{\prime\prime}] \) represents the \( n_\ell \) least significant bits, and \( [\boldsymbol{a}^\prime\ \boldsymbol{b}^\prime] \) refers to the \( n_h = \#(\boldsymbol{a},\boldsymbol{b}) - n_\ell \) most significant bits of \( [\boldsymbol{a}\ \boldsymbol{b}] \).
Additionally, the vectors \( [\boldsymbol{c}\ \boldsymbol{d}] \) satisfy
\begin{align}
  [\boldsymbol{a}\ \boldsymbol{b}]=
\begin{pmatrix}
a_1 & b_1\\
a_2 & b_2 
\end{pmatrix}=\begin{pmatrix}
c_{1} & d_{1}\\
c_{2} & d_{2} 
\end{pmatrix}\mathbf{M}^\prime=[\boldsymbol{c}\ \boldsymbol{d}]\mathbf{M}^\prime   
\end{align}
for the same unimodular matrix \( \mathbf{M}^\prime \). Then, we have
\begin{align}\label{eq:keyobser3}
 [\boldsymbol{c}\ \boldsymbol{d}]=[\boldsymbol{a}\ \boldsymbol{b}](\mathbf{M}^\prime)^{-1}=2^{n_\ell}[\boldsymbol{c}^\prime\ \boldsymbol{d}^\prime]+[\boldsymbol{a}^{\prime\prime}\ \boldsymbol{b}^{\prime\prime}](\mathbf{M}^\prime)^{-1} 
\end{align}

\Tian{Building on the above observation, we can halve the bit-length of \([\boldsymbol{a}\ \boldsymbol{b}]\) by employing a recursive reduction strategy.  Informally, in equation (\ref{eq:keyobser3}), we aim to make \([\boldsymbol{c}\ \boldsymbol{d}]\) as short as possible. However, by equation (\ref{eq:keyobser1}), there is a trade‑off between the sizes of \([\boldsymbol{c}'\ \boldsymbol{d}']\) and \(\mathbf{M}'\). Therefore, to minimize the size of \([\boldsymbol{c}\ \boldsymbol{d}]\) in (\ref{eq:keyobser3}), we require that the term \(2^{n_\ell}[\boldsymbol{c}'\ \boldsymbol{d}']\) and the low‑dimensional part \([\boldsymbol{a}''\ \boldsymbol{b}''](\mathbf{M}')^{-1}\) are of comparable magnitude. That is, we need \(n_\ell + t \approx n_\ell + (n_h - t)\), where \(\#(\boldsymbol{c}',\boldsymbol{d}') = t\), which implies \(t \approx n_h/2\).}

\Tian{The key to our algorithm design lies in precisely controlling this approximation. This requires that the vectors \([\boldsymbol{c}'\ \boldsymbol{d}']\) obtained from the high‑bit reduction are neither too large nor too small. For a unimodular matrix with non‑negative entries (or with controlled signs), a standard bound states that if the output \([\boldsymbol{c}', \boldsymbol{d}']\) has minimal size greater than \(t\), then the entries of \(\mathbf{M}'\) (and hence \((\mathbf{M}')^{-1}\)) are bounded by approximately \(n_h - t\), where \(n_h\) is the size of the input to that recursive call. If we only required the \emph{maximum} size of \(\boldsymbol{c}',\boldsymbol{d}'\) to exceed \(t\) while their \emph{minimum} size may be small, then the matrix \(\mathbf{M}'\) could become large. A large matrix would amplify the low‑part vectors \(\boldsymbol{a}'',\boldsymbol{b}''\) upon multiplication, potentially making \([\boldsymbol{a}''\ \boldsymbol{b}''](\mathbf{M}')^{-1}\) as large as or even larger than the dominant term \(2^{n_\ell}[\boldsymbol{c}'\ \boldsymbol{d}']\). In that case, we would lose control over the size of the adjusted \([\boldsymbol{c}\ \boldsymbol{d}]\).}

\Tian{Therefore, to avoid  backup steps, we adopt a conservative strategy that requires the \emph{minimum} size of \(\boldsymbol{c}',\boldsymbol{d}'\) is larger than \(n_h/2\), while also ensuring that they are as small as possible—so that the bit‑length decreases after at most one additional reduction. This is achieved by controlling \(\underline{\#}(\boldsymbol{c}'+\boldsymbol{d}',\boldsymbol{c}'-\boldsymbol{d}')\le n_h/2\). Thus, requiring \(\underline{\#}(\boldsymbol{c}',\boldsymbol{d}') > \lfloor n_h/2\rfloor + 1\) (where \(n_h \approx n/2\)) ensures that \(\mathbf{M}'\) is small enough so that the low‑part term remains dominated by the shifted high‑part term. The condition \(\underline{\#}(\boldsymbol{c}'+\boldsymbol{d}',\boldsymbol{c}'-\boldsymbol{d}')\le \lfloor n_h/2\rfloor + 1\) means that after the recursive call, the sum and difference of the two vectors are already “small” (their minimal component size is at most about half of \(n_h\)).}


\subsection{\textbf{HVec} algorithm description and  analysis}

\textbf{Algorithm \ref{alg:HVec}} provides the details of our design. We now provide a rigorous theoretical analysis to establish the correctness and complexity of \textbf{Algorithm \ref{alg:HVec}}. That is, we argue the following result:

\begin{algorithm}
\caption{\textbf{HVec}($\boldsymbol{a}\ \boldsymbol{b}$)} \label{alg:HVec}
\begin{algorithmic}[1]
\REQUIRE  Two vectors $\boldsymbol{a}=(a_1,a_2)^{T}, \boldsymbol{b}=(b_1,b_2)^{T}$ and $\#(\boldsymbol{a},\boldsymbol{b})=n,\underline{\#}(\boldsymbol{a},\boldsymbol{b})>\left\lfloor\frac{n}{2}\right\rfloor+1, \underline{\#}(\boldsymbol{a}+\boldsymbol{b},\boldsymbol{a}-\boldsymbol{b})>\lfloor\frac{n}{2}\rfloor+1$
\ENSURE  Two vectors $\boldsymbol{c}=(c_1,c_2)^{T},\boldsymbol{d}=(d_1,d_2)^{T}$ and an unimodular matrix $\mathbf{M}$ such that $[\boldsymbol{a}\ \boldsymbol{b}]=[\boldsymbol{c}\ \boldsymbol{d}]\mathbf{M}$. Meanwhile, either $(\#(\boldsymbol{c},\boldsymbol{d})\le n) \wedge (\underline{\#}(\boldsymbol{c},\boldsymbol{d})> \lfloor\frac{n}{2}\rfloor+1) \wedge (\underline{\#}(\boldsymbol{c}+\boldsymbol{d},\boldsymbol{c}-\boldsymbol{d})\leq\lfloor\frac{n}{2}\rfloor+1)$,   or $[\boldsymbol{c}\ \boldsymbol{d}]$ is reduced.
\STATE  $\boldsymbol{c}:=\boldsymbol{a},\boldsymbol{d}:=\boldsymbol{b}$, $\mathbf{M}:=\mathbf{I}$
\IF{$c_1c_2d_1d_2\leq0 \wedge (|c_1|-|c_2|)(|d_1|-|d_2|)\leq 0$}
\RETURN $\boldsymbol{c},\boldsymbol{d},\mathbf{M}$
\ENDIF
\STATE $n:=\#(\boldsymbol{c},\boldsymbol{d})$
\STATE $s:=\lfloor\frac{n}{2}\rfloor+1$
\STATE  \textbf{if} $\underline{\#}(\boldsymbol{c},\boldsymbol{d})\geq \lfloor\frac{3}{4}n\rfloor+2$, \textbf{then} \label{Hvec:stepif1start}
\STATE \quad $n_{\ell 1}:=\lfloor n/2\rfloor$
\STATE \quad $a_{1}^\prime:=\lfloor c_1/2^{n_{\ell 1}}\rfloor $, $a_{1}^{\prime\prime}:=c_1-2^{n_{\ell 1}}a_{1}^\prime$, $a_{2}^\prime:=\lfloor c_2/2^{n_{\ell 1}}\rfloor$, $a_{2}^{\prime\prime}:=c_2-2^{n_{\ell 1}}a_{2}^\prime$
\STATE \quad $b_{1}^\prime:=\lfloor d_1/2^{n_{\ell 1}}\rfloor $, $b_{1}^{\prime\prime}:=d_1-2^{n_{\ell 1}}b_{1}^\prime$, $b_{2}^\prime:=\lfloor d_2/2^{n_{\ell 1}}\rfloor$, $b_{2}^{\prime\prime}:=d_2-2^{n_{\ell 1}}b_{2}^{\prime}$
\STATE \quad $\boldsymbol{a}^{\prime}:=(a_{1}^\prime,a_{2}^\prime)^{T},\boldsymbol{b}^\prime :=(b_{1}^\prime,b_{2}^\prime)^{T}$, $\boldsymbol{a}^{\prime\prime}:=(a_{1}^{\prime\prime},a_{2}^{\prime\prime})^{T},\boldsymbol{b}^{\prime\prime} :=(b_{1}^{\prime\prime},b_{2}^{\prime\prime})^{T}$
\STATE \quad $(\boldsymbol{c}^\prime,\boldsymbol{d}^\prime,\mathbf{M}_1^\prime)\leftarrow \textbf{HVec}(\boldsymbol{a}^\prime, \boldsymbol{b}^\prime)$
\STATE \quad $[\boldsymbol{c}\ \boldsymbol{d}]:= 2^{n_{\ell 1}}[\boldsymbol{c}^\prime\  \boldsymbol{d}^\prime]+[\boldsymbol{a}^{\prime\prime}\ \boldsymbol{b}^{\prime\prime}](\mathbf{M}_1^\prime)^{-1}$\label{Hvec:stepcd1}
\STATE \quad $\mathbf{M}:=\mathbf{M}_1^\prime\mathbf{M}$
\STATE \textbf{end if}\label{Hvec:stepif1end}
\STATE \textbf{while}    $\#(\boldsymbol{c},\boldsymbol{d})>\lfloor 3n/4\rfloor+1$ and $\underline{\#}(\boldsymbol{c}-\boldsymbol{d},\boldsymbol{c}+\boldsymbol{d})>s$ \label{hvec:stepwhile1} \textbf{do}
\STATE \quad \textbf{if} ($c_1c_2d_1d_2\leq0 \wedge (|c_1|-|c_2|)(|d_1|-|d_2|)\leq 0$), \textbf{return} $\boldsymbol{c},\boldsymbol{d},\mathbf{M}$
\STATE \quad \textbf{if} ($c_1c_2d_1d_2>0$),
\STATE \qquad $(\boldsymbol{c}, \boldsymbol{d}, q)\leftarrow \textbf{UMTrans1}(\boldsymbol{c},\boldsymbol{d})$ 
\STATE \quad \textbf{else}
\STATE \qquad $(\boldsymbol{c}, \boldsymbol{d}, q)\leftarrow \textbf{UMTrans2}(\boldsymbol{c},\boldsymbol{d})$ 
\STATE \quad \textbf{if} ($\#\boldsymbol{d}\leq s$), $q := q-1, \boldsymbol{d} := \boldsymbol{c}+\boldsymbol{d}$\label{hvec:stephalf1}
\STATE \quad $\mathbf{M} := \begin{pmatrix}q&1\\ 1&0\end{pmatrix}\mathbf{M}$
\STATE \textbf{end while}
\STATE \textbf{if} $\underline{\#}(\boldsymbol{c},\boldsymbol{d})>s+2$, \textbf{then}\label{Hvec:stepif2start}
\STATE \quad  $n^\prime:=\#(\boldsymbol{c},\boldsymbol{d})$, $n_{\ell 2}:=2s-n^\prime+1$
\STATE \quad $a_{1}^\prime:=\lfloor c_1/2^{n_{\ell 2}}\rfloor$, $a_{1}^{\prime\prime}:=c_1-2^{n_{\ell 2}}a_{1}^\prime$, $a_{2}^\prime:=\lfloor c_2/2^{n_{\ell 2}}\rfloor$, $a_{2}^{\prime\prime}:=c_2-2^{n_{\ell 2}}a_{2}^\prime$
\STATE \quad $b_{1}^\prime:=\lfloor d_1/2^{n_{\ell 2}}\rfloor $, $b_{1}^{\prime\prime}:=d_1-2^{n_{\ell 2}}b_{1}^\prime$, $b_{2}^\prime:=\lfloor d_2/2^{n_{\ell 2}}\rfloor$, $b_{2}^{\prime\prime}:=d_2-2^{n_{\ell 2}}b_{2}^\prime$
\STATE \quad $\boldsymbol{a}^\prime :=(a_{1}^\prime,a_{2}^\prime)^{T},\boldsymbol{b}^\prime :=(b_{1}^\prime,b_{2}^\prime)^{T}$, $\boldsymbol{a}^{\prime\prime}:=(a_{1}^{\prime\prime},a_{2}^{\prime\prime})^{T},\boldsymbol{b}^{\prime\prime} :=(b_{1}^{\prime\prime},b_{2}^{\prime\prime})^{T}$
\STATE \quad $(\boldsymbol{c}^\prime,\boldsymbol{d}^\prime,\mathbf{M}_2^\prime)\leftarrow \textbf{HVec}(\boldsymbol{a}^\prime, \boldsymbol{b}^\prime)$
\STATE \quad $[\boldsymbol{c}\ \boldsymbol{d}]:= 2^{n_{\ell 2}}[\boldsymbol{c}^\prime\  \boldsymbol{d}^\prime]+[\boldsymbol{a}^{\prime\prime}\ \boldsymbol{b}^{\prime\prime}] (\mathbf{M}_2^\prime)^{-1}$
\STATE \quad $\mathbf{M}:=\mathbf{M}_2^\prime\mathbf{M}$
\STATE \textbf{end if}\label{Hvec:stepif2end}
\STATE \textbf{while}   $\underline{\#}(\boldsymbol{c}-\boldsymbol{d}, \boldsymbol{c}+\boldsymbol{d})>s$ \textbf{do}\label{hvec:stepwhile2}
\STATE \quad \textbf{if} ($c_1c_2d_1d_2\leq0 \wedge (|c_1|-|c_2|)(|d_1|-|d_2|)\leq 0$), \textbf{return} $\boldsymbol{a},\boldsymbol{b},\mathbf{M}$
\STATE \quad \textbf{if} ($c_1c_2d_1d_2>0$),
\STATE \qquad $(\boldsymbol{c}, \boldsymbol{d}, q)\leftarrow \textbf{UMTrans1}(\boldsymbol{c},\boldsymbol{d})$ 
\STATE \quad \textbf{else}
\STATE \qquad $(\boldsymbol{c}, \boldsymbol{d}, q)\leftarrow \textbf{UMTrans2}(\boldsymbol{c},\boldsymbol{d})$ 
\STATE \quad \textbf{if} ($\#\boldsymbol{d}\leq s$), $q := q-1, \boldsymbol{d} := \boldsymbol{c}+\boldsymbol{d}$\label{hvec:stephalf2}
\STATE \quad $\mathbf{M} := \begin{pmatrix}q&1\\ 1&0\end{pmatrix}\mathbf{M}$
\STATE \textbf{end while}
\STATE \textbf{return} $\boldsymbol{c},\boldsymbol{d},\mathbf{M}$\\
\end{algorithmic}
\end{algorithm}

\begin{theorem} \label{thm:HVec}
For any input basis $\mathbf{B}=[\boldsymbol{a}\ \boldsymbol{b}]
\in\mathbb{Z}^{2\times 2}$ with $\#(\boldsymbol{a},\boldsymbol{b})=n$, \textbf{Algorithm \ref{alg:HVec}} will output a new basis $\mathbf{B}^\prime=[\boldsymbol{c}\ \boldsymbol{d}]
$ 
and an unimodular matrix $\mathbf{M}$ such that $[\boldsymbol{a}\ \boldsymbol{b}]=[\boldsymbol{c}\ \boldsymbol{d}]\mathbf{M}$ in time $T(n)=O\left(M(n)\log n\right)$, 
where $M(n)$ refers to the time complexity of multiplying two  $n$-bit integers. Additionally, the basis $[\boldsymbol{c}\ \boldsymbol{d}]$ satisfies one of the following two conditions:
\begin{enumerate}
  \item $[\boldsymbol{c},\boldsymbol{d}]$ is reduced, or
  \item $\underline{\#}(\boldsymbol{c},\boldsymbol{d})> \lfloor\frac{n}{2}\rfloor+1$, and $\underline{\#}(\boldsymbol{c}-\boldsymbol{d}, \boldsymbol{c}+\boldsymbol{d})\leq \lfloor\frac{n}{2}\rfloor+1$.   
\end{enumerate}
\end{theorem}
\begin{proof}
First, we argue that the algorithm will inevitably terminate. Assume 
\[
(\boldsymbol{c}^{new},\ \boldsymbol{d}^{new},q)\gets \textbf{UMTrans1}(\boldsymbol{c}\ \boldsymbol{d}) (\textrm{resp.}\ \textbf{UMTrans2}(\boldsymbol{c}\ \boldsymbol{d})).
\]
Then, by \textbf{Lemma \ref{lem:key1}} and \textbf{Lemma \ref{lem:key2}}, if $[\boldsymbol{c}^{new},\ \boldsymbol{d}^{new}]$ is not reduced, the following inequalities hold except for a special case where  $c_1c_2d_1d_2>0$ and $c_1^{new}c_2^{new}d_1^{new}d_2^{new}\leq 0$, which occurs only once during execution and can be ignored:
\[(\boldsymbol{c}^{new}=\boldsymbol{d})\wedge(\|\boldsymbol{d}^{new}\|\le\min\{\|\boldsymbol{c}-\boldsymbol{d}\|,\|\boldsymbol{c}+\boldsymbol{d}\|\})\wedge(\|\boldsymbol{d}^{new}\|< \|\boldsymbol{c}^{new}\|)\wedge (\|\boldsymbol{d}^{new}\|<\frac{1}{2}\|\boldsymbol{c}\|),\] 
which implies that 
\begin{align}\label{eq:dnew}
  \#(\boldsymbol{c}^{(new)},\boldsymbol{d}^{(new)})<\#(\boldsymbol{c},\boldsymbol{d}), \ \#\boldsymbol{d}^{(new)}\le \underline{\#}(\boldsymbol{c}+\boldsymbol{d},\boldsymbol{c}-\boldsymbol{d}),\  \#\boldsymbol{d}^{(new)}<\#\boldsymbol{c}-1.
\end{align}  Meanwhile, steps \ref{hvec:stephalf1} and \ref{hvec:stephalf2} ensure that $\underline{\#}(\boldsymbol{c}^{(new)},\boldsymbol{d}^{(new)})>s=\lfloor\frac{n}{2}\rfloor+1$ 
Consequently, the \textbf{While} conditions in Step \ref{hvec:stepwhile1} and Step \ref{hvec:stepwhile2} will eventually be violated. That is, \textbf{Algorithm \ref{alg:HVec}} will terminate. Additionally,  by equations (\ref{eq:keyobser1})-(\ref{eq:keyobser3}) and the properties of the unimodular transformations \textbf{UMTrans1} and \textbf{UMTrans2}, the output $(\boldsymbol{c}, \boldsymbol{d},\mathbf{M})$ satisfies 
$[\boldsymbol{a}\ \boldsymbol{b}]=[\boldsymbol{c}\ \boldsymbol{d}]\mathbf{M}$ and one of the following conditions holds: 

(1) $[\boldsymbol{c}\ \boldsymbol{d}]$ is reduced, or

(2) $\underline{\#}(\boldsymbol{c},\boldsymbol{d})>\lfloor\frac{n}{2}\rfloor+1$ and $\underline{\#}(\boldsymbol{c}-\boldsymbol{d}, \boldsymbol{c}+\boldsymbol{d})\leq \lfloor\frac{n}{2}\rfloor+1$.

Now, we estimate the time complexity.  We need to estimate the bit size of the numbers and the time complexity of each \textbf{while} loop during execution. According to our recursive design, the underlying operations are the unimodular transformations \textbf{UMTrans1} and \textbf{UMTrans2}. For each recursive invocation
$(\boldsymbol{c}^\prime, \boldsymbol{d}^\prime, \mathbf{M}^\prime)\gets\textbf{HVec}(\boldsymbol{a}^\prime\ \boldsymbol{b}^\prime)$, we adopt the same notations as in the proof of \textbf{Theorem \ref{thm:Par-Euc}} for brevity. Assume that $\|\boldsymbol{a}^\prime\|\ge\|\boldsymbol{b}^\prime\|$ and the lattice basis sequence is given by 
\begin{align*}
    \mathbf{B}^{(0)}=[\boldsymbol{a}^{(0)}\ \boldsymbol{b}^{(0)}]=[\boldsymbol{a}^\prime\ \boldsymbol{b}^\prime]\to \mathbf{B}^{(1)}=[\boldsymbol{a}^{(1)}\ \boldsymbol{b}^{(1)}]\to \cdots\to \mathbf{B}^{(m)}=[\boldsymbol{a}^{(m)}\ \boldsymbol{b}^{(m)}]=[\boldsymbol{c}^\prime\ \boldsymbol{d}^\prime],
\end{align*}
where, for $i=1,\cdots,m$, 
\begin{align} \label{eq:Eucdiv}
\boldsymbol{a}^{(i)}=\boldsymbol{b}^{(i-1)}, \boldsymbol{a}^{(i-1)}=q^{(i)}\boldsymbol{b}^{(i-1)}+\boldsymbol{b}^{(i)}=q^{(i)}\boldsymbol{a}^{(i)}+\boldsymbol{a}^{(i+1)}
\end{align} and $q^{(i)}$ is defined in equation (\ref{eq:q}). Then $[\boldsymbol{a}^\prime\ \boldsymbol{b}^\prime]=[\boldsymbol{c}^\prime\ \boldsymbol{d}^\prime]\mathbf{M}^\prime$ with
\begin{align}\label{eq:M}
    \mathbf{M}^\prime=\begin{pmatrix}q^{(m)} & 1\\ 1 & 0\end{pmatrix}\begin{pmatrix}q^{(m-1)} & 1\\ 1 & 0\end{pmatrix}\cdot\cdots\cdot\begin{pmatrix}q^{(1)} & 1\\ 1 & 0\end{pmatrix}.
\end{align} 
Furthermore, by \textbf{Lemma \ref{lem:key1}} and \textbf{Lemma \ref{lem:key2}}, the following inequalities hold for almost all \( i \in \{1, \dots, m\} \), with at most one exception, which does not affect the proof and can be ignored:
\[
\|\boldsymbol{a}^{(i+1)}\| = \|\boldsymbol{b}^{(i)}\| < \frac{1}{2} \|\boldsymbol{a}^{(i-1)}\| \quad \text{and} \quad \|\boldsymbol{b}^{(i)}\| < \|\boldsymbol{a}^{(i)}\|.
\]
Then, the equation (\ref{eq:Eucdiv}) indicate that $\#\boldsymbol{a}^{(i-1)}=\# q^{(i)}+\#\boldsymbol{a}^{(i)}$.
While, the equation (\ref{eq:M}) implies that $\#\mathbf{M}^\prime<\sum_{i=1}^m\# q^{(i)}+1$.
Hence, \begin{align}\label{eq:bit}
    &\#\boldsymbol{a}^{(0)}=\# q^{(1)}+\#\boldsymbol{a}^{(1)}=\sum_{i=1}^m\# q^{(i)}+\# \boldsymbol{a}^{(i)}>\#\mathbf{M}^\prime-1+\#\boldsymbol{c}^\prime\nonumber\\
    &\implies \#\mathbf{M}^\prime<\#\boldsymbol{a}^{(0)}-\#\boldsymbol{c}^\prime+1\le\#(\boldsymbol{a}^\prime,\boldsymbol{b}^\prime)-\underline{\#}(\boldsymbol{c}^\prime,\boldsymbol{d}^\prime)+1.
\end{align}
Based on the above observation, in step \ref{Hvec:stepif1start}-step \ref{Hvec:stepif1end},
we have 
\begin{align*}
&\#(\boldsymbol{c}^\prime,\boldsymbol{d}^\prime)\le n-n_{\ell 1} \wedge \underline{\#}(\boldsymbol{c}^\prime,\boldsymbol{d}^\prime)>\lfloor\frac{n-n_{\ell 1}}{2}\rfloor+1 \wedge \underline{\#}(\boldsymbol{c}^\prime+\boldsymbol{d}^\prime,\boldsymbol{c}^\prime-\boldsymbol{d}^\prime)\le \lfloor\frac{n-n_{\ell 1}}{2}\rfloor+1 \wedge\\
& \#\mathbf{M}_1^\prime\le\#(\boldsymbol{a}^\prime,\boldsymbol{b}^\prime)-\underline{\#}(\boldsymbol{c}^\prime,\boldsymbol{d}^\prime)+1<n-n_{\ell 1}-(\lfloor\frac{n-n_{\ell 1}}{2}\rfloor+1)+1.
\end{align*}
Thus, the equation $[\boldsymbol{c}\ \boldsymbol{d}]= 2^{n_{\ell 1}}[\boldsymbol{c}^\prime\  \boldsymbol{d}^\prime]+[\boldsymbol{a}^{\prime\prime}\ \boldsymbol{b}^{\prime\prime}](\mathbf{M}_1^\prime)^{-1}$ in step \ref{Hvec:stepcd1} implies that
\begin{align*}
  &\#(\boldsymbol{c},\boldsymbol{d})\le \max\{n_{\ell 1}+\#(\boldsymbol{c}^\prime,\boldsymbol{d}^\prime), n_{\ell 1}+\#\mathbf{M}_1^\prime\}+1=n+1,\\
 &\underline{\#}(\boldsymbol{c}\pm\boldsymbol{d})=\underline{\#}(2^{n_{\ell 1}}(\boldsymbol{c}^\prime+\boldsymbol{d}^\prime)+(\boldsymbol{a}^{\prime\prime}+\boldsymbol{b}^{\prime\prime})(\mathbf{M}_1^\prime)^{-1}, 2^{n_{\ell 1}}(\boldsymbol{c}^\prime-\boldsymbol{d}^\prime)+(\boldsymbol{a}^{\prime\prime}-\boldsymbol{b}^{\prime\prime})(\mathbf{M}_1^\prime)^{-1})\\
 &=\min\{\# (2^{n_{\ell 1}}(\boldsymbol{c}^\prime+\boldsymbol{d}^\prime)+(\boldsymbol{a}^{\prime\prime}+\boldsymbol{b}^{\prime\prime})(\mathbf{M}_1^\prime)^{-1}), \#(2^{n_{\ell 1}}(\boldsymbol{c}^\prime-\boldsymbol{d}^\prime)+(\boldsymbol{a}^{\prime\prime}-\boldsymbol{b}^{\prime\prime})(\mathbf{M}_1^\prime)^{-1}))\}\\
 &=\min\{\max\{n_{\ell 1}+\#(\boldsymbol{c}^\prime+\boldsymbol{d}^\prime),\#(\boldsymbol{a}^{\prime\prime}+\boldsymbol{b}^{\prime\prime})+\#\mathbf{M}_1^\prime\}+1, \max\{n_{\ell 1}+\#(\boldsymbol{c}^\prime-\boldsymbol{d}^\prime),\#(\boldsymbol{a}^{\prime\prime}-\boldsymbol{b}^{\prime\prime})+\#\mathbf{M}_1^\prime\}+1\}\\
 &\le \max\{n_{\ell 1}+\underline{\#}(\boldsymbol{c}^\prime\pm \boldsymbol{d}^\prime), \max\{\#\boldsymbol{a}^{\prime\prime},\#\boldsymbol{b}^{\prime\prime}\}+1+\#\mathbf{M}_1^\prime\}+1\\
 &\le \max\{n_{\ell 1}+\lfloor\frac{n-n_{\ell 1}}{2}\rfloor+1, n_{\ell 1}+1+n-n_{\ell 1}-(\lfloor\frac{n-n_{\ell 1}}{2}\rfloor+1)+1\}\\
 &=\max\{n_{\ell 1}+\lfloor\frac{n_{\ell 1}}{2}\rfloor+1, n-\lfloor\frac{n_{\ell 1}}{2}\rfloor+1\}<\frac{3n}{4}+3<\lfloor\frac{3n}{4}\rfloor+4.
\end{align*}
By the equation (\ref{eq:dnew}), the bit length of the new vector generated by the unimodular matrix transformation \textbf{UMTrans1} (resp. \textbf{UMTrans2}) is less that $\min\{\underline{\#}(\boldsymbol{c}+\boldsymbol{d},\boldsymbol{c}-\boldsymbol{d}), \#\boldsymbol{c}-1\}$. Therefore, the \textbf{while} loop in step \ref{hvec:stepwhile1}  will terminate after at most $8$ iterations and thus the time complexity is $O(M(n))$. Similarly, in step \ref{Hvec:stepif2start}-step \ref{Hvec:stepif2end}, we have 
\begin{align*}
&\#\mathbf{M}_2^\prime \le\#(\boldsymbol{a}^\prime,\boldsymbol{b}^\prime)-\underline{\#}(\boldsymbol{c}^\prime,\boldsymbol{d}^\prime)+1<n^\prime-n_{\ell 2}-(\lfloor\frac{n^\prime-n_{\ell 2}}{2}\rfloor+1)+1= n^\prime-s,\\
  &\#(\boldsymbol{c},\boldsymbol{d})\le \max\{n_{\ell 2}+\#(\boldsymbol{c}^\prime,\boldsymbol{d}^\prime), n_{\ell 2}+\#\mathbf{M}_2^\prime\}+1=n^\prime+1,\\
 &\underline{\#}(\boldsymbol{c}\pm\boldsymbol{d})=\underline{\#}(2^{n_{\ell 2}}(\boldsymbol{c}^\prime+\boldsymbol{d}^\prime)+(\boldsymbol{a}^{\prime\prime}+\boldsymbol{b}^{\prime\prime})(\mathbf{M}_2^\prime)^{-1}, 2^{n_{\ell 2}}(\boldsymbol{c}^\prime-\boldsymbol{d}^\prime)+(\boldsymbol{a}^{\prime\prime}-\boldsymbol{b}^{\prime\prime})(\mathbf{M}_2^\prime)^{-1})<s+2.
\end{align*}
By the equation (\ref{eq:dnew}), the bit length of the new vector generated by the unimodular matrix transformation \textbf{UMTrans1} (resp. \textbf{UMTrans2}) is less that $\min\{\underline{\#}(\boldsymbol{c}+\boldsymbol{d},\boldsymbol{c}-\boldsymbol{d}), \#\boldsymbol{c}-1\}$. Therefore, the \textbf{while} loop in step \ref{hvec:stepwhile2}  will terminate after at most $4$ iterations  and thus the time complexity is $O(M(n))$. Overall, the total time complexity is 
\begin{align*}
    T(n)=2T\left(\frac{n}{2}\right)+O(M(n))=\left(2^2 T\left(\frac{n}{2^2}\right)+O\left(2\cdot M\left(
    \frac{n}{2}\right)\right)\right)+O(M(n))=O\left(\sum_{i=0}^{\log n}2^i M\left(\frac{n}{2^i}\right)\right).
\end{align*}
Since $M\left(\frac{n}{2^i}\right)\le \frac{1}{2^i}M\left(n\right)$ for $i=0,\cdots,\log n$, 
\begin{align*}
T(n)=O\left(\sum_{i=0}^{\log n}2^i M\left(\frac{n}{2^i}\right)\right)=O\left(M(n)\sum_{i=0}^{\log n}1\right)=O(M(n)\log n).
\end{align*}
\end{proof}

\subsection{\textbf{HVecSBP} algorithm and its complexity analysis}

By cyclically invoking the proposed \textbf{HVec} algorithm, the reduced basis can ultimately be obtained. This design, referred to as \textbf{HVecSBP}, is detailed in \textbf{Algorithm \ref{alg:HGCD-Par-Euc2}}. Clearly, by \textbf{Theorem \ref{thm:HVec}}, the \textbf{HVecSBP} algorithm will output a reduced basis with a worst-case time complexity of \begin{align*}
  O\left(T(n)+T\left(\frac{n}{2}\right)+\cdots+T\left(\frac{n}{2^{\log n}}\right) \right)=O\left(M(n)\log n +M\left(\frac{n}{2}\right) \log \frac{n}{2}+\cdots+ M\left(\frac{n}{2^{\log n}}\right) \log \frac{n}{2^{\log n}}\right)\\
  =O\left(M(n)\log n +\frac{1}{2}M\left(n\right) \log n+\cdots+ \frac{1}{2^{\log n}}M\left(n\right) \log n\right)=O(M(n)\log n)  
\end{align*}

\begin{algorithm}
\caption{\textbf{HVecSBP}($\boldsymbol{a}, \boldsymbol{b}$)}   \label{alg:HGCD-Par-Euc2}
\begin{algorithmic}[1]
\REQUIRE  A basis $[\boldsymbol{a}\ \boldsymbol{b}]$ with $\boldsymbol{a}=(a_1\ a_2)^T$, $\boldsymbol{b}=(b_1\ b_2)^T$
\ENSURE  A reduced base $[\boldsymbol{a}\ \boldsymbol{b}]$
\STATE $\mathbf{M}=\mathbf{I}$
\STATE \textbf{While}(1)
\STATE \quad \textbf{if} ($a_1a_2b_1b_2\leq0 \wedge (|a_1|-|a_2|)(|b_1|-|b_2|)\leq 0$) \textbf{break}
\STATE \quad \textbf{if} $\|\boldsymbol{a}\|< \|\boldsymbol{b}\|$, Swap($\boldsymbol{a}, \boldsymbol{b}$)
\STATE \quad $n := \#\boldsymbol{a}$
\STATE \quad \textbf{if} $\#\boldsymbol{b}\leq \left\lfloor\frac{n}{2}\right\rfloor +1 $
\STATE \qquad \textbf{if} ($|a_1|\geq|a_2| \wedge b_1\ne 0$), $q :=\left\lfloor\frac{a_1}{b_1}\right\rfloor_o$
\STATE \qquad \textbf{else} $q :=\left\lfloor\frac{a_2}{b_2}\right\rfloor_o$
\STATE \qquad ($\boldsymbol{a},\boldsymbol{b}$) := ($\boldsymbol{b},\boldsymbol{a}-q\boldsymbol{b}$)
\STATE \qquad \textbf{continue}
\STATE \quad $(\boldsymbol{a},\boldsymbol{b},\mathbf{M})\leftarrow \textbf{HVec}(\boldsymbol{a}, \boldsymbol{b})$
\STATE \quad \textbf{if}  $\#(\boldsymbol{a}+\boldsymbol{b})\leq\lfloor\frac{n}{2}\rfloor+1$, ($\boldsymbol{a},\boldsymbol{b}$) := ($\boldsymbol{b},\boldsymbol{a}+\boldsymbol{b}$)
\STATE \quad \textbf{else if} $\#(\boldsymbol{a}-\boldsymbol{b})\leq\lfloor\frac{n}{2}\rfloor+1$, ($\boldsymbol{a},\boldsymbol{b}$) := ($\boldsymbol{b},\boldsymbol{a}-\boldsymbol{b}$)
\STATE \textbf{If} (($a_1b_1=0 \wedge sgn(a_2)\ne sgn(b_2))$ or ($a_2b_2=0 \wedge sgn(a_1)\ne sng(b_1)$)), $\boldsymbol{b} := -\boldsymbol{b}$
\STATE  \textbf{If} $\|\boldsymbol{a}\|> \|\boldsymbol{b}\|$, Swap($\boldsymbol{a}, \boldsymbol{b}$)
\STATE  \textbf{If} $\boldsymbol{a} = [0,0]$, \textbf{Return}[$\boldsymbol{a}, \boldsymbol{b}$]
\STATE  \textbf{Else}
\STATE \quad \textbf{if} ($a_2b_2\leq0$)
\STATE  \quad \quad $\boldsymbol{b} :=\mathrm{argmin}\left\{\left\|\boldsymbol{b}-\left\lfloor\frac{|b_1|-|b_2|}{|a_1|+|a_2|}\right\rfloor \boldsymbol{a}\right\|, \left\|\boldsymbol{b}-\left\lceil\frac{|b_1|-|b_2|}{|a_1|+|a_2|}\right\rceil \boldsymbol{a}\right\|\right\}$
\STATE \quad \textbf{else}
\STATE  \quad \quad $\boldsymbol{b} :=\mathrm{argmin}\left\{\left\|\boldsymbol{b}-\left\lfloor\frac{|b_2|-|b_1|}{|a_1|+|a_2|}\right\rfloor \boldsymbol{a}\right\|, \left\|\boldsymbol{b}-\left\lceil\frac{|b_2|-|b_1|}{|a_1|+|a_2|}\right\rceil \boldsymbol{a}\right\|\right\}$
\STATE  \quad \textbf{Return} $[\boldsymbol{a}\ \boldsymbol{b}]$
\end{algorithmic}
\end{algorithm}

\section{Comparative Analysis and Practical Performance Evaluation}
\subsection{Theoretical Comparative Analysis}
\Tian{For a lattice basis $\mathbf{B} = [\boldsymbol{a}\ \boldsymbol{b}] = \begin{pmatrix} a_1 & b_1 \\ a_2 & b_2 \end{pmatrix}$ with $\|\boldsymbol{a}\|\ge \|\boldsymbol{b}\|$, the fundamental distinction between our approach and previous work lies in the basic reduction operation $\boldsymbol{c} = \boldsymbol{a} - q\boldsymbol{b}$, specifically in the selection of the integer $q$. Table~\ref{tab:nocompare} provides a comprehensive comparison of existing two-dimensional lattice basis reduction algorithms. Although our fundamental reduction algorithm \textbf{CrossEuc} and its recursively optimized variant \textbf{HVecSBP} achieve the same theoretical complexity as the best existing algorithms, a closer inspection reveals a crucial difference: the computational cost of determining $q$ in our method eliminates several large-integer multiplications. Moreover, as shown in the table, our algorithm offers comprehensive advantages in terms of input basis requirements, termination condition (through the introduction of a new reduced basis), and the norm employed. Therefore, we conservatively suggest that the improvement may be a constant factor, which is absorbed in the big-$O$ notation. Nevertheless, in the following section, we present extensive comparative experiments demonstrating that our method achieves significant practical efficiency gains over prior algorithms.}

\begin{table}[!ht]
\caption{\Tian{Comprehensive comparison of two-dimensional lattice basis reduction algorithms}}
\centering
\footnotesize
\begin{tabular}{@{}p{1.5cm}p{2cm}p{2.6cm}p{2.1cm}p{6cm}p{1.cm}@{}}
 \toprule
\textbf{Prior Works} & 
\textbf{Input Basis} & $q$ & 
\textbf{Cost of $q$} & 
\textbf{Termination Condition} &  
\textbf{Norms} \\
\midrule
  \cite{Lag73} & Arbitrary & 
$\left\lfloor \frac{\langle \boldsymbol{a},\boldsymbol{b} \rangle}{\langle \boldsymbol{b},\boldsymbol{b} \rangle} \right\rceil$ & 
$1\ \text{Div} + 4\ \text{Mul}s$ & $\left|\frac{\langle \boldsymbol{a},\boldsymbol{b}\rangle}{\langle \boldsymbol{b},\boldsymbol{b}\rangle}\right|\le \frac{1}{2}\ \wedge\ \left|\frac{\langle \boldsymbol{a},\boldsymbol{b}\rangle}{\langle \boldsymbol{a},\boldsymbol{a}\rangle}\right|\le \frac{1}{2}$ & $\ell_2$ \\ 
\addlinespace
 \cite{MGbook} & Arbitrary & $\operatorname{argmin}_{z\in\mathbb{Z}}\|\boldsymbol{a}-z\boldsymbol{b}\|_p$ & Binary search $O(\# (\boldsymbol{a},\boldsymbol{b}))$ \text{Mul}s & $\|\boldsymbol{a}\|_p\le\|\boldsymbol{b}\|_p\le \|\boldsymbol{a}\pm\boldsymbol{b}\|_p$  & $\ell_p$ \\ \addlinespace
  \cite{yap1992fast}& Arbitrary &  $\left\lfloor \frac{\langle \boldsymbol{a},\boldsymbol{b} \rangle}{\langle \boldsymbol{b},\boldsymbol{b} \rangle} \right\rfloor$ &  $1\ \text{Div} + 4\ \text{Mul}s$ & $0<\frac{\langle \boldsymbol{a},\boldsymbol{b}\rangle}{\langle \boldsymbol{b},\boldsymbol{b}\rangle}\le 1\ \wedge\ 0<\frac{\langle \boldsymbol{a},\boldsymbol{b}\rangle}{\langle \boldsymbol{a},\boldsymbol{a}\rangle}\le 1$ & $\ell_2$ \\ 
\addlinespace
\cite{Sch91}& Arbitrary & $\left\lfloor \frac{\langle \boldsymbol{a},\boldsymbol{b} \rangle-s}{\langle \boldsymbol{b},\boldsymbol{b} \rangle} \right\rfloor\ (s>0)$ &  $1\ \text{Div} + 4\ \text{Mul}s$ & $0< \langle \boldsymbol{a},\boldsymbol{b}\rangle\le \langle \boldsymbol{b},\boldsymbol{b}\rangle-s\ \wedge\ 0<\langle \boldsymbol{a},\boldsymbol{b}\rangle \le \langle \boldsymbol{a},\boldsymbol{a}\rangle-s $ &  $\ell_2$ \\ 
\addlinespace
 \cite{Por20}& Arbitrary & $2^d\ \bigl(d = \#\langle \boldsymbol{a},\boldsymbol{b} \rangle - \#\langle \boldsymbol{b},\boldsymbol{b} \rangle\bigr)$ & $4\ \text{Mul}s$ or $\text{Shift}s$ & $\langle \boldsymbol{a},\boldsymbol{a}\rangle \le\langle \boldsymbol{b},\boldsymbol{b}\rangle\wedge \left|\frac{\langle \boldsymbol{a},\boldsymbol{b}\rangle}{\langle \boldsymbol{a},\boldsymbol{a}\rangle}\right|\le \frac{1}{2}$ &  $\ell_2$ \\ 
\addlinespace
\cite{EIS01,CTJ22} 
& HNF ($a_1>b_1\ge 0\wedge a_2=0\wedge b_1\ge|b_2|$) & 
$\left\lfloor \frac{a_1}{b_1} \right\rfloor$ & 
$1\ \text{Div}$ & $a_1<|a_2| $ 
&  $\ell_\infty$  \\ 
\addlinespace
Our Work & Arbitrary & 
$\left\lfloor \frac{a_1}{b_1} \right\rfloor_o$ or $\left\lfloor \frac{a_2}{b_2} \right\rfloor_o$ & 
$1\ \text{Div}$ & $a_1a_2b_1b_2\le 0\ \wedge\ (|a_1|-|a_2|)(|b_1|-|b_2|)\le 0$ & $\ell_p$ \\ 
\bottomrule
\end{tabular}
\label{tab:nocompare}
\end{table}
\Tian{In the above table, \text{Div}, \text{Mul} and \text{Shift} denote the operations of division, multiplication, and shift on (large) integers, respectively. The norm $\ell_p$ allows $p\ge 1$ and $p=\infty$.}

\subsection{Experimental Methodology}

To comprehensively evaluate the practical efficiency of our proposed algorithm \textbf{CrossEuc} and its optimized variant \textbf{HVecSBP} against existing methods, we will conduct extensive experiments considering both metric types and lattice basis forms. In terms of the metric, we compare the practical performance of various algorithms under both the Euclidean and Minkowski metrics. Regarding the input basis, it is evident that the performance of existing algorithms is highly dependent on the structure of the input lattice basis. We focus on two common types: the \emph{HNF-form} basis \(\mathbf{B} = [\boldsymbol{a}\ \boldsymbol{b}] = \begin{pmatrix} a & b \\ 0 & c \end{pmatrix}\), where \(a > b > |c|> 0\), and the \emph{general-form} basis \(\mathbf{B} = [\boldsymbol{a}\ \boldsymbol{b}] = \begin{pmatrix} a_1 & b_1 \\ a_2 & b_2 \end{pmatrix}\). For general-form bases, the degree of linear dependence between the two basis vectors clearly affects algorithmic performance. 
We define the linear dependency measure \(\Delta(\boldsymbol{a}, \boldsymbol{b})\) as:  
\[
\Delta(\boldsymbol{a}, \boldsymbol{b}) = \left| \frac{a_1}{b_1} - \frac{a_2}{b_2} \right|.
\]  
It is clear that a smaller \(\Delta\) value indicates a higher degree of linear dependency between vectors \(\boldsymbol{a}\) and \(\boldsymbol{b}\).

(1)  Evaluation for $\ell_2$-shortest basis. 

We compare seven algorithms: four fundamental reduction algorithms, namely our algorithm \textbf{CrossEuc}, the classical Lagrange algorithm \textbf{LagRed} under the \(\ell_2\) metric \cite{Lag73, MGbook} (see Appendix \ref{App:LagRed}), Yap's fundamental algorithm \textbf{CRS} \cite{yap1992fast} (see Appendix \ref{App:CRS}), and Pornin's improved Lagrange algorithm \textbf{LagRed+} \cite{Por20} (see Appendix \ref{App:LagRed+}); and three recursively optimized algorithms: our optimized algorithm \textbf{HVecSBP}, Yap's optimized algorithm \textbf{HalfGaussianSBP} \cite{yap1992fast} (see Appendix \ref{App:HalfGaussian}), and \Tian{\textbf{Flatter} \cite{RH23}, which, in dimension 2, is an optimized implementation of Sch\"{o}nhage's algorithm \cite{Sch91}.}
 
 \begin{itemize}
\item For {HNF-form} bases, we assessed the practical performance of the above-mentioned algorithms under two scenarios: (1) fixed bit length of \(a\) and \(b\) with varying bit lengths of \(c\); (2) fixed bit length of \(c\) with varying bit lengths of \(a\) and \(b\).

\item For general-form bases, we assessed the practical performance of the above-mentioned algorithms across different values of \(\Delta\).
\end{itemize}

(2) Evaluation for $\ell_\infty$-shortest basis.
\begin{itemize}
\item For {HNF-form} bases, our algorithm \textbf{CrossEuc} is identical to the continued fraction-based algorithm (\textbf{ParEuc}) \cite{EIS01,CTJ22}. Thus, we only compare three algorithms: the optimized Lagrange algorithm under the \(\ell_\infty\) metric (\textbf{GolEuc}) \cite{Lag73,CTJ22}, the continued fraction-based algorithm (\textbf{ParEuc}) \cite{EIS01,CTJ22}, and our optimized algorithm (\textbf{HVecSBP}).

\item For general-form bases, we evaluate the practical performance of the following algorithms across different \(\Delta\) values.
\begin{itemize}
    \item \textbf{GolEuc}: The optimized Lagrange algorithm under the \(\ell_\infty\) metric presented in \cite{CTJ22}.  
\item \textbf{EEA-HNF-ParEuc}: An algorithm that first transforms the basis into HNF using the extended Euclidean algorithm (EEA) and then reduces it with the continued fraction-based algorithm (\textbf{ParEuc}).  Visually
\begin{align*} 
\begin{pmatrix} a_1 & b_1\\  a_2 & b_2 \end{pmatrix}\overset{\textbf{EEA}}{\longrightarrow} \begin{pmatrix} a & b\\ 0 & c \end{pmatrix} \overset{\textbf{ParEuc}}{\longrightarrow}\mbox{reduced\ basis}
\end{align*}

\item \textbf{CrossEuc}: Our proposed reduced algorithm without using \textbf{Hvec} .  

\item \textbf{HGCD-HNF-ParEuc}: Similar to \textbf{EEA-HNF-ParEuc}, but the basis is transformed into HNF using the HGCD-based extended Euclidean algorithm.   Visually
\begin{align*} 
\begin{pmatrix} a_1 & b_1\\  a_2 & b_2 \end{pmatrix}\overset{\textbf{HGCD}}{\longrightarrow}  \begin{pmatrix} a & b\\ 0 & c \end{pmatrix} \overset{\textbf{ParEuc}}{\longrightarrow} \mbox{reduced\ basis}
\end{align*}

\item \textbf{HGCD-HNF-HVecSBP}: Similar to \textbf{HGCD-HNF-ParEuc}, but the HNF is reduced using our proposed \textbf{HVecSBP} algorithm.   Visually
\begin{align*} 
\begin{pmatrix} a_1 & b_1\\  a_2 & b_2 \end{pmatrix}\overset{\textbf{HGCD}}{\longrightarrow}\begin{pmatrix} a & b\\ 0 & c \end{pmatrix} \overset{\textbf{HVecSBP}}{\longrightarrow}  \mbox{reduced\ basis}
\end{align*}
\item \textbf{HVecSBP}: Our proposed algorithm applied directly to the input basis without prior transformation to HNF.  
\end{itemize}

\end{itemize}

\subsection{Experimental Environment and  Result Analysis}

Our experiments were conducted on a Ubuntu 22.04 machine equipped with an Intel Core i7-9750H CPU running at 2.6 GHz and 16 GB of RAM.  \Tian{All algorithms were implemented in C. To ensure a uniform evaluation standard, all operations involving large integers were performed using the GMP (GNU Multiple Precision Arithmetic) library. To facilitate reproducibility and further research, the implementation of our algorithms has been made publicly available on GitHub at \url{https://github.com/zlhhhcrypto/The-two-dimensional-lattice-basis-reduction}. For a fair benchmarking comparison, the \textbf{Flatter} algorithm is built upon the source code publicly released by Ryan (available at \url{https://github.com/keeganryan/flatter}) without any modifications.}

\subsubsection{Experimental result and analysis for the $\ell_2$-shortest basis}
\begin{itemize} 
\item For the input lattice basis in HNF, Table \ref{tab:HNFl2} compares the running time of each algorithm as $n_2$ varies, where $n_1=2\times 10^5$ denotes the number of decimal digits in $a$ and $b$, and $n_2$ represents the number of decimal digits in $c$.  As shown in the table, smaller values of \(n_2\) lead to higher computational times for all algorithms. \Tian{Based on the data, we compute the speedup ratios for selected algorithm pairs as $n_2$ varies; the results are presented in Table \ref{tab:speedup_all} and lead to the following observations:}
\begin{itemize}
\item[(i)] Among the four fundamental reduction algorithms—\textbf{CrossEuc}, \textbf{CRS}, \textbf{LagRed}, and \textbf{LagRed+}—and for unbalanced bases (small \(n_2\)), \textbf{CrossEuc} achieves speedups of over \(500\times\) compared to \textbf{CRS} and over \(350\times\) compared to \textbf{LagRed}. It is also \(10\)–\(15\times\) faster than \textbf{LagRed+}, with the advantage becoming more pronounced as the basis becomes more unbalanced.

\item[(ii)] For the three optimized algorithms—\textbf{HalfGaussianSBP}, \textbf{Flatter}, and \textbf{HVecSBP}—our algorithm \textbf{HVecSBP} outperforms \textbf{HalfGaussianSBP} by a factor of approximately \(14\times\) across most settings and is consistently faster than \textbf{Flatter} by about \(1.8\)–\(3.4\times\).

\item[(iii)] Comparing our fundamental algorithm \(\mathbf{CrossEuc}\) with its optimized variant \(\mathbf{HVecSBP}\), we observe that the optimized version is significantly faster for unbalanced bases (up to \(2.7\times\) for very small \(n_2\)). For perfectly balanced bases (\(n_2 = 200000\)), however, the fundamental algorithm is faster by a factor of \(2.7\)—a reflection of the overhead introduced by recursion for well‑conditioned inputs. 
\end{itemize}

	\begin{table}[ht!]
		\caption{\Tian{Running time of various algorithms for the $\ell_2$-shortest basis with varying $n_2$ (in seconds)}}
		\setlength{\tabcolsep}{3pt}
		\centering
		\small
 \begin{tabular}{cc|cccc|ccc}
\toprule

\multicolumn{2}{c|}{} & 
\multicolumn{4}{c|}{\textbf{Fundamental Reduction Algorithms}} & 
\multicolumn{3}{c}{\textbf{Recursive Optimized Algorithms}} \\

\cmidrule(lr){3-6} \cmidrule(lr){7-9}

\textbf{length (dec) $n_1$} & 
\textbf{length (dec) $n_2$} & 
\textbf{CrossEuc} & \textbf{CRS} & \textbf{LagRed} & \Tian{\textbf{LagRed+}} & 
\textbf{HalfGaussianSBP} & \Tian{\textbf{Flatter}} & \textbf{HVecSBP} \\ 

\midrule
200000  &200000   & 0.000522 & 0.035561   & 0.034960   & 0.006361  & 0.072230 & 0.027000  & 0.001405 \\
200000  &180000   & 0.215646 & 112.676702 & 77.524088  & 2.236443  & 0.605316 & 0.137000  & 0.040622 \\
200000  &160000   & 0.394690 & 213.510331 & 146.752125 & 4.449395  & 1.156759 & 0.205000  & 0.078501 \\
200000  &140000   & 0.557541 & 296.380581 & 204.626063 & 6.428470  & 1.689438 & 0.278000  & 0.109761 \\
200000  &120000   & 0.694540 & 366.024648 & 255.231900 & 8.309822  & 2.243587 & 0.328000  & 0.152746 \\
200000  &100000   & 0.803900 & 423.430425 & 296.717711 & 9.994242  & 2.769396 & 0.398000  & 0.197548 \\
200000  &80000    & 0.902036 & 464.113459 & 323.928026 & 11.711186  & 3.277750 & 0.482000  & 0.222285 \\
200000  &60000    & 0.952363 & 498.170153 & 346.013871 & 13.531682  & 3.902740 & 0.525000  & 0.262041 \\
200000  &40000    & 1.014240 & 517.602416 & 358.009027 & 14.273912  & 4.380847 & 0.594000  & 0.312382 \\
200000  &20000    & 1.037861 & 529.151738 & 365.802161 & 15.385821  & 4.944157 & 0.667000  & 0.350908 \\
200000  &10000    & 1.044446 & 530.145559 & 370.650970 & 16.015354  & 5.307488 & 0.686000  & 0.366273 \\
200000  &1000     & 1.046709 & 531.093756 & 371.334628 & 16.314675  & 5.534359 & 0.723000  & 0.394135 \\ 
200000  &100      & 1.052189 & 531.633906 & 371.504885 & 16.335464  & 5.561485 & 0.735000  & 0.394579 \\
200000  &1        & 1.055789 & 531.637407 & 371.524263 & 16.342724  & 5.567000 & 0.737000  & 0.402265 \\ 
\bottomrule
\end{tabular}
\label{tab:HNFl2}
\end{table}

\begin{table}[htp!]
\caption{\Tian{Speedup of competing algorithms over our algorithms as \(n_2\) decreases (with \(n_1=2\times10^5\) fixed).}}
\centering
\small
\setlength{\tabcolsep}{3pt}
\begin{tabular}{c|cccccc}
\toprule
\makecell{$n_2$} &
\makecell{\textbf{CRS} /\\ \textbf{CrossEuc}} &
\makecell{\textbf{LagRed} /\\ \textbf{CrossEuc}} &
\makecell{\textbf{LagRed+} / \\ \textbf{CrossEuc}} &
\makecell{\textbf{HalfGaussianSBP} / \\ \textbf{HVecSBP}} &
\makecell{\textbf{Flatter} / \\ \textbf{HVecSBP}} &
\makecell{\textbf{CrossEuc} / \\ \textbf{HVecSBP}} \\
\midrule
200000 & 68.12 & 66.96 & 12.18 & 51.42 & 19.22 & 0.3715 \\
180000 & 522.6 & 359.5 & 10.37 & 14.90 & 3.372 & 5.308 \\
160000 & 540.9 & 371.9 & 11.27 & 14.73 & 2.611 & 5.028 \\
140000 & 531.7 & 367.1 & 11.53 & 15.39 & 2.533 & 5.079 \\
120000 & 527.0 & 367.5 & 11.97 & 14.69 & 2.147 & 4.547 \\
100000 & 526.7 & 369.1 & 12.43 & 14.02 & 2.015 & 4.067 \\
80000  & 514.5 & 359.1 & 12.98 & 14.74 & 2.168 & 4.059 \\
60000  & 523.1 & 363.4 & 14.21 & 14.89 & 2.004 & 3.634 \\
40000  & 510.4 & 353.0 & 14.07 & 14.02 & 1.901 & 3.247 \\
20000  & 510.0 & 352.5 & 14.83 & 14.09 & 1.901 & 2.956 \\
10000  & 507.7 & 354.9 & 15.33 & 14.48 & 1.872 & 2.851 \\
1000   & 507.4 & 354.8 & 15.58 & 14.04 & 1.834 & 2.654 \\
100    & 505.4 & 353.0 & 15.53 & 14.09 & 1.863 & 2.666 \\
1      & 503.6 & 351.8 & 15.48 & 13.84 & 1.832 & 2.624 \\
\bottomrule
\end{tabular}
\label{tab:speedup_all}
\end{table}

\item \Tian{For the input lattice basis in HNF, Table \ref{tab:HNFl2add} compares the running time of each algorithm as $n_1$ varies, where $n_1$ denotes the number of decimal digits in $a$ and $b$, and $n_2=1$ represents the number of decimal digits in $c$.  As shown in the table, larger values of \(n_1\) lead to higher computational times for all algorithms. Based on the data, we compute the speedup ratios as the runtime of the competitor divided by the runtime of our algorithm (\textbf{CrossEuc} or \textbf{HVecSBP}); the results are presented in Table \ref{tab:speedup_n1} and lead to the following observations:
 \begin{itemize}
\item[(i)] Among the four fundamental reduction algorithms—\textbf{CrossEuc}, \textbf{CRS}, \textbf{LagRed}, and \textbf{LagRed+}—and for large \(n_1\) (e.g., \(\ge 10^4\)), \textbf{CrossEuc} consistently achieves speedups of hundreds to thousands over \textbf{CRS} and \textbf{LagRed}, and roughly \(10\)–\(15\times\) over \textbf{LagRed+}. 
\item[(ii)] For the three optimized algorithms—\textbf{HalfGaussianSBP}, \textbf{Flatter}, and \textbf{HVecSBP}—our algorithm \textbf{HVecSBP} is typically \(12\)–\(15\times\) faster than \textbf{HalfGaussianSBP} and \(1.5\)–\(6\times\) faster than \textbf{Flatter} for most medium‑to‑large inputs.  
\item[(iii)] Comparing our fundamental algorithm \(\mathbf{CrossEuc}\) with its optimized variant \(\mathbf{HVecSBP}\), The ratio \textbf{CrossEuc / HVecSBP} shows that for very large \(n_1\) the fundamental algorithm is faster (reflecting recursion overhead), while for moderately sized \(n_1\) the optimized version takes the lead. For very small \(n_1\) (e.g., \(1000\)), the overhead dominates, but absolute times are negligible.  
\end{itemize}}
 
		\begin{table}[!htbp]
			\caption{\Tian{Running time of various algorithms for the $\ell_2$-shortest basis with varying $n_1$ (in seconds)}}
			\setlength{\tabcolsep}{3pt}
			\centering
			\small
	\begin{tabular}{cc|cccc|ccc}
\toprule

\multicolumn{2}{c|}{} & 
\multicolumn{4}{c|}{\textbf{Fundamental Reduction Algorithms}} & 
\multicolumn{3}{c}{\textbf{Recursive Optimized Algorithms}} \\

\cmidrule(lr){3-6} \cmidrule(lr){7-9}
				\textbf{length (dec) $n_1$} & \ \textbf{length (dec) $n_2$} & \ \textbf{CrossEuc} & \ \textbf{CRS} &\ \textbf{LagRed} & \ \textbf{LagRed+} & \ \textbf{HalfGaussianSBP} & \ \textbf{Flatter} & \ \textbf{HVecSBP} \\ \midrule
200000 &1   & 1.055789 & 531.637407 & 371.524263 & 16.342724 & 5.567000 & 0.737000  & 0.402265 \\
180000 &1   & 0.951559 & 464.087204 & 324.317322 & 13.578352 & 5.541714 & 0.678000  & 0.374289 \\
160000 &1   & 0.847952 & 356.992835 & 244.174815 & 10.643685 & 4.740843 & 0.616000  & 0.314115 \\
140000 &1   & 0.557353 & 254.783202 & 178.288337 & 8.034135 & 4.113942 & 0.546000  & 0.308759 \\
120000 &1   & 0.433006 & 193.123427 & 134.624508 & 6.087771 & 3.674104 & 0.413000  & 0.267532 \\
100000 &1   & 0.298789 & 126.346206 & 88.425548 & 4.310278 & 3.126494 & 0.379000  & 0.238366 \\
80000  &1   & 0.181331 & 69.108312 & 47.991715 & 2.527737 & 2.364476 & 0.325000  & 0.192734 \\
60000  &1   & 0.123001 & 36.467855 & 25.930139 & 1.471530 & 2.038620 & 0.222000  & 0.129469 \\
40000  &1   & 0.047514 & 13.649372 & 9.355868 & 0.577352 & 1.152887 & 0.182000  & 0.078961 \\
20000  &1   & 0.012565 & 2.869850 & 2.059330 & 0.148317 & 0.606376 & 0.123000  & 0.040865 \\
10000  &1   & 0.003528 & 0.483666 & 0.332132 & 0.032946 & 0.281085 & 0.071000  & 0.019048 \\
1000   &1   & 0.000139 & 0.001484 & 0.001305 & 0.000681 & 0.026947 & 0.012000  & 0.001953\\
100    &1   & 0.000031 & 0.000035 & 0.000056 & 0.000046 & 0.003570 & 0.010000  & 0.000263 \\ \bottomrule
			\end{tabular}\label{tab:HNFl2add}
		\end{table}
\begin{table}[!htbp]
\caption{\Tian{Speedup of competing algorithms over our algorithms (varying \(n_1\), \(n_2 = 1\))}}
\centering
\small
\setlength{\tabcolsep}{3pt}
\begin{tabular}{c|cccccc}
\toprule
\makecell{$n_1$} &
\makecell{\textbf{CRS} / \\ \textbf{CrossEuc}} &
\makecell{\textbf{LagRed} / \\ \textbf{CrossEuc}} &
\makecell{\textbf{LagRed+} / \\ \textbf{CrossEuc}} &
\makecell{\textbf{HalfGaussianSBP} / \\ \textbf{HVecSBP}} &
\makecell{\textbf{Flatter} / \\ \textbf{HVecSBP}} &
\makecell{\textbf{CrossEuc} / \\ \textbf{HVecSBP}} \\
\midrule
200000 & 503.6 & 351.8 & 15.48 & 13.84 & 1.832 & 2.624 \\
180000 & 487.8 & 340.8 & 14.27 & 14.81 & 1.812 & 2.542 \\
160000 & 421.0 & 288.0 & 12.55 & 15.09 & 1.961 & 2.699 \\
140000 & 457.2 & 320.0 & 14.42 & 13.32 & 1.768 & 1.805 \\
120000 & 445.9 & 310.9 & 14.06 & 13.73 & 1.544 & 1.619 \\
100000 & 422.9 & 296.0 & 14.43 & 13.12 & 1.590 & 1.254 \\
80000  & 381.1 & 264.6 & 13.94 & 12.27 & 1.686 & 0.941 \\
60000  & 296.5 & 210.8 & 11.96 & 15.75 & 1.715 & 0.950 \\
40000  & 287.2 & 196.9 & 12.15 & 14.60 & 2.305 & 0.602 \\
20000  & 228.4 & 163.9 & 11.80 & 14.84 & 3.010 & 0.308 \\
10000  & 137.1 & 94.2 & 9.34 & 14.76 & 3.728 & 0.185 \\
1000 & 10.68 & 9.39 & 4.90 & 13.80 & 6.144 & 0.071 \\
100    & 1.13 & 1.81 & 1.48 & 13.57 & 38.0 & 0.118 \\
\bottomrule
\end{tabular}
\label{tab:speedup_n1}
\end{table}

\item  For the input lattice basis in general form, Table \ref{tab:noHNFl2} compares the time cost of each algorithm for a fixed size  $n=2\times 10^5$, as the difference \(\Delta = |a_1/b_1 - a_2/b_2| \approx 10^{-\delta}\) varies, where $n$ denotes the maximum number of decimal digits in $a_1, a_2, b_1$ and $b_2$. The last column, \(\kappa\), represents the proportion of differing terms in the continued fraction expansions of \(a_1/b_1\) and \(a_2/b_2\), counted from right to left, relative to the total number of terms in the expansions.
Specially, the first $\kappa=1$ indicates that only the last term of the continued fraction expansions of \(a_1/b_1\) and \(a_2/b_2\) differs.
Notably, as \(\Delta\) decreases, $\delta$ increases and \(\kappa\) decreases, indicating that \(a_1/b_1\) and \(a_2/b_2\) become closer to being equal. This further implies that the lattice basis vectors \(\boldsymbol{a} = (a_1\ a_2)^T\) and \(\boldsymbol{b} = (b_1\ b_2)^T\) approach linear dependence.  As shown in Table \ref{tab:noHNFl2}, the closer the two input lattice basis vectors are to being linearly dependent ($\delta$ increases), the higher the computational time required for all algorithms. Additionally, \Tian{based on the data, we compute the speedup ratios as the runtime of the competitor divided by the runtime of our algorithm (\textbf{CrossEuc} or \textbf{HVecSBP}); the results are presented in Table  \ref{tab:speedup_general} and lead to the following observations:}
 
\begin{itemize}
    \item[(i)]  \textbf{CrossEuc} vs. other fundamental reduction algorithms: For all tested \(\delta\) values, \textbf{CrossEuc} is approximately \(500\times\) faster than \textbf{CRS} and \(330\)–\(380\times\) faster than \textbf{LagRed}. The speedup over \textbf{LagRed+} is around \(11\times\) on average.

   \item[(ii)]  \textbf{HVecSBP} vs. optimized algorithms: \textbf{HVecSBP} consistently outperforms \textbf{HalfGaussianSBP} by a factor of about \(14\)–\(17\times\) and is \(1.5\)–\(2.4\times\) faster than \textbf{Flatter}. The advantage over \textbf{Flatter} tends to increase as the basis becomes more skewed (smaller \(\delta\)).  
   
   \item[(iii)]  \textbf{Optimization gain}: The ratio \textbf{CrossEuc}/\textbf{HVecSBP} shows that the recursive variant is \(2.5\)–\(5.1\times\) faster than its fundamental counterpart. This confirms the effectiveness of the recursive strategy for challenging inputs.
 
\end{itemize}
\end{itemize}

		\begin{table}[!h]
		\caption{\Tian{Running time of various algorithms for the $\ell_2$-shortest basis with varying $\delta$ (in seconds)}}	
		\setlength{\tabcolsep}{1pt}
		\centering
		\small
	\begin{tabular}{c|cccc|ccc|c}
\toprule

\multicolumn{1}{c|}{} & 
\multicolumn{4}{c|}{\textbf{Fundamental Reduction Algorithms}} & 
\multicolumn{3}{c|}{\textbf{Recursive Optimized Algorithms}} & 
\multicolumn{1}{c}{} \\
\cmidrule(lr){2-5} \cmidrule(lr){6-8}
			\textbf{$\delta$}   & \ \textbf{CrossEuc} &\  \textbf{CRS} & \ \textbf{LagRed}& \ \Tian{\textbf{LagRed+}}& \ \textbf{HalfGaussianSBP} & \ \Tian{\textbf{Flatter}} & \ \textbf{HVecSBP} & \ $\kappa$\\ \midrule
400000 & 2.187857 & 1061.482896 & 732.066229 & 23.616684 & 14.621288 & 1.269000  & 0.865833 & 1 \\
350000 & 2.039634 & 1029.768712 & 675.253707 & 23.280258 & 9.911299  & 1.175000  & 0.705745 & 12.5\% \\
300000 & 1.963530 & 953.469648  & 640.119371 & 22.310885 & 8.779507  & 1.046000  & 0.597472 &  25\% \\
250000 & 1.760551 & 890.452219  & 595.714993 & 20.852116 & 7.692156  & 0.927000  & 0.530357 & 37.5\% \\
200000 & 1.541461 & 820.747442  & 564.586658 & 18.883664 & 6.102788  & 0.766000  & 0.430760 & 50\%  \\
150000 & 1.361687 & 705.880787  & 461.356443 & 15.976803 & 4.899072  & 0.636000  & 0.344573 & 62.5\%  \\
100000 & 1.059967 & 573.571627  & 395.397606 & 12.897524 & 3.676309  & 0.497000  & 0.260802 & 75\%  \\
50000  & 0.753551 & 416.365192  & 283.668019 & 9.043053 & 2.415546  & 0.341000  & 0.168508 & 87.5\%  \\
20000  & 0.571839 & 306.532068  & 206.223716 & 6.482485 & 1.718590  & 0.262000  & 0.131812 & 95\% \\
10000  & 0.488057 & 263.113941  & 178.243286 & 5.707449 & 1.572602  & 0.226000  & 0.103363 & 97.5\% \\
2000   & 0.458032 & 227.192431  & 159.118865 & 4.913506 & 1.255080  & 0.218000  & 0.090067 & 99.5\% \\
1000   & 0.446958 & 224.975630  & 158.739315 & 4.807506 & 1.243584  & 0.209000  & 0.087849 & 99.75\% \\
200    & 0.439207 & 223.499143  & 153.864744 & 4.785735 & 1.234583  & 0.206000  & 0.085077 & 99.95\% \\
100    & 0.432651 & 222.870675  & 153.261555 & 4.695844 & 1.233634  & 0.202000  & 0.084719 & 99.975\% \\
20     & 0.423252 & 222.065486  & 152.929364 & 4.687234 & 1.232835  & 0.198000  & 0.084328 & 99.995\% \\
0      & 0.422037 & 221.981643  & 152.823700 & 4.666861 & 1.231388  & 0.195000  & 0.084183 & 100\% \\
	\bottomrule
\end{tabular}
\label{tab:noHNFl2}
\end{table}

\begin{table}[!h]
\caption{\Tian{Speedup of competing algorithms over our algorithms (general‑form bases, varying \(\delta\))}}
\centering
\small
\setlength{\tabcolsep}{3pt}
\begin{tabular}{c|cccccc}
\toprule
\makecell{\(\delta\)} &
\makecell{\textbf{CRS} / \\ \textbf{CrossEuc}} &
\makecell{\textbf{LagRed} / \\ \textbf{CrossEuc}} &
\makecell{\textbf{LagRed+} / \\ \textbf{CrossEuc}} &
\makecell{\textbf{HalfGaussianSBP} / \\ \textbf{HVecSBP}} &
\makecell{\textbf{Flatter} / \\ \textbf{HVecSBP}} &
\makecell{\textbf{CrossEuc} / \\ \textbf{HVecSBP}} \\
\midrule
400000 & 485.2 & 334.6 & 10.80 & 16.89 & 1.466 & 2.527 \\
350000 & 504.9 & 331.1 & 11.42 & 14.05 & 1.665 & 2.891 \\
300000 & 485.6 & 326.0 & 11.36 & 14.69 & 1.751 & 3.287 \\
250000 & 505.8 & 338.4 & 11.84 & 14.50 & 1.748 & 3.320 \\
200000 & 532.5 & 366.3 & 12.25 & 14.17 & 1.778 & 3.579 \\
150000 & 518.5 & 338.9 & 11.73 & 14.22 & 1.846 & 3.952 \\
100000 & 541.1 & 373.0 & 12.17 & 14.10 & 1.906 & 4.065 \\
50000 & 552.7 & 376.5 & 12.00 & 14.34 & 2.024 & 4.472 \\
20000 & 536.0 & 360.6 & 11.33 & 13.04 & 1.988 & 4.339 \\
10000 & 539.1 & 365.2 & 11.69 & 15.22 & 2.186 & 4.722 \\
2000 & 496.0 & 347.4 & 10.73 & 13.94 & 2.420 & 5.085 \\
1000 & 503.4 & 355.2 & 10.76 & 14.16 & 2.379 & 5.089 \\
200 & 508.9 & 350.4 & 10.90 & 14.51 & 2.421 & 5.162 \\
100 & 515.1 & 354.2 & 10.85 & 14.56 & 2.384 & 5.108 \\
20 & 524.7 & 361.4 & 11.08 & 14.62 & 2.348 & 5.020 \\
0 & 526.0 & 362.1 & 11.06 & 14.63 & 2.317 & 5.014 \\
\bottomrule
\end{tabular}
\label{tab:speedup_general}
\end{table}

\subsubsection{Experimental result and analysis for the $\ell_\infty$-shortest basis}
\begin{itemize} 
\item For the input lattice basis in HNF,  TABLE \ref{tab:HNF} compares the time cost of each algorithms as the variance of $n_2$, where $n_1=10^6$ denotes the number of decimal digits in $a,b$, and $n_2$ represents the number of decimal digits in $c$. These results highlight \textbf{HVecSBP}'s superior performance in handling HNF-structured bases, particularly in scenarios with highly unbalanced parameter lengths (\(n_2 \ll n_1\)). Such scenarios are common in cryptographic applications, such as lattice-based attacks, where \textbf{HVecSBP} achieves at least $13.5\times$ efficiency improvement over previous approaches.

\begin{table}[ht!]
\caption{Running time of various algorithms for the $\ell_\infty$-shortest basis with varying $n_2$ (in seconds)}
\setlength{\tabcolsep}{3pt}
\centering
\small
\begin{tabular}{ccccccc}
\hline
\textbf{length (dec) $n_1$} & \ \textbf{length (dec) $n_2$} & \ \textbf{HVecSBP}  & \ \textbf{CrossEuc}/\textbf{ParEuc}  & \ \textbf{GolEuc} \\ \hline
1000000		&1000000   & 0.006217    & 0.002445     & 0.002459      \\
1000000		&800000    & 0.426119    & 10.279363    & 49.977750     \\
1000000		&600000    & 0.832923    & 18.195970    & 90.955213     \\ 
1000000		&400000    & 1.162091    & 26.268203    & 121.361453    \\ 
1000000		&200000    & 1.607177    & 27.132104    & 141.404999    \\
1000000		&100000    & 1.947272    & 28.337137    & 146.895788    \\ 
1000000		&80000     & 1.919498    & 28.082500    & 147.245142    \\ 
1000000		&60000     & 1.951394    & 28.259035    & 148.577950    \\
1000000		&40000     & 1.976183    & 28.483309    & 148.651564    \\
1000000		&20000     & 2.039936    & 28.457997    & 148.667026    \\ 
1000000		&10000     & 2.061487    & 29.026893    & 150.630187    \\ 
1000000		&5000      & 2.114340    & 29.114340    & 150.116029    \\ 
1000000		&1000      & 2.153284    & 29.424312    & 150.084888    \\ 
1000000		&100       & 2.164948    & 29.353622    & 150.895833    \\
1000000		&1         & 2.168530    & 29.583845    & 150.976715     \\ \hline
\end{tabular}\label{tab:HNF}
\end{table}

\item  For the input lattice basis in general form, Table \ref{tab:noHNF}  compares the time cost of each algorithm for a fixed size  $n=10^6$, as the difference \(\Delta = |a_1/b_1 - a_2/b_2| \approx 10^{-\delta}\) varies, where $n$ denotes the maximum number of decimal digits in $a_1, a_2, b_1$ and $b_2$. 
The meaning of the parameter $\kappa$ is same as that in Table \ref{tab:noHNFl2}.  
Notably, as $\Delta$ decreases, $\delta$ increases and $\kappa$ decreases, suggesting that $\frac{a_1}{b_1}$ and $\frac{a_2}{b_2}$ become increasingly similar. This, in turn, implies that the lattice basis vectors $\boldsymbol{a} = (a_1\ a_2)^T$ and $\boldsymbol{b} = (b_1\ b_2)^T$ are approaching linear dependence. As shown in Table \ref{tab:noHNF},  the following observations can be made:  

\begin{itemize}
    \item[(i)]   For the three fundamental algorithms without HGCD optimization—\textbf{CrossEuc}, \textbf{EEA-HNF-ParEuc}, and \textbf{GolEuc}—the closer the two input lattice basis vectors are to being linearly dependent, the more computational time is required by all three. Meanwhile, our proposed algorithm, \textbf{CrossEuc}, consistently outperforms the other two. Compared to \textbf{GolEuc}, \textbf{CrossEuc} demonstrates a relatively stable efficiency advantage that is largely unaffected by the linear dependency of the input lattice basis vectors. Specifically, the time cost of \textbf{CrossEuc} is approximately 20\% of that of \textbf{GolEuc}, representing a 5× speedup. In contrast, when compared to \textbf{EEA-HNF-ParEuc}, the efficiency gain of \textbf{CrossEuc} is highly sensitive to the linear dependency of the input vectors. As the degree of dependency decreases (i.e., as $\delta$ decreases), \textbf{CrossEuc} achieves a speedup ranging from $4\times$ to $18\times$, with the advantage becoming increasingly significant.

   \item[(ii)] For the HGCD-optimized algorithms—\textbf{HGCD-HNF-ParEuc}, \textbf{HGCD-HNF-HVecSBP}, and our algorithm \textbf{HVecSBP}, we have (1) When the continued fraction expansions of \(a_1/b_1\) and \(a_2/b_2\) differ only in the last term, \textbf{HGCD-HNF-ParEuc} achieves the highest efficiency, followed by \textbf{HGCD-HNF-HVecSBP}. The time cost of \textbf{HVecSBP} is approximately $2\times$ that of the other two algorithms, making it the least efficient. This is intuitive, as the input lattice basis vectors nearly degenerate into two integers, making integer-based HGCD processing more efficient. The vectorized \textbf{HVec} algorithm, which handles two dimensions, naturally incurs double the time cost. 
(2)   As linear dependence between lattice basis vectors diminishes ($\delta$ decreases)—specifically when continued fraction expansions of \(a_1/b_1\) and \(a_2/b_2\) diverge in approximately the last 27.5\% of terms—our algorithm \(\mathbf{HVecSBP}\) achieves dominant performance. Beyond this threshold, its efficiency advantage scales progressively: accelerating from marginal parity  $1\times$ to 13-fold speedup over \(\mathbf{HGCD\text{-}HNF\text{-}HVecSBP}\), and from $2\times$ to $179\times$ acceleration against \(\mathbf{HGCD\text{-}HNF\text{-}ParEuc}\), with gains intensifying at smaller $\delta$ values.

\item[(iii)] 
 For our proposed fundamental reduction algorithm, $\mathbf{CrossEuc}$, and its optimized variant, $\mathbf{HVecSBP}$, the runtime of $\mathbf{HVecSBP}$ is approximately 4.2\% to 8.5\% of that of $\mathbf{CrossEuc}$ as the linear dependence between the input lattice basis vectors decreases—resulting in an efficiency gain of roughly $11\times$ to $23\times$.

\end{itemize}
\end{itemize}

\begin{table}[ht!]
\captionsetup{justification=centering}
\caption{Running time of various algorithms for the $\ell_\infty$-shortest basis with varying $\delta$ varies (in seconds)}

		\setlength{\tabcolsep}{1pt}
		\centering
		\small
		\begin{tabular}{cccccccc}
			\hline
			\textbf{$\delta$}   & \ \textbf{HVecSBP} &\  \textbf{HGCD-HNF-HVecSBP} & \ \textbf{HGCD-HNF-ParEuc} & \ \textbf{CrossEuc} & \ \textbf{EEA-HNF-ParEuc} & \ \textbf{GolEuc} & \ $\kappa$\\ \hline
			2000000  & 4.634715   & 2.474747   & 2.335431    &  54.612207  & 250.625570  & 269.501145  & 1  \\
			1750000	      & 3.759720   & 2.881473   & 3.509663    &  53.949515  & 231.295727  & 262.973879  & 12.5\% \\
			1500000        & 3.281029   & 3.130119   & 7.128568    &  53.507613  & 229.640975  & 249.979270 &  25\% \\
            1450000         & 3.177414   & 3.165364  & 8.725851   &  51.548659  & 210.162490  & 249.096180  & 27.5\% \\
			1400000         & 3.092249   & 3.232044  & 9.014794   &  50.038115  & 208.529073  & 248.094480  & 30\% \\
			1250000         & 2.765591   & 3.551487   &  14.045340   &  45.819027  & 204.894112  & 231.718764  & 37.5\% \\
			1000000          & 2.412870   & 4.063529   &  21.208691   &  40.766721  & 200.351373  & 206.916571 & 50\%  \\
			750000           & 1.855896   & 4.274508   &  31.767607   &  35.590796  & 198.790794  & 173.002045 & 62.5\%  \\
			500000             & 1.417675   & 4.729311   & 44.924310    &  29.244469  & 197.095717  & 139.971889 & 75\%  \\
			250000            & 0.917597   & 5.205109   &  59.637123   &  19.667034  & 196.227950  & 97.178470 & 87.5\%  \\
			100000            & 0.685130   & 5.302298   &  70.196164   & 14.644298  & 196.028048  & 71.587244 & 95\% \\
			50000            & 0.568257   & 5.336118   &  74.210315   &  13.054622  & 196.026205  & 62.898395 & 97.5\% \\
			10000            & 0.467744   & 5.331523   & 75.565236   & 11.143312  & 195.498931  & 54.231681 & 99.5\% \\
			5000            & 0.459721   & 5.425184   & 76.523376   & 10.757789  & 195.399899  & 53.463294 & 99.75\% \\
			1000            & 0.448851   & 5.505665   &  76.635827   &  10.684234  & 195.360483  & 52.133793 & 99.95\% \\
			500            & 0.447047   & 5.662536   & 76.843067   & 10.612306  & 195.329952  & 51.892477 & 99.975\% \\
			100            & 0.439681   & 5.689787   &  77.879209   &  10.525084  & 194.966942  & 51.416951 & 99.995\% \\
			0           & 0.437364   & 5.701985   &  78.296140   &  10.396514  & 194.718627  & 51.037448 & 100\% \\
			\hline
		\end{tabular}\label{tab:noHNF}
	\end{table}

\section{Conclusion and discussion}
This paper introduces a newly defined reduced basis for two-dimensional lattices and develops a foundamental reduction algorithm, \textbf{CrossEuc}, which effectively solves the SVP and SBP in two-dimensional lattices. Furthermore, we extend the integer version of the HGCD algorithm and propose a vector-based \textbf{HVec} algorithm, providing detailed implementation and analysis. This leads to an optimized version, \textbf{HVecSBP}, which further accelerates the performance of \textbf{CrossEuc}. Finally, we conduct extensive experiments to evaluate the practical performance of our newly designed algorithms.  \Tian{Given the role of classical Lagrange reduction and its improvement in the design and analysis of high-dimensional LLL algorithms, we hope that our work will provide new perspectives and foundations for the design of lattice reduction algorithms in higher dimensions.}

\bibliographystyle{abbrv}
\bibliography{myreferences}

\appendices
\section{The optimized Lagrange algorithm under the \(\ell_\infty\) metric \textbf{GolEuc}.}\label{App:GolEuc}
\begin{algorithm}\scriptsize
	\caption{\textbf{GloEuc}\cite{CTJ22}}   \label{alg:Gol-Euc}
	\begin{algorithmic}[1]
		\REQUIRE A basis $[\boldsymbol{a} \ \boldsymbol{b}]$.
		\ENSURE  An $\ell_\infty$-shortest basis for $\mathcal{L}([\boldsymbol{a}\ \boldsymbol{b}])$
		\STATE  \textbf{If} $\|\boldsymbol{a}\|>\|\boldsymbol{b}\|,$ \textbf{then} swap ($\boldsymbol{a}, \boldsymbol{b}$)
		\STATE  \textbf{If} $\|\boldsymbol{a}-\boldsymbol{b}\|>\|\boldsymbol{a}+\boldsymbol{b}\|,$ \textbf{then} $\boldsymbol{b}: = -\boldsymbol{b}$
		\STATE \textbf{If} $\|\boldsymbol{b}\|\leq\|\boldsymbol{a}-\boldsymbol{b}\|,$ \textbf{then} return $[\boldsymbol{a} \ \boldsymbol{b}]$
		\STATE \textbf{If} $\|\boldsymbol{a}\|\leq\|\boldsymbol{a}-\boldsymbol{b}\|,$ \textbf{then} goto \emph{loop}
		\STATE \textbf{If} $\|\boldsymbol{a}\|=\|\boldsymbol{b}\|,$ \textbf{then} return $[\boldsymbol{a} \ \boldsymbol{a}-\boldsymbol{b}]$
		\STATE  $[\boldsymbol{a} \ \boldsymbol{b}]:=[\boldsymbol{b}-\boldsymbol{a} \ \boldsymbol{a}]$
		\STATE  \emph{loop}
		\STATE  \quad \textbf{If} $\frac{b_1}{a_1}, \frac{b_2}{a_2} \neq 0$ and $sgn\left(\frac{b_1}{a_1}\right)=sgn\left(\frac{b_2}{a_2}\right)$ \textbf{then}
		\STATE  \quad \quad Find $\mu \in \left\{ \left\lceil \frac{|b_1| + |b_2|}{|a_1| + |a_2|} \right\rceil, \left\lfloor \frac{|b_1| + |b_2|}{|a_1| +|a_2|} \right\rfloor \right\},$ s.t. $\|\boldsymbol{b} - \mu\boldsymbol{a}\|$ is minimal
		\STATE  \quad \textbf{Else}
		\STATE  \quad \quad Find $\mu \in \left\{ \left\lceil \frac{\big||b_2| - |b_1|\big|}{|a_2| + |a_1|} \right\rceil, \left\lfloor \frac{\big||b_2| - |b_1|\big|}{|a_2| + |a_1|} \right\rfloor\right\},$ s.t. $\|\boldsymbol{b} - \mu\boldsymbol{a}\|$ is minimal
		\STATE \quad $ \boldsymbol{b}: = \boldsymbol{b} - \mu\boldsymbol{a}$
		\STATE \quad  \textbf{If} $\|\boldsymbol{a} - \boldsymbol{b}\| > \|\boldsymbol{a} + \boldsymbol{b}\|$ \textbf{then} $\boldsymbol{b}: = -\boldsymbol{b}$
		\STATE \quad Swap($\boldsymbol{a},\boldsymbol{b}$)
		\STATE \quad \textbf{If} $[\boldsymbol{a} \ \boldsymbol{b}]$ is Lagrange-reduced, \textbf{then} return $[\boldsymbol{a} \ \boldsymbol{b}]$
		\STATE \quad goto \emph{loop}
	\end{algorithmic}
\end{algorithm}

\section{Lagrange's reduction algorithm under the \(\ell_2\) metric \textbf{LagRed}} \label{App:LagRed}

\begin{algorithm}\scriptsize
	\caption{\textbf{LagRed} \cite{MGbook}} \label{alg:Lagl2norm}
	\def\temptablewidth{0.5\textwidth}
	\begin{algorithmic}[1]
		\REQUIRE A basis $[\boldsymbol{a} \ \boldsymbol{b}]$.
		\ENSURE  An $\ell_2$-shortest basis for $\mathcal{L}([\boldsymbol{a}\ \boldsymbol{b}])$
		\STATE  \textbf{If} $\|\boldsymbol{a}\|_{2}>\|\boldsymbol{b}\|_{2},$ \textbf{then} swap ($\boldsymbol{a}, \boldsymbol{b}$)
		\STATE  \textbf{If} $\|\boldsymbol{a}-\boldsymbol{b}\|_{2}>\|\boldsymbol{a}+\boldsymbol{b}\|_{2},$ \textbf{then} $\boldsymbol{b}: = -\boldsymbol{b}$
		\STATE \textbf{If} $\|\boldsymbol{b}\|_{2}\leq\|\boldsymbol{a}-\boldsymbol{b}\|_{2},$ \textbf{then} return $[\boldsymbol{a} \ \boldsymbol{b}]$
		\STATE \textbf{If} $\|\boldsymbol{a}\|_{2}\leq\|\boldsymbol{a}-\boldsymbol{b}\|_{2},$ \textbf{then} goto \emph{loop}
		\STATE \textbf{If} $\|\boldsymbol{a}\|_{2}=\|\boldsymbol{b}\|_{2},$ \textbf{then} return $[\boldsymbol{a} \ \boldsymbol{a}-\boldsymbol{b}]$
		\STATE  $[\boldsymbol{a} \ \boldsymbol{b}]:=[\boldsymbol{b}-\boldsymbol{a} \ \boldsymbol{a}]$
		\STATE  \emph{loop}
		\STATE  \quad $\mu := \left \lceil \frac{\langle\boldsymbol{a},\boldsymbol{b}\rangle}{\|\boldsymbol{a}\|_{2}^{2}} \right \rfloor $
		\STATE \quad $ \boldsymbol{b}: = \boldsymbol{b} - \mu\boldsymbol{a}$
		\STATE \quad  \textbf{If} $\|\boldsymbol{a} - \boldsymbol{b}\|_{2} > \|\boldsymbol{a} + \boldsymbol{b}\|_{2}$ \textbf{then} $\boldsymbol{b}: = -\boldsymbol{b}$
		\STATE \quad Swap($\boldsymbol{a},\boldsymbol{b}$)
		\STATE \quad \textbf{If} $[\boldsymbol{a} \ \boldsymbol{b}]$ is Lagrange-reduced, \textbf{then} return $[\boldsymbol{a} \ \boldsymbol{b}]$
		\STATE \quad goto \emph{loop}
	\end{algorithmic}
\end{algorithm}

\section{Yap's coherent remainder sequence Reduction Algorithm \textbf{CRS}}\label{App:CRS}
\begin{algorithm}\scriptsize
	\caption{\textbf{CRS} \cite{yap1992fast}} \label{alg:CRS-red}
	\def\temptablewidth{0.5\textwidth}
	\begin{algorithmic}[1]
		\REQUIRE  An admissible basis $[\boldsymbol{a} \ \boldsymbol{b}]$ refers to a pair of vectors $\boldsymbol{a}$ and $\boldsymbol{b}$ that form an acute angle and satisfy $\|\boldsymbol{a}\|_2 \geq \|\boldsymbol{b}\|_2$.
		\ENSURE  An $\ell_2$-shortest basis for $\mathcal{L}([\boldsymbol{a}\ \boldsymbol{b}])$
		\STATE \textbf{While}(true)
		\STATE  \quad $\mu := \left\lfloor\frac{<\boldsymbol{a},\boldsymbol{b}>}{||\boldsymbol{b}||_{2}^{2}}\right\rfloor$
		\STATE \quad \textbf{If} $||\boldsymbol{a}-\mu\boldsymbol{b}||_{2}= ||\boldsymbol{a}||_{2}$, \textbf{break}
		\STATE \quad $[\boldsymbol{a} \ \boldsymbol{b}] := [\boldsymbol{b} \ \boldsymbol{a} - \mu\boldsymbol{b}]$
		\STATE \textbf{If} $\|\boldsymbol{a}-\boldsymbol{b}\|_2<\|\boldsymbol{b}\|_2$, \textbf{then} $[\boldsymbol{a} \ \boldsymbol{b}]:=[\boldsymbol{b} \ \boldsymbol{a}-\boldsymbol{b}]$
		\STATE $[\boldsymbol{a} \ \boldsymbol{b}]:=\left[\boldsymbol{b} \ \boldsymbol{a}-\left \lceil \frac{\langle\boldsymbol{a},\boldsymbol{b}\rangle}{\|\boldsymbol{b}\|_{2}^{2}} \right \rfloor \boldsymbol{b} \right]$
		\STATE \textbf{Return} $[\boldsymbol{a} \ \boldsymbol{b}]$ 
	\end{algorithmic}
\end{algorithm}

\section{Pornin's improved Lagrange reduction algorithm \textbf{LagRed+}}\label{App:LagRed+}

	\begin{algorithm}\scriptsize
		\caption{\textbf{LagRed+} \cite{}} \label{alg:Lag+l2norm}
		\def\temptablewidth{0.5\textwidth}
		\begin{algorithmic}[1]
			\REQUIRE A basis for a two-dimensional lattice $\mathcal{L}$ is $[\boldsymbol{a}, \boldsymbol{b}] \in \mathbb{Z}^2 \times \mathbb{Z}^2$.
			\ENSURE  A short non-zero vector of $\mathcal{L}$.
			\STATE $N_{a}:=\|\boldsymbol{a}\|_{2}^{2}$
			\STATE $N_{b}:=\|\boldsymbol{b}\|_{2}^{2}$
			\STATE $p:=\langle\boldsymbol{a},\boldsymbol{b}\rangle$
			\STATE \emph{loop}
			\STATE \quad \textbf{If} $N_{a}<N_{b}$,\textbf{then}
			\STATE \qquad Swap($\boldsymbol{a},\boldsymbol{b}$)
			\STATE \qquad Swap($N_{a},N_{b}$)
			\STATE \quad \textbf{If} $2|p|\leq N_{b}$, \textbf{then}
			\STATE \qquad return $[\boldsymbol{b} \ \boldsymbol{a}]$
			\STATE \quad $s:=\max ( 0, len(p)-len(N_{b}) )$ // \textbf{If} $x\geq0,len(x)=\# x$, \textbf{else} $len(x)=\# (-x-1)$
			\STATE \quad \textbf{If} $p>0$,\textbf{then}
			\STATE \qquad $\boldsymbol{a}:=\boldsymbol{a}-2^{s}\boldsymbol{b}$
			\STATE \qquad $N_{a}:=N_{a}+2^{2s}N_{b}-2^{s+1}p$
			\STATE \qquad $p:=p-2^{s}N_{b}$
			\STATE \quad \textbf{else}
			\STATE \qquad $\boldsymbol{a}:=\boldsymbol{a}+2^{s}\boldsymbol{b}$
			\STATE \qquad $N_{a}:=N_{a}+2^{2s}N_{b}+2^{s+1}p$
			\STATE \qquad $p:=p+2^{s}N_{b}$
			\STATE \emph{end loop}
		\end{algorithmic}
	\end{algorithm}

\newpage
\section{Yap's Half-Gaussian Algorithm \textbf{Half-Gaussian}}\label{App:HalfGaussian}
\begin{algorithm}\scriptsize
	\caption{\textbf{Half-Gaussian} \cite{yap1992fast}} \label{alg:Half-Gaussian}
	\def\temptablewidth{0.5\textwidth}
	\begin{algorithmic}[1]
		\REQUIRE An admissible basis $[\boldsymbol{a} \ \boldsymbol{b}]$ refers to a pair of vectors $\boldsymbol{a}$ and $\boldsymbol{b}$ that form an acute angle and satisfy $\|\boldsymbol{a}\|_2 \geq \|\boldsymbol{b}\|_2$.
		\ENSURE Regular matrix $\mathbf{M^{*}}$ that reduce $[\boldsymbol{a}\ \boldsymbol{b}]$ to $[\boldsymbol{a}^{*}\ \boldsymbol{b}^{*}]$ such that $[\boldsymbol{a}^{*}\ \boldsymbol{b}^{*}]$ either satisfy $(\log\|\mathbf{a^{*}}\|_2\geq (\log\|\boldsymbol{a}\|_2/2)+c > \log\|\mathbf{b^{*}}\|_2$ (also called straddles) or is terminal with $\log \|\mathbf{b^{*}}\|_2\geq \log\|\boldsymbol{a}\|_2+c$, $c$ is a very large positive integer.
		\STATE \textbf{If} $\cos (\boldsymbol{a},\boldsymbol{b})\geq 1/2$, \textbf{return} $\mathbf{M^{*}} := \mathbf{E}$ 
		\STATE $n := \left \lceil \log ||\boldsymbol{a}||_{2} \right \rceil , m := \left \lceil n/2 \right \rceil $
		\STATE \textbf{If} $\log ||\boldsymbol{b}||_{2}<n-(m/2)+c$ or $m\leq 8$, \textbf{then} $\mathbf{M}_1^{*}:= \mathbf{E}$ and \textbf{goto} step $20$
		\STATE $(\mathbf{a_0},\mathbf{b_0}) := (\left\lceil\boldsymbol{a}/2^{m}\right\rceil, \left\lfloor\boldsymbol{b}/2^{m}\right\rfloor)$
		\STATE \textbf{If} $||\mathbf{a_0}||_{2}=||\mathbf{b_0}||_{2}$, \textbf{then} $\mathbf{M}_1^{*}:=\begin{pmatrix}0&1\\ 1&0\end{pmatrix}$ and \textbf{goto} step $20$
		\STATE $\mathbf{M_1} := \textbf{Half-Gaussian}(\mathbf{a_0},\mathbf{b_0})$
		\STATE $(\mathbf{a^{\prime}},\mathbf{b^{\prime}}) := (\boldsymbol{a},\boldsymbol{b})\mathbf{M_{1}^{-1}}$
		\STATE $(\mathbf{a_0^{\prime}},\mathbf{b_0^{\prime}}) := (\mathbf{a_0},\mathbf{b_0})\mathbf{M_{1}^{-1}}$
		\STATE \textbf{If} $\log||\boldsymbol{a}_0^{\prime}||_{2}\geq \log ||\boldsymbol{a}_0||_{2}-(m/2)+c>\log ||\boldsymbol{b}_0^{\prime}||_{2}$
		\STATE \quad \textbf{If} $\log ||\boldsymbol{a}^{\prime}||_{2}\leq \log||\boldsymbol{b}^{\prime}||_{2}$, \textbf{then} $\mathbf{M}_1^{*}:=\begin{pmatrix}0&1\\ 1&-q_k\end{pmatrix}\mathbf{M_1}$
		\STATE \quad \textbf{If} $\cos (\mathbf{a^{\prime}},\mathbf{b^{\prime}})<0$
		\STATE \quad \quad $\mathbf{a^{\prime\prime}} := \mathbf{a^{\prime}}+\mathbf{b^{\prime}}$
		\STATE \quad \quad \textbf{If} $q_k>1$
		\STATE \quad \quad \quad \textbf{If} $\log||\mathbf{a^{\prime}}||_{2}>\log||\mathbf{a^{\prime\prime}}||_{2}$, \textbf{let} $q_k := q_k-1$ and update $\mathbf{M_1}$ to $\mathbf{M}_1^{*}$
		\STATE \quad \quad \quad \textbf{Else} \textbf{let} $q_k := q_k-1$ and update $\mathbf{M_1}$ to $\mathbf{M}_1^{*}$, \textbf{then} make one toggle: $\mathbf{M}_1^{*}:=\begin{pmatrix}1&1\\ 1&0\end{pmatrix}\mathbf{M}_1^{*}$
		\STATE \quad \quad \textbf{Else}
		\STATE \quad \quad \quad \textbf{If} $\log||\mathbf{a^{\prime\prime}}||_{2}>\log ||\mathbf{a^{\prime}}||_{2}$, backup one step and update $\mathbf{M_1}$ to $\mathbf{M}_1^{*}$
		\STATE \quad \quad \quad \textbf{Else} backup two steps and update $\mathbf{M_1}$ to $\mathbf{M}_1^{*}$
		\STATE \textbf{Else} $\mathbf{M}_1^{*}:= \mathbf{M_1}$
		\STATE $(\mathbf{a_1^{*}},\boldsymbol{b}_1^{*}) := (\boldsymbol{a},\boldsymbol{b})\mathbf{M_{1}^{*}}^{-1}$
		\STATE \textbf{If} $(\mathbf{a_1^{*}},\boldsymbol{b}_1^{*})$ is terminal$(\left\lfloor \frac{<\mathbf{a_1^{*}},\boldsymbol{b}_1^{*}>}{||\boldsymbol{b}_1^{*}||_{2}^{2}} \right\rfloor = 0)$ or straddles $(n/2)+c$, \textbf{return} $\mathbf{M}_1^{*}$
		\STATE \textbf{Else} do one step CRS-red and compute the elementary matrix $\mathbf{M}$ to transform $(\boldsymbol{a}_1^{*},\boldsymbol{b}_1^{*})$ to $ (\boldsymbol{b}_1^{*},\boldsymbol{c}_1^{*})$, \textbf{let} $(\boldsymbol{x},\boldsymbol{y}) := (\boldsymbol{b}_1^{*},\boldsymbol{c}_1^{*})$
		\STATE \textbf{If} $(\boldsymbol{x},\boldsymbol{y})$ straddles $(n/2)+c$, \textbf{return} $\mathbf{M}\mathbf{M}_1^{*}$
		\STATE $n^{\prime} := \left \lceil \log ||\boldsymbol{x}||_{2} \right \rceil , m^{\prime} := 2(n^{\prime}-m) $
		\STATE \textbf{If} $\log||\boldsymbol{y}||_{2}<n^{\prime}-(m^{\prime}/2)+c$ or $m^{\prime}\leq 8$, \textbf{then} $\mathbf{M_2^{*}}:= \mathbf{E}$ and \textbf{goto} step $42$
		\STATE $(\boldsymbol{x_0},\boldsymbol{y_0}) := (\left\lceil\boldsymbol{x}/2^{m^{\prime}}\right\rceil, \left\lfloor\boldsymbol{y}/2^{m^{\prime}}\right\rfloor)$
		\STATE \textbf{If} $||\boldsymbol{x_0}||_{2}=||\boldsymbol{y_0}||_{2}$, \textbf{then} $\mathbf{M_2^{*}}:=\begin{pmatrix}0&1\\ 1&0\end{pmatrix}$ and \textbf{goto} step $42$
		\STATE $\mathbf{M_2} := \textbf{Half-Gaussian}(\boldsymbol{x_0},\boldsymbol{y_0})$
		\STATE $(\boldsymbol{x^{\prime}},\boldsymbol{y^{\prime}}) := (\boldsymbol{x},\boldsymbol{y})\mathbf{M_{2}^{-1}}$
		\STATE $(\boldsymbol{x_0^{\prime}},\boldsymbol{y_0^{\prime}}) := (\boldsymbol{x_0},\boldsymbol{y_0})\mathbf{M_{2}^{-1}}$
		\STATE \textbf{If} $\log||\boldsymbol{x}_0^{\prime}||_{2}\geq \log||\boldsymbol{x_0}||_{2}-(m^{\prime}/2)+c>\log ||\boldsymbol{y}_0^{\prime}||_{2}$
		\STATE \quad \textbf{If} $\log||\boldsymbol{x}^{\prime}||_{2}\leq \log||\boldsymbol{y}^{\prime}||_{2}$, \textbf{then} $\mathbf{M_2^{*}}:=\begin{pmatrix}0&1\\ 1&-q_k\end{pmatrix}\mathbf{M}_2$
		\STATE \quad \textbf{If} $\cos (\boldsymbol{x^{\prime}},\boldsymbol{y^{\prime}})<0$
		\STATE \quad \quad $\boldsymbol{x^{\prime\prime}} := \boldsymbol{x^{\prime}}+\boldsymbol{y^{\prime}}$
		\STATE \quad \quad \textbf{If} $q_k>1$
		\STATE \quad \quad \quad \textbf{If} $\log||\boldsymbol{x^{\prime}}||_{2}>\log ||\boldsymbol{x^{\prime\prime}}||_{2}$, \textbf{let} $q_k := q_k-1$ and update $\mathbf{M_2}$ to $\mathbf{M_2^{*}}$
		\STATE \quad \quad \quad \textbf{Else} \textbf{let} $q_k := q_k-1$ and update $\mathbf{M_2}$ to $\mathbf{M_2^{*}}$, \textbf{then} make one toggle: $\mathbf{M_2^{*}}:=\begin{pmatrix}1&1\\ 1&0\end{pmatrix}\mathbf{M_2^{*}}$
		\STATE \quad \quad \textbf{Else}
		\STATE \quad \quad \quad \textbf{If} $\log||\boldsymbol{x^{\prime\prime}}||_{2}>\log ||\boldsymbol{x^{\prime}}||_{2}$, backup one step and update $\mathbf{M_2}$ to $\mathbf{M_2^{*}}$
		\STATE \quad \quad \quad \textbf{Else} backup two steps and update $\mathbf{M_2}$ to $\mathbf{M_2^{*}}$
		\STATE \textbf{Else} $\mathbf{M_2^{*}}:= \mathbf{M_2}$
		\STATE \textbf{return} $\mathbf{M^{*}}:=\mathbf{M_2^{*}}\mathbf{M}\mathbf{M}_1^{*}$
	\end{algorithmic}
\end{algorithm}
\begin{algorithm}
	\caption{\textbf{HalfGaussianSBP}($\boldsymbol{a}, \boldsymbol{b}$)}   \label{alg:Half-Gaussian-SBP}
	\begin{algorithmic}[1]
		\REQUIRE  A base $[\boldsymbol{a}\ \boldsymbol{b}]$ with $\boldsymbol{a}=(a_1\ a_2)^T$, $\boldsymbol{b}=(b_1\ b_2)^T$
		\ENSURE  An $\ell_2$-shortest basis $[\boldsymbol{a}\ \boldsymbol{b}]$
		\STATE $\mathbf{M}:=\mathbf{I}$
		\STATE \textbf{While} (true)
		\STATE \quad \textbf{If} $||\boldsymbol{a}||_{2}<||\boldsymbol{b}||_{2}$, Swap($\boldsymbol{a}, \boldsymbol{b}$)
		\STATE \quad \textbf{If} $\cos (\boldsymbol{a},\boldsymbol{b})<0$, $\boldsymbol{b}:=-\boldsymbol{b}$
		\STATE \quad $\mathbf{M} := \textbf{Half-Gaussian}(\boldsymbol{a},\boldsymbol{b})$
		\STATE \quad $(\boldsymbol{a},\boldsymbol{b}) := (\boldsymbol{a},\boldsymbol{b})\mathbf{M^{-1}}$
		\STATE \quad $\mu:=\left\lfloor\frac{<\boldsymbol{a},\boldsymbol{b}>}{||\boldsymbol{b}||_{2}^{2}}\right\rfloor$		
		\STATE \quad \textbf{If} $||\boldsymbol{a}-\mu\boldsymbol{b}||_{2}= ||\boldsymbol{a}||_{2}$, \textbf{break}
		\STATE \quad $[\boldsymbol{a} \ \boldsymbol{b}] := [\boldsymbol{b} \ \boldsymbol{a} - \mu\boldsymbol{b}]$
		\STATE \textbf{If} $||\boldsymbol{a}-\boldsymbol{b}||_2<||\boldsymbol{b}||_2$, \textbf{then} $[\boldsymbol{a} \ \boldsymbol{b}]:=[\boldsymbol{b} \ \boldsymbol{a}-\boldsymbol{b}]$
		\STATE $[\boldsymbol{a} \ \boldsymbol{b}]:=[\boldsymbol{b} \ \boldsymbol{a}-\left \lceil \frac{<\boldsymbol{a},\boldsymbol{b}>}{||\boldsymbol{b}||_{2}^{2}} \right \rfloor \boldsymbol{b} ]$
		\STATE \textbf{Return} $[\boldsymbol{a}\ \boldsymbol{b}]$
	\end{algorithmic}
\end{algorithm}


\newpage

\begin{IEEEbiographynophoto}
{Lihao Zhao} received the B.E. degree in information security from Qingdao University in 2023. He is currently pursuing the M.S. degree in the School of Computer Science and Technology of Qingdao University. His research interests include lattice algorithms and privacy computing.
\end{IEEEbiographynophoto}

\vspace{-2.5cm}
\begin{IEEEbiographynophoto}{Chengliang Tian}
  received the B.S. and M.S. degrees in mathematics from Northwest University, Xi'an, China, in 2006 and 2009, respectively, and the Ph.D. degree in information security from Shandong University, Ji'nan, China, in 2013. He held a post-doctoral position with the State Key Laboratory of Information Security, Institute of Information Engineering, Chinese Academy of Sciences, Beijing. He is currently with the College of Computer Science and Technology, Qingdao University, as an Associate  Professor. His research interests include lattice-based cryptography and cloud computing security.
\end{IEEEbiographynophoto}

\vspace{-2.5cm}
\begin{IEEEbiographynophoto}{Jingguo Bi}  received the B.Sc. and Ph.D. degrees in information security from Shandong University, Jinan, China, in 2007 and 2012, respectively.He is currently an Associate Researcher with the School of Cyberspace Security, Beijing University of Posts and Telecommunications, Beijing, China. His research interests are public key cryptography, cloud computing, and post-quantum cryptography.
\end{IEEEbiographynophoto}
\vspace{-1.5cm}
\begin{IEEEbiographynophoto}{Guangwu Xu}(Senior Member)
received his Ph.D. in mathematics from SUNY Buffalo.
He is now with the School of Cyber Science and Technology, Shandong University, China.
His research interests include  cryptography, arithmetic number theory, compressed sensing,
 algorithms, and functional analysis.
\end{IEEEbiographynophoto}

\vspace{-1.5cm}
\begin{IEEEbiographynophoto}{Jia Yu}
(Member, IEEE) received the B.S. and M.S. degrees from the School of Computer Science and Technology, Shandong University, Jinan, China, in 2003 and 2000, respectively, and the Ph.D. degree from the Institute of Network Security, Shandong University in 2006. He is a Professor with the College of Computer Science and Technology, Qingdao University, Qingdao, China. He was a Visiting Professor with the Department of Computer Science and Engineering, State University of New York at Buffalo, Buffalo, NY, USA, from November 2013 to November 2014. His research interests include cloud computing security, key evolving cryptography, digital signature, and network security
\end{IEEEbiographynophoto}

\end{document}